\documentclass{article}
\usepackage{arxiv}
\usepackage{fancyhdr}
\fancypagestyle{topright}{
  \fancyhf{}                 %
  \fancyhead[R]{\thepage}    %
}
\usepackage{ulem}
\usepackage{amsthm}
\usepackage{booktabs}
\usepackage{multirow}
\newcommand{\sign}{\text{sign}}
\usepackage{array,ragged2e}
\newcolumntype{L}[1]{>{\RaggedRight\arraybackslash}p{#1}}

\newtheorem{remark}{Remark}
\newtheorem{lemma}{Lemma}

\usepackage{bm}
\usepackage{xcolor}
\definecolor{darkgreen}{RGB}{0, 150, 0}
\usepackage{mathtools}

\usepackage{accents}
\newcommand{\ubar}[1]{\underaccent{\bar}{#1}}
\usepackage{placeins}%
\def\mydefb#1{\expandafter\def\csname b#1\endcsname{\bm{#1}}}
\def\mydefallb#1{\ifx#1\mydefallb\else\mydefb#1\expandafter\mydefallb\fi}
\mydefallb ABCDEFGHIJKLMNOPQRSTUVWXYZabcdeghijklnopqrstuvwxyz\mydefallb

\def\mydefgreek#1{\expandafter\def\csname b#1\endcsname{\text{\boldmath$\mathbf{\csname #1\endcsname}$}}}
\def\mydefallgreek#1{\ifx\mydefallgreek#1\else\mydefgreek{#1}%
	\lowercase{\mydefgreek{#1}}\expandafter\mydefallgreek\fi}
\mydefallgreek {beta}{Gamma}{alpha}{Delta}{epsilon}{Omega}{Theta}{Iota}{Lambda}{kappa}{mu}{nu}{Xi}{Pi}{chi}{rho}{Sigma}{Psi}{tau}\mydefallgreek
\newcommand{\Tau}{\mathrm{T}}
\newcommand{\bff}{\bm f}

\usepackage{hyperref}   
\hypersetup{%
bookmarksnumbered,
pdfstartview={FitH},
linkcolor=blue,
citecolor=blue,
colorlinks=true,
pdfpagemode={UseNone},
plainpages=false
}%

\newcommand{\C}{\text{\normalfont C}}
\newcommand{\D}{\text{\normalfont D}}

\newcommand{\Tend}{n}
\newcommand{\Iend}{n}

\newcommand\blfootnote[1]{%
  \begingroup
  \renewcommand\thefootnote{}\footnote{#1}%
  \addtocounter{footnote}{-1}%
  \endgroup
}
\usepackage{natbib}
\usepackage{amsfonts, amsmath }

\title{A computational method for type I error rate control in power-maximizing response-adaptive randomization}

\author{Stef Baas\thanks{Corresponding author} \\
	MRC Biostatistics Unit,
	University~of~Cambridge, Cambridge, United Kingdom\\\phantom{.}\\
	{\bf Lukas Pin} \\
	{MRC Biostatistics Unit}, {University of Cambridge}, {Cambridge}, {United Kingdom},\\\phantom{.}\\
	{\bf Sof\'ia S. Villar}\\
	{MRC Biostatistics Unit}, {University of Cambridge}, {Cambridge}, {United Kingdom}\\\phantom{.}\\
	{\bf William F. Rosenberger}\\
	{Department of Statistics}, 
{George Mason University}, {Arlington}, {Virginia}, {United States of America}
}

\date{}

\renewcommand{\shorttitle}{A computational method for type I error rate control in power-maximizing RAR}

\hypersetup{
pdftitle={A computational method for type I error rate control in power-maximizing response-adaptive randomization},
pdfauthor={Stef Baas, Lukas Pin, Sof\'ia S. Villar, William F. Rosenberger}
}

\begin{document}

\maketitle
\begin{abstract}
	Maximizing statistical power in experimental design often involves imbalanced treatment allocation, but several challenges hinder its practical adoption: (1) the misconception that equal allocation always maximizes power,
(2) when only targeting maximum power, more than half the participants
may be expected to obtain inferior treatment, and (3) response-adaptive
randomization (RAR) targeting maximum statistical power may inflate type
I error rates substantially. Recent work identified issue (3) and proposed a
novel allocation procedure combined with the asymptotic score test. Instead,
the current research focuses on finite-sample guarantees. First, we analyze
the power for traditional power-maximizing RAR procedures under exact
tests, including a novel generalization of Boschloo's test. Second, we evaluate
constrained Markov decision process (CMDP) RAR procedures under exact
tests. These procedures target maximum average power under constraints
on pointwise and average type I error rates, with averages taken across the
parametric space. A combination of the unconditional exact test and the
CMDP procedure protecting allocations to the superior arm gives the best
performance, providing substantial power gains over equal allocation while allocating more participants in expectation to the superior treatment. Future
research could focus on the randomization test, in which CMDP procedures
exhibited lower power compared to other examined RAR procedures.
	\blfootnote{\scriptsize\newline {\bf Abbreviations}: RAR: response-adaptive randomization, CMDP: constrained Markov decision process, CMDP-P: CMDP maximizing power, CMDP-BP: CMDP maximizing power constraining patient benefit, DBCD: doubly adaptive biased coin design.}
	\hspace{-1.5mm}
\end{abstract}

\keywords{Boschloo's test \and Conditional exact test \and Constrained Markov decision processes \and  Neyman allocation \\\and Sequential experiments\and Unconditional exact test
 }

\section{Introduction}

A two-armed randomized controlled trial evaluates a potential causal effect between administering an intervention~(instead of a control treatment) to trial participants and an outcome variable of interest.
Randomized controlled trials traditionally keep the target proportion of participants randomized to each treatment fixed during the trial. 
  In this setting, the inferential properties of standard statistical tests are often well understood, while such designs can take less effort to implement in practice in comparison to alternatives.

A common misconception in the literature is that two-armed randomized controlled trials with a fixed equal proportion of participants allocated to each treatment~(i.e., equal allocation) lead to maximal statistical power~\citep{pin202511powerfulunequalallocation}. This statement does not hold in general;~\citet{azrielmandelrinott2011} showed that equal allocation optimizes power for binary outcomes,  in one~(the Pitman) asymptotic testing framework, whereas unbalanced allocation is optimal in another~(the Bahadur) asymptotic testing framework. 
\citet{Neyman1934} showed that the treatment group sizes minimizing the variance of the difference in means -- hence maximizing power of the Wald test in the normal outcomes model and by a normal approximation in many other statistical models -- are proportional to the standard deviations of the respective outcomes. 
\citet{Pin2024} showed that when changing the measure of interest, e.g., from the mean difference to the log odds ratio, the optimal proportion is almost never an equal proportion, while its functional form differs substantially from the optimal proportion for the mean difference.

The allocation proportions maximizing power often depend on unknown model parameters. 
To this end, several procedures have been proposed in the literature to learn and target the optimal allocation proportion.
These procedures can be classified as {\it response-adaptive randomization}~(RAR) procedures, which sequentially adjust the allocation probabilities of participants to treatments based on the trial history of observed participant outcomes and
treatment allocations. 
Power-maximizing RAR procedures are RAR procedures that maximize a metric related to statistical power.
Frequentist examples of power-maximizing RAR procedures are sequential estimation procedures targeting the proportion of~\citet{Neyman1934} based on the doubly adaptive biased coin design~\citep[DBCD,][]{eisele1994dabcd}, the efficient randomized design~\citep{hu2009ERADE}, or the mass-weighted urn design introduced in~\citet{ZHAO2015209}.
\citet{giovagnoli2017bopt} provided an overview and extension of Bayesian optimal designs for binary responses and showed that the allocated proportion under the $A$-optimal Bayesian RAR procedure~(using the allocation probability that minimizes a utility function describing the variance of the treatment effect estimator) converges to the proportion in~\citet{Neyman1934}, and is hence a power-maximizing RAR procedure.

An often overlooked aspect in the literature of power-maximizing RAR procedures is the finite-sample type I error rate. \citet{baldi_2018_wald} indicated for general RAR procedures that, although asymptotically the type I error rate is controlled by the classical Wald test, its inferential properties are compromised for smaller sample sizes. \citet{pin2025revisitingoptimalallocationsbinary} showed that for a specific RAR procedure targeting the proportions in~\citet{Neyman1934}, the type I error rate can become as high as 82.2\% for a trial size of~50 participants. 
Such a degree of type I error rate inflation is detrimental from a regulatory perspective, as type I error rate control is a minimal and first-line requirement for any clinical trial.
To control the type I error rate inflation, \citet{pin2025revisitingoptimalallocationsbinary} derived two novel optimal allocation proportions based on more robust score tests instead of the (unpooled) Wald test. They showed through simulations that these novel target proportions led to a better type I error rate control. 

In the current paper, %
we evaluate approaches with finite-sample guarantees for two-arm clinical trials with binary outcomes. In particular, we investigate two broad classes of alternative solutions to control the type I error rate for power-maximizing RAR procedures.
First, we investigate the use of alternative statistical tests. As part of this strategy, we consider conditional and unconditional exact tests for~(power-maximizing) RAR designs.
Exact tests are statistical tests ensuring a~(finite-sample) type I error rate bounded by a nominal significance level. Conditional exact tests provide this guarantee for several different reference sets~(e.g., all sets of trial realizations with a given total number of successes), while unconditional tests only provide this guarantee when the reference set just contains all trial realizations.
We construct exact tests through the network algorithm introduced for randomization tests in~\citet{mehta1988constructing} and first applied to construct exact tests~(based on the population model) for RAR procedures in~\citet{wei1990statistical}. We use the implementation of the network algorithm in~\citet{baas2025exact}, shown to be computationally tractable for trials with up to~$1,000$ participants, where previously the computational limit was around 100 participants. 
As part of this strategy, we develop and introduce a novel generalization of Boschloo's exact test~\citep{boschloo1970raised} for data collected using an RAR procedure. This exact test increases the unconditional power of the conditional exact test, which was often shown to already yield highest power out of all comparators for RAR designs in~\citet{baas2025exact}~(where only conditional and unconditional exact tests were considered). 

Our second solution approach is to consider the {\it constrained Markov decision process}~(CMDP) approach introduced in~\citet{baas2024CMDP} to construct novel (average) power-maximizing RAR procedures that control the type I error rate both pointwise and on average. Such RAR procedures are novel, as~\citet{baas2024CMDP} only considered CMDP procedures maximizing expected treatment outcomes under constraints on the average power and average type I error rate. In this work, we combine this approach with the exact testing strategy to balance power and type I error trade-offs.
It is hypothesized that the use of exact tests for power-maximizing CMDP designs leads to a lower degree of power deflation than when they are applied to traditional power-maximizing procedures. This is because the CMDP designs ensure that the type I error rate will already be close to the nominal level under the asymptotic Wald test, hence the performance will be similar under exact tests.   
While~\citet{baas2025exactstatisticaltestsusing} show that tests optimizing an unweighted average power measure can provide almost uniform power gains over standard tests in equal allocation designs, we note that, similar to the approach in~\citet{baas2025exactstatisticaltestsusing}, informative priors~(or weighted averages) can be used to tune the performance of the test to the specific trial parameters deemed realistic for a given trial application.

The paper is structured as follows: \autoref{subsect:model} introduces the considered clinical trial model,  exact statistical tests~(including our generalized version of Boschloo's test in~\autoref{sect:generalized_boschloo}), and the operating characteristics. \autoref{sect:RA_procedures} lists the~(power-maximizing) RAR procedures considered in the paper.
\autoref{sect:numerical_results} provides a numerical comparison of the considered RAR procedures and exact tests, as well as results from the randomization test. \autoref{sect:conclusion_discussion} concludes the paper and provides topics of future research.

\section{Exact tests for a two-arm response-adaptive design with binary outcomes}\label{subsect:model}
\subsection{Model:  Two-arm response-adaptive design with binary outcomes}
In this section, we describe the mathematical model underlying our numerical results.
Consider a clinical trial comparing two treatments with unknown outcome distributions: a control~(C) treatment and a developmental~(D, i.e., experimental) treatment. Participants are enrolled sequentially and assigned to one of the treatments using a {\it response-adaptive randomization}~(RAR) procedure. This procedure updates the allocation probabilities based on the current trial history. After each participant is assigned a treatment, a binary outcome --sampled from a Bernoulli distribution-- is observed and incorporated into the trial history before assigning the next participant.

We now formalize this trial description. Let~$\btheta=(\theta_\C,\theta_\D)\in(0,1)^2$ 
be a tuple containing (unknown) success probabilities. Throughout this paper, the same ordering convention (i.e., first~C then~D) will be used when forming tuples from variables for the control and developmental treatment.
For a fixed trial size of~$n\in\mathbb{N}$ participants, define~$\bY_{a} = (Y_{a,1}, Y_{a,1}, \dots, Y_{a,n})$
to be a sequence of independent Bernoulli random variables for each treatment arm~$a\in\{\C,\D\}$, hence, independently,~$\mathbb{P}_\btheta(Y_{a,t}=1)=\theta_a$ for~$a\in\{\C,\D\}$. Each~$Y_{a,t}$ denotes a potential outcome for trial participant~$t$ under treatment~$a$.
The set of possible trial histories of treatment allocations and outcomes is defined as~$\mathcal{H}=\bigcup_{t=0}^n\mathcal{H}_t$ where~$\mathcal{H}_0=\{()\}$ only contains the empty tuple and~$$
     \mathcal{H}_t=\{(a_1,y_{1},a_2, y_{2},\dots, a_t,y_{t}): y_{w}\in\{0,1\},\,a_w\in\{\C,\D\},\;\forall w\leq t\} \label{defn:histories}
$$ denotes the set of potential trial histories up to participant~$t.$ 

An RAR procedure~$\pi:\mathcal{H}\mapsto [0,1]$ is a mapping, taking as input the available trial history and producing as output the probability that the next participant is allocated to treatment~C. It is assumed that the trial is fully sequential, i.e., trial~participants are allocated to treatment one after another, and the corresponding treatment outcome for each participant is observed before the next participant needs to be allocated. Letting~$\bH_0=()$, we recursively define the trial history as\mbox{~$\bH_t=(A_1,Y_{A_1,1},A_2,Y_{A_2,2},\dots, A_t, Y_{A_t,t})$}, where the distribution of each~(random) \hbox{allocation~$A_{t}\in\{\C,\D\}$} is defined according to~$\mathbb{P}(A_t=\C\mid \bH_{t-1})=\pi(\bH_{t-1})$ for~$t=1,\dots, n$. We denote the probability distribution over trial histories under parameter vector~$\btheta$ and RAR procedure~$\pi$ by~$\mathbb{P}^\pi_\btheta$.

Let~$S_{a,t}$ denote the (random) number of successes and~$N_{a,t}$ denote the~(random) number of allocations for treatment arm~$a$ after allocating participant~$t$, i.e., for all~$a\in\{\C,\D\}, \;t\in\{1,\dots,n\}$~$$S_{a,t}= \sum_{u=1}^t Y_{a,u}\mathbb{I}(A_u=a),\quad N_{a,t} = \sum_{u=1}^t \mathbb{I}(A_u = a),$$ where~$\mathbb{I}$ is the indicator function. 
We now introduce a variable representing the trial \mbox{state,~$\bX_t = (\bS_{t},\bN_{t})$}. which summarizes the trial information through the successes and allocations for each arm up to allocating participant~$t$. 
Letting~$\bX_0=((0,0),(0,0))$ be the initial state, the state space for~$\bX=(\bX_t)_{t=0}^n$ equals~$\mathcal{X} = \cup_{t=0}^n\mathcal{X}_t$ where for all~$t$ we define:$$ \mathcal{X}_t=\{((s'_{\C},s'_{\D}),(n'_\C,n'_\D)): \bs',\bn'\in \{0,\dots, t\}^2,\,\bs'\leq \bn,\;\textstyle \sum_{a\in\{\C,\D\}}n_{a}=t \}.$$
By, e.g.,~\citet[Equation (1)]{yi2013exact}, it follows that
   for all states~$\bx_t \in\mathcal{X}_t$\begin{equation}\mathbb{P}_\btheta^\pi(\bX_t=\bx_t)= g_t^\pi(\bx_t)\prod_{a\in\{\C,\D\}}\theta_a^{s_a(\bx_t)}(1-\theta_a)^{n_{a}(\bx_t)-s_{a}(\bx_t)}\label{expression_likelihood},\end{equation}
   where~$n_{a}$, $s_{a}$ record the treatment group sizes and successes encoded in~$\bx_t$. The coefficients~$g_t^\pi(\bx_t)$, representing the sum of allocation path probabilities leading up to a state~$\bx_t$,  can be calculated using a network algorithm~\citep{mehta1988constructing}. An important thing of note is that network algorithms usually compute the probabilities~\eqref{expression_likelihood}, whereas~\citet{baas2025exact} uses it to compute the parameter-free component~$g_t^\pi$.
   \citet{baas2025exact} additionally provided a computationally efficient algorithm to compute~$g_t^\pi(\bx_t)$ for trial sizes up to~$1,000$ participants.  

\subsection{Asymptotic and exact Wald tests}~\label{sect:tests}
We consider conditional and unconditional exact tests for testing $$H_0:\theta_\C=\theta_\D\text{ versus }H_1:\theta_\C\neq\theta_\D$$
using the Wald statistic~$\Tau(\bX_{\Tend})$which, assuming~$\min_a n_a(\bX_\Iend)>0$, is defined as:
\begin{equation}\Tau(\bX_{\Tend})=\begin{cases} \frac{\hat{\theta}_\D(\bX_{\Tend})-\hat{\theta}_\C(\bX_{\Tend})}{\sqrt{\frac{\hat{\theta}_\C(\bX_{\Tend})(1-\hat{\theta}_\C(\bX_{\Tend}))}{{n}_{\C}(\bX_{\Tend})} + \frac{\hat{\theta}_\D(\bX_{\Tend})(1-\hat{\theta}_\D(\bX_{\Tend}))}{{n}_{\D}(\bX_{\Tend})}}},&\text{if~$\exists a\;\text{s.t.}\; \hat{\theta}_a(\bX_\Iend)\in(0,1),$}\\
\sign(\hat{\theta}_\D(\bX_\Iend)-\hat{\theta}_\C(\bX_\Iend))\cdot\infty,&\text{else},
\end{cases}\label{defn:WS}\end{equation}
where~$\hat{\theta}_a(\bX_{\Tend}) = s_{a}(\bX_{\Tend})/{n}_{a}(\bX_{\Tend})$ for all~$a\in\{\C,\D\}$ and we use the convention~$0\cdot\infty = 0 $.

\citet{Yi_Wang_2011} showed that, under several specific conditions for the RAR procedure, the restricted maximum likelihood estimator under~$H_0$ exists and satisfies a central limit theorem, in which case~T$(\bX_{\Iend})$ converges to a standard normal random variable.  This result asymptotically justifies the two-sided asymptotic Wald test that rejects~$H_0$ when~$|\Tau(\bX_{\Iend})|\geq z_{1-\alpha/2}$ for~$z_{1-\alpha/2}$ the inverse cumulative standard normal distribution function \mbox{at~$1-\alpha/2$}
and a nominal significance level~$\alpha\in[0,1]$. 

\begin{remark}[Burn-in period]
In~\eqref{defn:WS} it is assumed that~$\min_a n_a(\bX_\Iend)>0$, and, in the current paper, this is enforced through a {\it burn-in period} in which~$b
$ participants are allocated treatment in a way such \mbox{that~$N_{\C,2b}=N_{\D,2b}=b.$} 
Several allocation methods can be used to ensure a balanced assignment of participants to treatment arms during the burn-in period. In practice, this choice is particularly important when the burn-in length is small, as it is then crucial to minimize the probability of treatment imbalances.
Commonly used procedures in this setting include the truncated binomial design, big stick design, permuted block design, and the random allocation rule, each with its own benefits and challenges~\citep{berger2021roadmap}.
Despite the practical relevance of choosing an appropriate burn-in allocation method, it is important to note that, for our evaluation, the specific procedure used does not impact our results. 
Under the model of~\autoref{subsect:model}, the distribution of~$\bX_{\Iend}$ is invariant to the specific allocation procedure used during the burn-in phase, as long as balanced treatment group sizes are achieved. Hence, we will not make an assumption on the procedure used to enforce balanced allocation during the burn-in period.
\end{remark}

\subsubsection{Conditional and unconditional exact tests}

This section describes conditional and unconditional exact tests for clinical trials employing an RAR procedure. Such tests were, to the best of our knowledge, first considered in~\citet{wei1990statistical}. The conditional exact test we consider is not based on the treatment group sizes and hence coincides with the test suggested by a reviewer of~\citet{begg1990inferences}. 

We first describe the conditional exact test for RAR designs based on total successes. 
For equal allocation, the conditional exact test reduces to Fisher's exact test~\citep{Fisher1934method}. 
In this test, the nuisance parameter~$\theta = \theta_\C=\theta_\D$ under the null hypothesis is eliminated by conditioning on the total sum of successes in the trial~$s(\bX_n)=\sum_{a\in\{\C,\D\}}s_a(\bX_\Iend)$.
Following, e.g.,~\citet{baas2025exact}, for~$\btheta$ such that~$\theta_\C=\theta_\D$:
\begin{equation}\mathbb{P}_\btheta^\pi(\bX_n=\bx_n\mid s(\bX_n) = s') = \binom{\Iend}{s'}^{-1}g_{{\Iend}}^\pi(\bx_{\Iend})\cdot \mathbb{I}(s(\bx_n)=s'),\label{cond_dist}\end{equation}
as the right-hand side above does not depend on~$\btheta$, one can calculate the left and right conditional two-sided critical values~(both of level~$\alpha/2$) based on~$\Tau(\bX_n)$ given~$s(\bX_n) = s'$, denoted~$\bar{c}(s')$ and~$\ubar{c}(s')$, respectively. The resulting test is conditionally exact for each total sum of successes, i.e., for all possible total sums of successes~$s'$ we have~$$\mathbb{P}(\Tau(\bX_n)\notin (\ubar{c}(s'), \bar{c}(s')) \mid s(\bX_n)=s')\leq \alpha.$$

Next, we consider the unconditional exact test for RAR designs, extending Barnard's unconditional exact test~\citep{barnard1945new}. In comparison to the conditional exact test, this test only uses two critical values, ensuring a maximum unconditional type I error rate below the nominal significance level~$\alpha$, i.e., it determines the upper critical value~$\bar{c}$ such that 

       $$\max_{\substack{\btheta \in [0,1]^2,\\\theta_\C = \theta_\D}} \sum_{\bx_{\Iend}  : \,\Tau(\bx_{\Iend}) \geq \bar{c}} \mathbb{P}^\pi_{\btheta}(\bX_{\Iend} = \bx_{\Iend}) \leq \alpha/2.$$
    The lower critical value~$\ubar{c}$ of the unconditional exact test is defined similarly using the left tail of the distribution.
An efficient algorithm to compute~$\bar{c}$ and~$\ubar{c}$ with a desired precision is provided in~\citet{baas2025exact}. The resulting test is unconditionally exact, meaning~$\mathbb{P}_\btheta^{\pi}(\Tau(\bX_n)\notin (\ubar{c}, \bar{c}))\leq\alpha$ for all~$\btheta$ such that~$\theta_\C=\theta_\D$, but not necessarily conditionally exact based on the total sum of successes, i.e., it may be that for some~$s'$ and~$\btheta$ such that~$\theta_\C=\theta_\D$ we have~$\mathbb{P}_\btheta^{\pi}(\Tau(\bX_n)\notin (\ubar{c}, \bar{c})\mid s(\bX_n)=s')>\alpha$.

\subsubsection{Generalized Boschloo test}\label{sect:generalized_boschloo}

The generalized version of Boschloo's test
uses, in the vein of~\citep{boschloo1970raised},  the p-values corresponding to the conditional exact Wald test as test statistics in the unconditional exact test to raise the unconditional maximum type I error rate to~$\alpha$.
Unlike Boschloo's test, the generalized version of Boschloo's test depends on a test statistic, as we only condition on the sum of successes. Following~\eqref{cond_dist}, the conditional p-value for the Wald test based on the sum of total successes is given by
$$\Tau_{\text{GB}}(\bx_n) = \sum_{\substack{\bx'_n:s(\bx'_n)=s(\bx_n),\\ |\Tau(\bx'_n)|\geq |\Tau(\bx_n)|}} \binom{n}{s(\bx'_n)}^{-1}g^\pi(\bx'_n).$$
In the generalized version of Boschloo's test, these p-values are used as a test statistic in the one-sided unconditional exact test, where we only reject if~$\Tau_{\text{GB}}$~(i.e., the p-value) is low.

Let~$\varphi(((s'_\C,s'_\D),(n'_\C,n'_\D)) = ((s'_\D,s'_\C),(n'_\D,n'_\C))$ swap the successes and treatment group sizes per arm for \mbox{all~$((s'_\C,s'_\D),(n'_\C,n'_\D))\in\mathcal{X}_n$.} Let~$\Tau'$ be a antisymmetric test statistic if-and-only-if~$\Tau'(\varphi(\bx_n))=-\Tau'(\bx_n)$ for \mbox{all~$\bx_n\in\mathcal{X}_n$}, it can be verified that the Wald statistic~$\Tau$ as defined in~\eqref{defn:WS} satisfies this property.
Let~$\pi$ be a symmetric RAR procedure  \mbox{when~$g^\pi_n(\bx_n)=g^\pi_n(\varphi(\bx_n))$} for all~$\bx_n\in\mathcal{X}_n$.

The next lemma provides conditions under which the generalized version of Boschloo's test is uniformly more powerful~(unconditionally) than the conditional exact test.
\begin{lemma}
Let~$\ubar{c}_{\mathrm{GB}}$ be the (left) unconditional exact critical value for significance level~$\alpha$ based on~$\Tau_{\mathrm{GB}}$.
    If~$\pi$ is a symmetric RAR procedure, then:
    $$\Tau(\bX_{\Iend})\geq \bar{c}(s(\bX_{\Iend}))\text{ or  }\Tau(\bX_{\Iend})\leq \ubar{c}(s(\bX_{\Iend}))\implies \Tau_{\mathrm{GB}}(\bX_n)\leq \ubar{c}_{\mathrm{GB}}.$$
\end{lemma}
\begin{proof}
 We first show that the conditional Wald test coincides with the test rejecting when T$_{\text{GB}}\leq \alpha$~(which holds if~$\pi$ is symmetric). As~$s(\bx_n)=s(\varphi(\bx_n))$, we have that for every total sum of successes~$s'$, by the properties of~$\Tau$ and~$\pi$:
 $$\sum_{\substack{\bx_{\Iend}: s(\bx_n)=s',\\\Tau(\bx_{\Iend})\geq c}} g_{{\Iend}}^\pi(\bx_{\Iend}) = \sum_{\substack{\bx_{\Iend}: s(\bx_n)=s',\\-\Tau(\varphi(\bx_{\Iend}))\geq c}} g_{{\Iend}}^\pi(\varphi(\bx_{\Iend}))=\sum_{\substack{\bx'_{\Iend}: s(\bx_n')=s',\\-\Tau(\bx'_{\Iend})\geq c}} g_{{\Iend}}^\pi(\bx_{\Iend}'). $$
 Hence~$\ubar{c}(s')=-\bar{c}(s')$ as both critical values are based on a level~$\alpha/2$, hence the conditional exact test rejects if-and-only-if~$|\Tau(\bx_{\Iend})|\geq \bar{c}(s(\bx_{\Iend})),$ i.e.,
\begin{align*}
    &|\Tau(\bx_{\Iend})|\geq \bar{c}(s(\bx_n))\iff \Tau_{\text{GB}}=\sum_{\substack{\bx'_n:s(\bx'_n)=s(\bx_n),\\ |\Tau(\bx'_n)|\geq |\Tau(\bx_n)|}} \binom{n}{s(\bx'_n)}^{-1}g^\pi(\bx'_n)\leq \alpha.
\end{align*}
As the test rejecting when~$\Tau_{\text{GB}}\leq\alpha$ equals the conditional exact test, we have~$\ubar{c}_{\text{GB}}\geq \alpha$ and hence
\begin{align*}
 \Tau(\bX_{\Iend})\geq \bar{c}(s(\bX_{\Iend}))\text{ or  }\Tau(\bX_{\Iend})\leq \ubar{c}(s(\bX_{\Iend})\iff& |\Tau(\bX_{\Iend})|\geq \bar{c}(s(\bX_n))\\\implies& \Tau_{\text{GB}}(\bX_n)\leq \ubar{c}_{\mathrm{GB}}.
\end{align*}
\end{proof}

\subsection{Operating characteristics and their exact calculation}\label{OCscalc}
The likelihood~\eqref{expression_likelihood} allows for the calculation of exact operating characteristics, such as the power or type I error rate.
Any operating characteristic can be written as~$\mathbb{E}_{\btheta}^\pi[f(\bX_{\Iend})]$ for a function~$f:\mathcal{X}_{\Iend}\mapsto\mathbb{R}$.
The operating characteristics considered in the current paper will be:
\begin{itemize}
     \item {\bf Rejection rate (type I error rate or power)}:\\ This operating characteristic is calculated as $\mathbb{E}_{\btheta}^\pi[f(\bX_{\Iend})]$ where~$$f(\bx_{\Iend})=\mathbb{I}\Big(\Tau(\bx_{\Iend})\leq \ubar{c}(\bx_n)\text{ or }\Tau(\bx_{\Iend})\geq \bar{c}(\bx_n)\Big)$$ for a pair of critical values~$\ubar{c},\bar{c}$, which may depend on~$\bx_n$ through a summary of~$\bx_n$~(e.g., the total sum of successes in the conditional exact test). 
    If~$\btheta$ satisfies~$H_0$, the rejection rate equals the type I error rate, whereas the rejection rate equals statistical power when~$\theta_\C\neq \theta_\D$. For a nominal significance level~$\alpha$, it is desired to have the type I error rate bounded by~$\alpha$, while higher power is better.
    \item {\bf Patient benefit~(expected
proportion of allocations on the superior arm)}:\\ This operating characteristic equals~$$
\left[\sum_{a\in\{\C,\D\}}\mathbb{E}^\pi_\btheta[n_a(\bX_{\Iend})/\Iend]\mathbb{I}(\theta_a=\max_a\theta_a)\right]-1/2\mathbb{I}(\theta_\C=\theta_\D).$$ The patient benefit represents the proportion of participants on the superior arm. Other measures, such as the expected number of successes~$\mathbb{E}_{\btheta}^\pi[s(\bX_n)]$ could also be used to denote patient benefit, but an advantage of using the proportion is that we should aim to have a patient benefit higher than an absolute value of~$50\%$. 
\end{itemize}

We have from the likelihood~\eqref{expression_likelihood} that
\begin{equation}\mathbb{E}_{\btheta}^\pi[f(\bX_{\Iend})]=\sum_{\bx_{\Iend}\in\mathcal{X}_{\Iend}}f(\bx_{\Iend})g^\pi_{\Iend}(\bx_{\Iend})\prod_{a\in\{\C,\D\}}\theta_a^{s_a(\bx_{\Iend})}(1-\theta_a)^{n_{a}(\bx_{\Iend})-s_a(\bx_{\Iend})},\label{eqn:OCcalc}\end{equation}
which can be written as~$(\bff\circ\bg_{\Iend}^ \pi)^\top\bp_{\btheta}$ with~$\circ$ the Hadamard product and~$\bp_{\btheta}$ containing the product-terms in the above expression. Hence, we can \mbox{store~$\bff\circ\bg_{\Iend}^\pi$} and then only have to take the inner product with~$\bp_\btheta$ for different vectors~$\btheta$ when we want to calculate~\eqref{eqn:OCcalc} for different parameters~$\btheta$.

\section{Response-adaptive randomization procedures}\label{sect:RA_procedures}
In this section, we list the power-maximizing RAR procedures considered in this paper, which are RAR procedures that maximize a metric related to power. 
We will first list response-adaptive procedures known from the literature~(including modifications to protect patient benefit), after which we introduce our considered novel power-maximizing {\it constrained Markov decision process}~(CMDP) procedures. 
\subsection{Response-adaptive randomization procedures known from the literature}\label{NA_description}%
\begin{itemize}
    \item {\bf DBCD Neyman allocation}:\\
      Under Neyman allocation~\citep{Neyman1934} the expected proportion of  participants randomized to the control group  equals
\begin{align}
	\rho(\btheta) = \frac{\sqrt{\theta_\C(1-\theta_\C)}}{\sqrt{\theta_\C(1-\theta_\C)}+\sqrt{\theta_\D(1-\theta_\D)}}\label{opt_prop}.
\end{align}
As indicated in~\citet{Hu01092003,rosenberger2004}, the allocation proportion~\eqref{opt_prop} maximizes the~(asymptotic) power function for the asymptotic Wald test described at the start of~\autoref{sect:tests}.
The parameters~$\btheta$ are unknown, and in this paper the proportion~\eqref{opt_prop} is targeted through a {\it doubly adaptive biased coin design}~\citep[DBCD,][]{eisele1994dabcd, hu2006theory}, based on the sequential estimation procedure for~$\theta_a$ in Example~2 of~\citet{hu2009ERADE}, which is a flexible targeting procedure which, in comparison to, e.g., efficient randomized-adaptive designs~\citep{hu2009ERADE}, satisfies the property that the randomization probability to control converges almost surely to~\eqref{opt_prop}, resulting in stability for large samples. This power-maximizing RAR procedure is referred to as DBCD Neyman allocation. Examples of other RAR procedures to target optimal proportions, such as~$\eqref{opt_prop}$
 are the efficient randomized design~\citep{hu2009ERADE} and the mass-weighted urn design~\citep{ZHAO2015209}.
\item {\bf Tempered DBCD Neyman allocation}:\label{sect:TNA}
Let~$\pi_\text{NA}(\bx_t)$ be the probability of randomizing a participant to the control group based on a state~$\bx_t$ under DBCD Neyman allocation.
Depending on~$\bx_t$ it can be the case that~$\pi_\text{NA}(\bx_t)>0.5$ while~$\tilde{\theta}_\C(\bx_t)< \tilde{\theta}_\D(\bx_t)$ or~$\pi_\text{NA}(\bx_t)<0.5$ while~$\tilde{\theta}_\C(\bx_t)> \tilde{\theta}_\D(\bx_t)$, i.e., the probability of allocating to the inferior treatment~(based on the estimates~$\tilde{\theta}_a(\bx_t)$) is higher than~50\%. In~\citet{zhangrosenberger}, this problem was identified in a different setting, where it was suggested to let the target allocation proportion  be~$50\%$ in this case, i.e., in the  tempered DBCD Neyman allocation procedure, we have
$$\pi_\text{TNA}(\bx_t) = \begin{cases}
    \pi_{\text{NA}}(\bx_t),\quad&\text{if~$  \pi_{\text{NA}}(\bx_t)>0.5\text{ and }\tilde{\theta}_\C(\bx_t)> \tilde{\theta}_\D(\bx_t)$}\text{ or }\\
    &\text{if~$  \pi_{\text{NA}}(\bx_t) < 0.5\text{ and }\tilde{\theta}_\D(\bx_t)> \tilde{\theta}_\C(\bx_t)$,}\\
    1/2,&\text{else.}
\end{cases}$$
The tempered DBCD Neyman allocation procedure is hence a power-maximizing RAR procedure under the constraint that the allocation probability does not lead to more participants being allocated inferior treatment on average.
\item {\bf Equal allocation}: In case of equal allocation, assuming~$\Iend$ is even, we have:
    $$g_n^{\pi_{\text{EA} }}(\bx_{\Iend}) = \binom{n_\C(\bx_n)}{s_\C(\bx_n)}\binom{n_\D(\bx_n)}{s_\D(\bx_n)}\mathbb{I}(n_\C(\bx_{\Iend})/\Iend=n_\D(\bx_{\Iend})/\Iend=1/2),$$
    i.e., the function~$g_n^{\pi_{\text{EA}}}$ concentrates on states where the proportion of participants allocated to either treatment group is 1/2, and equals the product of the number of possible outcome paths length~$n/2$ having~$s_a(\bx_n)$ successes for each arm~$a$. Similar as for the burn-in period, the actual procedure~$\pi_{\text{EA}}$ corresponding to~$g^{\pi_{\text{EA}}}$ is irrelevant for the analysis at hand (as we only look at the distribution over final states), and can, e.g, correspond to a truncated binomial, big stick, or permuted block design~\citep{berger2021roadmap}. 
    As shown in~\citet{azrielmandelrinott2011}, equal allocation is a power-maximizing randomization procedure under the Pitman asymptotic testing framework.
    
        \item {\bf Bayesian RAR}: We consider a Bayesian RAR design following~\citet{thompson1933likelihood,THALL2007859}. 
        For a prior probability measure~$\mathbb{Q}$ on~$[0,1]^2$, we have
        \begin{align*}
             {\pi}_\text{B}(\bx_t)&=\frac{\mathbb{Q}(\theta_\C\geq\theta_\D\mid \bX_t=\bx_t)^{t/(2\Iend)}}{\mathbb{Q}(\theta_\C\geq\theta_\D\mid \bX_t=\bx_t)^{t/(2\Iend)} + \mathbb{Q}(\theta_\C<\theta_\D\mid \bX_t=\bx_t)^{t/(2\Iend)}}.
        \end{align*}
        Hence, the probability~$\pi_{\text{B}}(\bx_{\Iend})$ equals the power-tuned posterior probability that the control treatment is better, based on the prior distribution~$\mathbb{Q}$~(usually the uniform distribution for~$\btheta$). This randomization procedure targets high expected treatment outcomes. The probabilities~$\mathbb{Q}(\theta_\C\geq\theta_\D\mid \bX_t=\bx_t)$ can be calculated efficiently based on Algorithm~1 in~\citet{kaddaj2025thompsonulamgaussmulticriteria}.
        Bayesian RAR is not a power-maximizing RAR procedure, and is included as a comparator RAR procedure as it is commonly applied in RAR trials~\citep{Pin2025}. 
\end{itemize}

\subsection{Constrained Markov decision process procedures}\label{CMDP_P}
For a prior probability measure~$\mathbb{Q}$ for~$\btheta$ with support~$[0,1]^2$, let a Bayesian average operating characteristic equal~$\int_{[0,1]^2}\mathbb{E}_\btheta^\pi[f(\bX_{\Iend})]\mathbb{Q}(d\btheta)$.
The goal of the {\it constrained Markov decision process}~(CMDP) approach~\citep{baas2024CMDP} is to find an RAR procedure that sequentially randomizes participants to treatment to maximize a Bayesian average operating characteristic, while ensuring control of other Bayesian average operating characteristics. For a fixed maximum randomized allocation rate~$p\in[0.5,1]$ the randomization probabilities~$\pi(\bx_t)$ are restricted to the set~$\{1-p,1/2,p\}$, e.g., when~$p=0.95$ participants can only be assigned to a treatment with probability $5\%$, $50\%$ or~$95\%$.
In~\citet{baas2024CMDP}, CMDP procedures are found through an optimization approach that uses a cutting plane algorithm applying backward recursion in every iteration. 

The CMDP procedures constructed in this paper are novel as~\citet{baas2024CMDP} (i) only considered maximizing expected treatment outcomes and did not consider maximizing power,~(ii) did not have point-wise constraints (i.e., constraints for specific parameter values), and~(iii) used an uncorrected version of Fisher's exact test to analyze the data, whereas in this paper exact Wald tests are constructed based on CMDP procedures.

The novel {\it CMDP-P} (CMDP-power) procedure targets maximum Bayesian average power~(hence it is a power-maximizing RAR procedure), while the constraints of the procedure relate to type I error rate control.  A CMDP-P procedure is an approximately optimal policy~\citep[found through the procedure in][Section~3.3]{baas2024CMDP} for the following optimization problem:
\begin{subequations}\label{cmdp_p}
    \begin{align}
    &\underset{\pi}{\max}\,\mathbb{P}_1
^{\pi}\left(|\Tau(\bX_n)|\geq z_{1-\alpha/2}\right)\label{obj_CMDP_p}\\
    &\textnormal{s.t.}\nonumber\\
    &\mathbb{P}^\pi_0(|\Tau(\bX_n)|\geq z_{1-\alpha/2})\leq \alpha',\label{avg_typeI_CMDP_P}\\
&\mathbb{P}^\pi_{\btheta_0}(|\Tau(\bX_n)|\geq z_{1-\alpha/2})\leq \alpha'',\quad \forall\btheta_0\in\hat{\Theta}_0.\label{PW_typeI_CMDP_P}
\end{align}
\label{CMDP_POWER}
\end{subequations}
In the above,~$\hat{\Theta}_0\subset\{\btheta_0\in[0,1]^2:\theta_{0,\C}=\theta_{0,\D}\}=:\Theta_0$ has finite cardinality and, following~\autoref{expression_likelihood}, we have
\begin{align}
    \mathbb{P}^\pi_1(\bX_n=\bx_n) &= \int_{[0,1]^2}\hspace{-3mm}\mathbb{P}_{\btheta}^\pi(\bX_n=\bx_n)d\btheta =g_n^{\pi}(\bx_n)\prod_{a\in\{\C,\D\}}\hspace{-3mm}\text{B}(\tilde{s}_a(\bx_n), \tilde{f}_a(\bx_n)),\label{defP1}\\
      \mathbb{P}^\pi_0(\bX_n=\bx_n) &= \int_{0}^1\mathbb{P}_{(\theta,\theta)}^\pi(\bX_n=\bx_n)d\theta =g_n^{\pi}(\bx_n)\text{B}(\tilde{s}(\bx_n), \tilde{f}(\bx_n)),\label{defP0}
\end{align}
where~$\tilde{s}_a=s_a(\bx_n)+1,\tilde{f}_a(\bx_n)=n_a(\bx_n)-s_a(\bx_n)+1$ equal one plus the number of treatment successes and failures in arm~$a$,~$\tilde{s}(\bx_n)=1+\sum_as_a(\bx_n)$, \mbox{$\tilde{f}(\bx_n)=1+\sum_af_a(\bx_n)$} equal one plus the total amount of successes and failures, and~$\text{B}$ denotes the beta function.

The objective function in~\eqref{obj_CMDP_p} equals the Bayesian average power of the Wald test under the uniform prior~$\mathbb{Q}$ on~$[0,1]^2$, which assumes all success rates to be equally likely a priori.
The measure represents power as the null hypothesis~$H_0$ has zero measure under the integral in~\eqref{defP1}. The constraint~\eqref{avg_typeI_CMDP_P} enforces the Bayesian average type I error rate over the set~$\Theta_0$, assuming a uniform prior, to be less than~$\alpha'$, where it is recommended to set~$\alpha'\leq\alpha$. The constraint~\eqref{PW_typeI_CMDP_P} ensures point-wise type I error rate control for all \mbox{points~$\theta_0\in\hat{\Theta}_0,$} where it is necessary that~$\alpha''\leq\alpha.$
A novelty in the current paper, not considered in~\citet{baas2024CMDP}, is that the set of RAR procedures~$\pi$ that we optimize over is restricted to RAR procedures with a burn-in period~$b.$ This is enforced through a restriction of the action space for the first~$2b$ decision epochs.
We note that in practical situations the priors used in~\eqref{CMDP_POWER} as well as the values in~$\hat{\Theta}_0$ (used in~\eqref{PW_typeI_CMDP_P}) 
may also be based on historical information about the treatment success rates.

In order to protect patient benefit in the CMDP framework, we define the {\it CMDP-BP} (CMDP-benefit-and-power) procedure.
The CMDP-BP procedure has the same objective and constraints as the CMDP-P procedure~(and is hence a power-maximizing RAR procedure), with the additional constraint that the average patient benefit should at least be higher than 0.5 for several subsets of the parameter range.  
Choosing lower and upper interval \mbox{endpoints~$\ell_{a,\sigma}\leq u_{a,\sigma}\in[0,1]$} for all~$a\in\{\C,\D\},\sigma\in\mathcal{S}\subset\mathbb{N}$ such that we \hbox{have~$[\ell_{\C,\sigma},u_{\C,\sigma}]\cap[\ell_{\D,\sigma},u_{\D,\sigma}]=\emptyset$,} defining~$ \mathbb{P}^\pi_\sigma(\bX_n=\bx_n)$ as the Bayesian average probability of observing~$\bx_n$ under an uniform prior \hbox{over~$[\ell_{\C,\sigma},u_{\C,\sigma}]\times[\ell_{\D,\sigma},u_{\D,\sigma}]$,} and letting~$F_{\text{B}(\tilde{s}, \tilde{f})}$ denote the cumulative Beta distribution function for parameters~$\tilde{s}, \tilde{f}$, we have
\begin{align*}
\mathbb{P}^\pi_\sigma(\bX_n=\bx_n)&=\frac{\int_{\ell_{\D,\sigma}}^{u_{\D,\sigma}}\int_{\ell_{\C,\sigma}}^{u_{\C,\sigma}}\mathbb{P}_{\btheta}^\pi(\bX_n=\bx_n)d\theta_\C d\theta_\D}{\prod_{a\in\{\C,\D\}}(u_{a,\sigma}-\ell_{a,\sigma})}\\ &=g_n^{\pi}(\bx_n)\prod_{a\in\{\C,\D\}}\frac{F_{\text{B}(\tilde{s}_a(\bx_n), \tilde{f}_{a}(\bx_n))}(u_{a,\sigma})-F_{\text{B}(\tilde{s}_a(\bx_n), \tilde{f}_{a}(\bx_n))}(\ell_{a,\sigma})}{(u_{a,\sigma}-\ell_{a,\sigma})/\text{B}(\tilde{s}_a(\bx_n), \tilde{f}_{a}(\bx_n))},
\end{align*}
A CMDP-BP procedure is an approximately optimal policy for the optimization problem, obtained by adding the following constraint to~\eqref{cmdp_p}:
    \begin{align*}
&\mathbb{E}^\pi_{\sigma}[N_{\arg\max_{a\in\{\C,\D\}} \theta_a}/n]\geq 1/2,\quad \forall \sigma\in\mathcal{S}.\label{allocfraccons}
\end{align*}
Note that the restriction~$[\ell_{\C,\sigma},u_{\C,\sigma}]\cap[\ell_{\D,\sigma},u_{\D,\sigma}]=\emptyset$ enforces that one of the arms always has a higher success rate in every rectangle considered.

\section{Numerical results}\label{sect:numerical_results}
This section reports the main numerical results of the paper, while more detailed results are deferred to the Appendix. All appendix tables are summarized in~\autoref{summary_table}. For generating the results, we set the significance levels~$\ubar{\alpha}=\bar{\alpha}=0.025$, evaluate trial sizes~$\Iend\in\{50,250\}$ (representing a small, e.g., early-phase, trial size and moderate, e.g., confirmatory, trial size, respectively), evaluate the performance of the conditional exact, unconditional exact, asymptotic, and generalized version of Boschloo's Wald tests, set the burn-in period length equal to~$6$ participants per arm as in~\citet{pin2025revisitingoptimalallocationsbinary}, set the tuning parameter~$
\gamma$ to~$2$ for the DBCD procedure~\citep[as suggested in][]{rosenberger2004}, set~$p=0.95$ as in~\citet{baas2024CMDP}, \mbox{$\alpha'=0.045$}, \mbox{$\alpha''=0.05$}, \mbox{$\hat{\Theta}_0=\{0.00, 0.05,\dots, 1.00\}$} for the CMDP procedures, while choosing the interval endpoints as~$$(\ell_{a,\sigma}, u_{a,\sigma}) \in\{(0, 0.05),(0.05, 0.1), (0.1, 0.25), (0.25, 0.5),(0.5, 0.75), (0.75, 1.0)\}$$ for the CMDP-BP procedure.

\subsection{Response-adaptive randomization  targeting optimal proportions}\label{sect:eval_proportion_RA}
We first describe the results for (more traditional) RAR procedures based on targeting an optimal proportion~(either in terms of power or expected outcomes), which are the DBCD Neyman allocation, tempered DBCD Neyman allocation, and Bayesian RAR procedures. 

\subsubsection{DBCD Neyman allocation: Type I error rates}
The top row of~\autoref{fig:NA_plots} shows type I error rate profiles under the different tests for DBCD Neyman allocation~(\autoref{NA_description}). The asymptotic Wald test shows substantial type I error rate inflation under DBCD Neyman allocation for null success rates close to the endpoints of the parameter range, reaching values around~15\% for both considered trial sizes. \autoref{fig:simulation_paths} shows a potential cause for this behavior: the variance in the allocation proportions is higher under large and low success rates, causing high variance of the test statistic under the null hypothesis. This could be explained by the higher sensitivity of the target proportion to small estimated treatment effects for these parameters, combined with an increased variance through employing a RAR procedure. For instance, for~$\hat{\theta}_\C=0.50,\;\hat{\theta}_\D=0.52$ we have~$\rho(\hat{\btheta})=0.50$ whereas for~$\hat{\theta}_\C=0.97,\;\hat{\theta}_\D=0.99$ we have~$\rho(\hat{\btheta})=0.63$.

The unconditional exact test, which adjusts the critical values based on the type I error rate function for the asymptotic test, shows a similar type I error rate profile shape as the asymptotic test and is conservative for null success rates around~50\%. On the contrary, the conditional exact and the generalized version of Boschloo's test~(henceforth referred to simply as Boschloo's test) show a less conservative type I error rate profile, which is often close to the significance level, especially for the larger trial horizon~$n=250.$ A similar result for the conditional exact test was found for RAR procedures targeting high expected treatment outcomes in~\citet{baas2025exact}.

\subsubsection{DBCD Neyman allocation: Power}
The middle row of~\autoref{fig:NA_plots} shows differences in power under  DBCD Neyman allocation and equal allocation~(subtracted). Each open marker and vertical line corresponds to a point such \mbox{that~$\theta_\C=\theta_\D=\tilde{\theta}$,} while each point on the respective curves starting from an open marker corresponds to a parameter vector such that~$\theta_\D\geq \theta_\C=\tilde{\theta}$. When~$\theta_\D$ becomes larger than the point~$\tilde{\theta}'$ corresponding to a new open marker, a new curve starts, and the visibility of the old curve is decreased. 

As expected, the asymptotic Wald test has higher power under DBCD Neyman allocation than under equal allocation. The unconditional exact test, on the other hand, has roughly uniformly lower power than the unconditional exact test under equal allocation. This brings the added value of using  DBCD Neyman allocation to optimize power into question, as the unconditional exact test is often the most powerful exact test under equal allocation~\citep[][see also~\autoref{summary_table}]{mehrotra2003cautionary}. The conditional exact test often shows higher power under DBCD Neyman allocation compared to equal allocation, which is likely due to the conservativeness of the conditional exact test~(i.e., Fisher's exact test) for equal allocation. In contrast, Boschloo's test shows a slightly worse performance in terms of power for DBCD Neyman allocation when compared to equal allocation.

\begin{figure}[h!]
\centering \includegraphics[width = .85\textwidth]{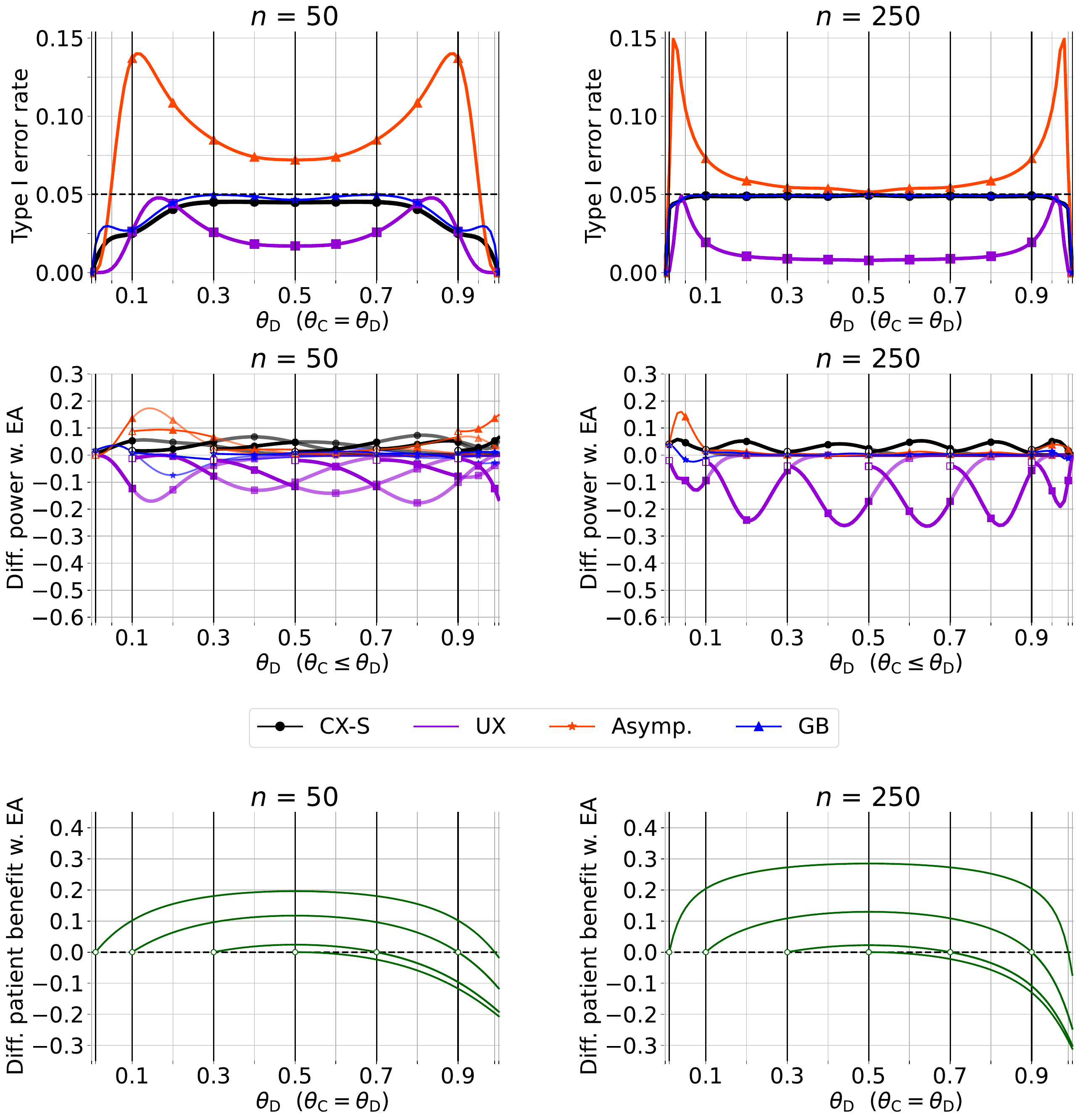}
    \caption{Results for doubly adaptive biased coin design Neyman allocation; type I error rate profiles~(top), power difference curves with respect to the same test under equal allocation~(EA, middle), and patient benefit difference curves with respect to equal allocation~(EA, bottom) for trial sizes~$n\in\{50,250\}$. The evaluated tests are the conditional exact, unconditional exact~(UX), asymptotic~(Asymp.), and generalized version of Boschloo's~(GB) Wald tests. For all tests, the significance level was set to~$0.05$.}
    \label{fig:NA_plots}
\end{figure}

\subsubsection{DBCD Neyman allocation: Patient benefit}
The bottom row of~\autoref{fig:NA_plots} shows the difference in patient benefit under DBCD Neyman allocation compared to equal allocation~(subtracted). Each marker point corresponds \hbox{to~$\theta_\C=\theta_\D=\tilde{\theta}$,} while each respective curve corresponds to a parameter vector such that~$\theta_\D\geq \theta_\C=\tilde{\theta}$. The bottom row of~\autoref{fig:NA_plots} shows that the patient benefit under DBCD Neyman allocation becomes lower than under equal allocation whenever~$\theta_\C$ is closer to~$50\%$ than~$\theta_\D$, an issue also identified in~\citet{hu2006theory}.

\subsubsection{DBCD Neyman allocation: Additional results}
\autoref{tab_RR_neyman} shows type I error rate and power values under both Neyman~(four leftmost columns) and equal allocation~(four rightmost columns).  
Perhaps surprisingly, the asymptotic Wald test under equal allocation shows type I error rate inflation, albeit to a lesser extent than under DBCD Neyman allocation. The highest power for DBCD Neyman allocation is found under Boschloo's test, whereas the highest power for equal allocation is found under the unconditional exact test. DBCD Neyman allocation with  Boschloo's test outperforms the unconditional exact test for equal allocation in several cases; however, the power gains are lower than the power losses, in particular for the case~$\theta_\C=0.01$, where a power loss of around~$20\%$ is seen for~$n=50$ \mbox{and~$\theta_\D=0.2$} and a power loss of around~$10\%$ is seen for~$n=250$ \hbox{and~$\theta_\D=0.05$.} The highest power gain under Boschloo's test and DBCD Neyman allocation is around 2\% for $n=50$ and around~$0.3\%$ for~$\Iend=250.$

\subsubsection{Tempered DBCD Neyman allocation and Bayesian RAR}

\autoref{RR_TNA}, \autoref{RR_BRAR},  \autoref{fig:TNA_plots}, and \autoref{fig:TW_Bayesian RAR_plots} provide the detailed results for tempered DBCD Neyman allocation and Bayesian RAR.

\autoref{fig:TNA_plots} shows that the type I error rate profile looks similar to that of DBCD Neyman allocation for a common success rate lower than~0.5, while being more similar to that of equal allocation for common success rates higher than~0.5. An explanation is that the highest estimated success rate is often closer to 0.5 in the former setting. 
Power for the asymptotic Wald test under tempered DBCD Neyman allocation is about the same as the power for equal allocation for~$\theta_\C\geq 0.5$ and~$\theta_\D\geq \theta_\C$. As there is now only one type I error peak, the unconditional exact test yields lower power than for DBCD Neyman allocation. The conditional exact and Boschloo's Wald test show an improvement over Neyman and equal allocation. \autoref{RR_TNA} shows that Boschloo's test shows highest power for the tempered DBCD Neyman allocation procedure, and power values are often higher than for DBCD Neyman allocation, indicating that an increase in patient benefit does not need to stand in the way of an increase in power. 

\autoref{fig:TW_Bayesian RAR_plots} shows an asymmetry in the type I error rate profile for Bayesian RAR that is, curiously, opposite to asymmetries presented in~\citet{baas2024CMDP,Tang2025}. This is likely due to the change in the Wald statistic definition, which is based on an Agresti-Caffo correction in the papers listed above. Another thing of note is that the conditional exact test already often shows power improvements over equal allocation for the Bayesian RAR procedure. \autoref{RR_BRAR} again shows that Boschloo's test performs best for the Bayesian RAR procedure.

\subsection{Procedures based on constrained Markov decision processes}

 In this section we look at the CMDP-P and CMDP-BP procedures~(\autoref{CMDP_P}), where we expect that the asymptotic Wald test more closely matches its exact variants as the average and pointwise type I error rates ensure that the type I error rates are often bounded by the nominal significance level, while aiming for maximum average power is expected to raise the type I error rate close to the nominal significance level.

\subsubsection{CMDP-P: Type I error rates}
 The top row of~\autoref{fig:CMDP_P_plots} shows the type I error rate profiles under the different tests for the CMDP-P RAR procedure. As expected, the type I error rate of the asymptotic Wald test is often close to the nominal significance level and below the nominal significance level for all considered parameter values. 
 Due to this,  the type I error rate profile for the unconditional exact  Wald test under CMDP-P is roughly the same as that for the asymptotic Wald test. The conditional exact test is more conservative than the asymptotic and unconditional exact  Wald test, whereas Boschloo's test is substantially less conservative than the conditional exact test for~$\Iend=50$ whereas for~$\Iend=250$ no difference is observed between the two tests.

\begin{figure}[h!]
	\centering
	\includegraphics[width = .85\textwidth]{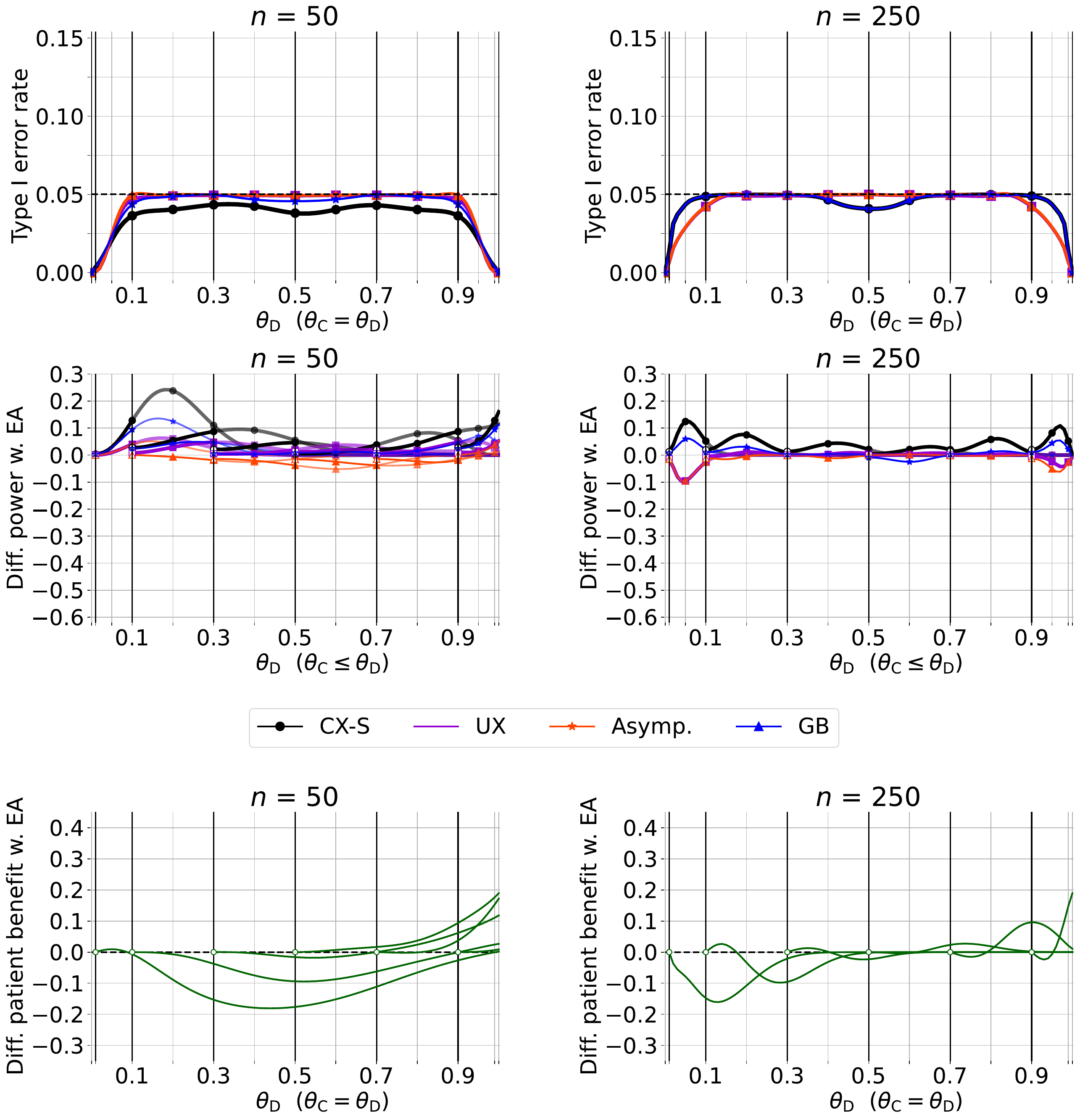}
	\caption{Constrained Markov decision process targeting high power; type I error rate profiles~(top), power difference curves compared to the same test under equal allocation~(subtracted, middle), and patient benefit difference curves with respect to equal allocation~(subtracted, bottom) for trial sizes~$n\in\{50,250\}$. We consider the conditional exact, unconditional exact~(UX), asymptotic~(Asymp.), and generalized version of Boschloo's~(GB) Wald tests. For all tests, the significance level was set to~$0.05$.}
	\label{fig:CMDP_P_plots}
\end{figure}

\subsubsection{CMDP-P: Power}
The middle row of~\autoref{fig:CMDP_P_plots} shows that for~$n=50$ the CMDP-P procedure roughly uniformly increases power for all considered exact tests as compared to equal allocation, where the largest differences are seen for the conditional exact test and Boschloo's test.  The lower power for the asymptotic Wald test as compared to equal allocation could be explained by the fact that the asymptotic Wald test shows type I error rate inflation under the equal allocation procedure, while it does not show type I error rate inflation under the CMDP-P procedure.
For~$n=250$, lower power gains and power losses are observed for the unconditional exact Wald test and Boschloo's test.
A major positive finding is that, in comparison to DBCD Neyman allocation, outperformance is now seen for the unconditional exact  Wald test under CMDP-P in comparison to equal allocation, where the major gains are seen for~$n=50$. 

\subsection{CMDP-P: Patient benefit}
The bottom row of~\autoref{fig:CMDP_P_plots} shows the difference in patient benefit under the CMDP-P procedure compared to the equal allocation procedure (subtracted). In comparison to DBCD Neyman allocation, the patient benefit differences for CMDP-P are less structured and show a less explainable behaviour. For~$n=50$, the patient benefit is often less than~$50\%.$

\subsection{CMDP-P: Additional results}
\autoref{tab:RR_CMDPP} shows power values for CMDP-P. For this RAR procedure, the exact test with highest power is the unconditional exact Wald test. For~$n=50$, the CMDP-P procedure uniformly shows higher power than equal allocation under this test. For~$n=250$ lower power values than under equal allocation are found.
\subsubsection{Results for CMDP-BP}

 \autoref{fig:CMDP-EP_plots} and~\autoref{tab:RR_CMDP_EP} provide results for the CMDP-BP procedure. \autoref{fig:CMDP-EP_plots}  shows that the type I error rate profiles and power differences do not differ much from those for CMDP-P, while the bottom row shows there are substantial patient benefit gains over CMDP-P. For~$\Iend=250$, the patient benefit goes to about~$95\%$ (the highest possible value for maximum randomized allocation rate~$p=0.95$). Comparing \autoref{fig:TNA_plots} and~\autoref{fig:CMDP-EP_plots}, CMDP-BP is not far off in terms 

\clearpage
\begin{table}[p]
\thispagestyle{empty} 
    \centering
      \caption{Highest power values for each response-adaptive procedure, the superscripts denote which exact Wald test(s) resulted in the highest power. The highest~(lowest) power in each row is overlined~(underlined) and indicated in bold font.}
    \footnotesize\label{tab:my_label}\begin{tabular}{lllllllll}
\toprule
 &  &  & NA & TNA & CMDP-P & CMDP-BP & BRAR & EA \\
$\Iend$ & $\theta_\text{C}$ & $\theta_{\text{D}}$ &  &  &  &  &  &  \\
\midrule
\multirow[t]{26}{*}{50} & \multirow[t]{5}{*}{0.01} & 0.10 & $8.93^{b}$ & $\uline{\textcolor{red}{\bf 8.88}}^{b}$ & $\overline{\textcolor{darkgreen}{\bf 22.57}}^{u}$ & $22.14^{u}$ & $15.87^{b}$ & $18.58^{u}$ \\
 &  & 0.20 & $51.34^{u}$ & $\uline{\textcolor{red}{\bf 51.33}}^{u}$ & $70.14^{u}$ & $\overline{\textcolor{darkgreen}{\bf 70.52}}^{u}$ & $54.96^{b}$ & $64.17^{u}$ \\
 &  & 0.30 & $85.43^{u}$ & $\uline{\textcolor{red}{\bf 85.4}}^{u}$ & $91.98^{u}$ & $\overline{\textcolor{darkgreen}{\bf 92.65}}^{u}$ & $86.3^{b}$ & $89.35^{u}$ \\
 &  & 0.40 & $97.02^{b}$ & $\uline{\textcolor{red}{\bf 97.01}}^{b}$ & $98.28^{b}$ & $\overline{\textcolor{darkgreen}{\bf 98.61}}^{u}$ & $97.34^{b}$ & $97.92^{u}$ \\
 &  & 0.50 & $99.72^{b}$ & $99.72^{b}$ & $99.7^{u}$ & $\overline{\textcolor{darkgreen}{\bf 99.84}}^{b}$ & $\uline{\textcolor{red}{\bf 99.68}}^{b}$ & $99.78^{u}$ \\
\cline{2-9}
 & \multirow[t]{5}{*}{0.10} & 0.25 & $21.47^{c}$ & $\uline{\textcolor{red}{\bf 21.19}}^{c}$ & $28.78^{u}$ & $\overline{\textcolor{darkgreen}{\bf 29.04}}^{u}$ & $22.59^{c}$ & $24.65^{u}$ \\
 &  & 0.40 & $67.58^{b}$ & $67.43^{b}$ & $72.65^{u}$ & $\overline{\textcolor{darkgreen}{\bf 72.73}}^{u}$ & $\uline{\textcolor{red}{\bf 64.02}}^{b}$ & $69.03^{u}$ \\
 &  & 0.50 & $88.8^{b}$ & $88.66^{b}$ & $\overline{\textcolor{darkgreen}{\bf 90.88}}^{u}$ & $90.87^{u}$ & $\uline{\textcolor{red}{\bf 85.65}}^{b}$ & $89.55^{u}$ \\
 &  & 0.60 & $97.64^{b}$ & $97.59^{b}$ & $\overline{\textcolor{darkgreen}{\bf 98.09}}^{u}$ & $98.08^{u}$ & $\uline{\textcolor{red}{\bf 96.26}}^{b}$ & $97.88^{u,b}$ \\
 &  & 0.70 & $99.74^{b}$ & $99.74^{b}$ & $\overline{\textcolor{darkgreen}{\bf 99.79}}^{u}$ & $\overline{\textcolor{darkgreen}{\bf 99.79}}^{u}$ & $\uline{\textcolor{red}{\bf 99.43}}^{b}$ & $99.76^{u,b}$ \\
\cline{2-9}
 & \multirow[t]{5}{*}{0.30} & 0.45 & $17.61^{c}$ & $17.86^{c}$ & $\overline{\textcolor{darkgreen}{\bf 19.28}}^{u}$ & $19.1^{u}$ & $\uline{\textcolor{red}{\bf 16.14}}^{c}$ & $18.2^{u}$ \\
 &  & 0.60 & $55.74^{b}$ & $56.21^{b}$ & $\overline{\textcolor{darkgreen}{\bf 57.13}}^{u}$ & $56.97^{u}$ & $\uline{\textcolor{red}{\bf 51.57}}^{b}$ & $53.5^{u,b}$ \\
 &  & 0.70 & $81.13^{b}$ & $81.42^{b}$ & $\overline{\textcolor{darkgreen}{\bf 82.11}}^{u}$ & $82.06^{u}$ & $\uline{\textcolor{red}{\bf 77.14}}^{b}$ & $79.3^{u,b}$ \\
 &  & 0.80 & $95.81^{b}$ & $95.84^{b}$ & $96.17^{u}$ & $\overline{\textcolor{darkgreen}{\bf 96.18}}^{u}$ & $\uline{\textcolor{red}{\bf 93.74}}^{b}$ & $95.39^{u,b}$ \\
 &  & 0.90 & $99.74^{b}$ & $99.74^{b}$ & $\overline{\textcolor{darkgreen}{\bf 99.79}}^{u}$ & $\overline{\textcolor{darkgreen}{\bf 99.79}}^{u}$ & $\uline{\textcolor{red}{\bf 99.31}}^{b}$ & $99.76^{u,b}$ \\
\cline{2-9}
 & \multirow[t]{5}{*}{0.50} & 0.65 & $17.18^{c}$ & $17.74^{c}$ & $18.56^{u}$ & $\overline{\textcolor{darkgreen}{\bf 18.58}}^{u}$ & $16.71^{c}$ & $\uline{\textcolor{red}{\bf 16.56}}^{u}$ \\
 &  & 0.70 & $29.03^{b}$ & $29.2^{b}$ & $30.15^{u}$ & $\overline{\textcolor{darkgreen}{\bf 30.21}}^{u}$ & $\uline{\textcolor{red}{\bf 27.22}}^{b}$ & $28.18^{u,b}$ \\
 &  & 0.80 & $59.78^{b}$ & $60.5^{b}$ & $61.99^{u}$ & $\overline{\textcolor{darkgreen}{\bf 62.05}}^{u}$ & $\uline{\textcolor{red}{\bf 56.81}}^{b}$ & $60.64^{u,b}$ \\
 &  & 0.90 & $88.8^{b}$ & $89.76^{b}$ & $90.87^{u}$ & $\overline{\textcolor{darkgreen}{\bf 90.88}}^{u}$ & $\uline{\textcolor{red}{\bf 86.49}}^{b}$ & $89.55^{u}$ \\
 &  & 0.95 & $96.98^{b}$ & $97.51^{b}$ & $\overline{\textcolor{darkgreen}{\bf 97.8}}^{u}$ & $97.79^{u}$ & $\uline{\textcolor{red}{\bf 95.87}}^{b}$ & $97.32^{u}$ \\
\cline{2-9}
 & \multirow[t]{4}{*}{0.70} & 0.80 & $11.59^{b}$ & $12.26^{b}$ & $\overline{\textcolor{darkgreen}{\bf 12.8}}^{u}$ & $\overline{\textcolor{darkgreen}{\bf 12.8}}^{u}$ & $11.91^{b}$ & $\uline{\textcolor{red}{\bf 11.17}}^{u}$ \\
 &  & 0.90 & $\uline{\textcolor{red}{\bf 37.05}}^{b}$ & $42.05^{b}$ & $\overline{\textcolor{darkgreen}{\bf 43.76}}^{u}$ & $43.75^{u}$ & $40.19^{b}$ & $39.05^{u}$ \\
 &  & 0.95 & $\uline{\textcolor{red}{\bf 59.58}}^{b}$ & $68.02^{b}$ & $69.61^{b}$ & $\overline{\textcolor{darkgreen}{\bf 69.63}}^{b}$ & $65.73^{b}$ & $64.19^{u}$ \\
 &  & 0.99 & $\uline{\textcolor{red}{\bf 85.43}}^{u}$ & $91.01^{b}$ & $\overline{\textcolor{darkgreen}{\bf 92.01}}^{u}$ & $92.0^{u}$ & $89.09^{b}$ & $89.35^{u}$ \\
\cline{2-9}
 & \multirow[t]{2}{*}{0.90} & 0.95 & $\uline{\textcolor{red}{\bf 4.34}}^{b}$ & $8.5^{b}$ & $8.66^{u}$ & $8.67^{u}$ & $\overline{\textcolor{darkgreen}{\bf 9.05}}^{b}$ & $7.21^{u}$ \\
 &  & 0.99 & $\uline{\textcolor{red}{\bf 8.93}}^{b}$ & $21.63^{b}$ & $22.57^{u}$ & $22.57^{u}$ & $\overline{\textcolor{darkgreen}{\bf 23.79}}^{b}$ & $18.58^{u}$ \\
\cline{1-9} \cline{2-9}
\multirow[t]{26}{*}{250} & \multirow[t]{5}{*}{0.01} & 0.05 & $36.04^{u}$ & $\uline{\textcolor{red}{\bf 36.03}}^{u}$ & $42.23^{c,b}$ & $42.45^{c,b}$ & $39.71^{b}$ & $\overline{\textcolor{darkgreen}{\bf 45.48}}^{u}$ \\
 &  & 0.10 & $89.27^{b}$ & $\uline{\textcolor{red}{\bf 89.25}}^{b}$ & $92.64^{c,b}$ & $92.72^{c,b}$ & $89.55^{b}$ & $\overline{\textcolor{darkgreen}{\bf 92.97}}^{u}$ \\
 &  & 0.13 & $97.98^{c}$ & $97.97^{c}$ & $98.79^{c}$ & $98.82^{c}$ & $\uline{\textcolor{red}{\bf 97.78}}^{c}$ & $\overline{\textcolor{darkgreen}{\bf 98.83}}^{u}$ \\
 &  & 0.17 & $99.88^{c}$ & $99.88^{c}$ & $\overline{\textcolor{darkgreen}{\bf 99.93}}^{c}$ & $\overline{\textcolor{darkgreen}{\bf 99.93}}^{c}$ & $\uline{\textcolor{red}{\bf 99.83}}^{c}$ & $\overline{\textcolor{darkgreen}{\bf 99.93}}^{u}$ \\
 &  & 0.20 & $\overline{\textcolor{darkgreen}{\bf 99.99}}^{c,b}$ & $\overline{\textcolor{darkgreen}{\bf 99.99}}^{c,b}$ & $\overline{\textcolor{darkgreen}{\bf 99.99}}^{c,u,b}$ & $\overline{\textcolor{darkgreen}{\bf 99.99}}^{c,u,b}$ & $\uline{\textcolor{red}{\bf 99.98}}^{c,b}$ & $\overline{\textcolor{darkgreen}{\bf 99.99}}^{u,b}$ \\
\cline{2-9}
 & \multirow[t]{5}{*}{0.10} & 0.20 & $58.97^{b}$ & $58.89^{b}$ & $61.3^{b}$ & $\overline{\textcolor{darkgreen}{\bf 61.33}}^{b}$ & $\uline{\textcolor{red}{\bf 53.81}}^{b}$ & $60.13^{u}$ \\
 &  & 0.25 & $88.1^{c}$ & $88.08^{c}$ & $89.38^{u}$ & $\overline{\textcolor{darkgreen}{\bf 89.42}}^{u}$ & $\uline{\textcolor{red}{\bf 83.93}}^{c}$ & $88.73^{u}$ \\
 &  & 0.30 & $98.15^{b}$ & $98.14^{b}$ & $\overline{\textcolor{darkgreen}{\bf 98.39}}^{b}$ & $\overline{\textcolor{darkgreen}{\bf 98.39}}^{u,b}$ & $\uline{\textcolor{red}{\bf 96.78}}^{b}$ & $98.26^{u}$ \\
 &  & 0.35 & $99.84^{c}$ & $99.84^{c}$ & $\overline{\textcolor{darkgreen}{\bf 99.86}}^{c,u}$ & $\overline{\textcolor{darkgreen}{\bf 99.86}}^{c,u}$ & $\uline{\textcolor{red}{\bf 99.61}}^{c}$ & $99.85^{u}$ \\
 &  & 0.40 & $\overline{\textcolor{darkgreen}{\bf 99.99}}^{c,b}$ & $\overline{\textcolor{darkgreen}{\bf 99.99}}^{c,b}$ & $\overline{\textcolor{darkgreen}{\bf 99.99}}^{c,u,b}$ & $\overline{\textcolor{darkgreen}{\bf 99.99}}^{c,u,b}$ & $\uline{\textcolor{red}{\bf 99.97}}^{c,b}$ & $\overline{\textcolor{darkgreen}{\bf 99.99}}^{c,u,b}$ \\
\cline{2-9}
 & \multirow[t]{5}{*}{0.30} & 0.40 & $37.85^{b}$ & $37.76^{b}$ & $38.07^{u}$ & $\overline{\textcolor{darkgreen}{\bf 38.14}}^{u}$ & $\uline{\textcolor{red}{\bf 33.82}}^{b}$ & $37.78^{u}$ \\
 &  & 0.50 & $90.05^{b}$ & $90.03^{b}$ & $90.24^{u}$ & $\overline{\textcolor{darkgreen}{\bf 90.27}}^{u}$ & $\uline{\textcolor{red}{\bf 86.0}}^{b}$ & $89.7^{u,b}$ \\
 &  & 0.55 & $98.13^{c}$ & $98.12^{c}$ & $\overline{\textcolor{darkgreen}{\bf 98.19}}^{u}$ & $\overline{\textcolor{darkgreen}{\bf 98.19}}^{u}$ & $\uline{\textcolor{red}{\bf 96.54}}^{c}$ & $98.08^{u}$ \\
 &  & 0.60 & $99.81^{c,b}$ & $99.81^{c,b}$ & $\overline{\textcolor{darkgreen}{\bf 99.82}}^{u}$ & $\overline{\textcolor{darkgreen}{\bf 99.82}}^{u}$ & $\uline{\textcolor{red}{\bf 99.49}}^{b}$ & $99.81^{u,b}$ \\
 &  & 0.65 & $\overline{\textcolor{darkgreen}{\bf 99.99}}^{c}$ & $\overline{\textcolor{darkgreen}{\bf 99.99}}^{c}$ & $\overline{\textcolor{darkgreen}{\bf 99.99}}^{c,u}$ & $\overline{\textcolor{darkgreen}{\bf 99.99}}^{c,u}$ & $\uline{\textcolor{red}{\bf 99.95}}^{c}$ & $\overline{\textcolor{darkgreen}{\bf 99.99}}^{u}$ \\
\cline{2-9}
 & \multirow[t]{5}{*}{0.50} & 0.60 & $35.4^{b}$ & $\overline{\textcolor{darkgreen}{\bf 35.48}}^{b}$ & $\overline{\textcolor{darkgreen}{\bf 35.48}}^{u}$ & $35.45^{u}$ & $\uline{\textcolor{red}{\bf 32.05}}^{b}$ & $35.17^{u,b}$ \\
 &  & 0.70 & $90.05^{b}$ & $90.1^{b}$ & $\overline{\textcolor{darkgreen}{\bf 90.23}}^{u}$ & $90.17^{u}$ & $\uline{\textcolor{red}{\bf 86.24}}^{b}$ & $89.7^{u,b}$ \\
 &  & 0.75 & $98.51^{c}$ & $98.52^{c}$ & $\overline{\textcolor{darkgreen}{\bf 98.57}}^{u}$ & $98.56^{u}$ & $\uline{\textcolor{red}{\bf 97.16}}^{c}$ & $98.48^{u}$ \\
 &  & 0.80 & $\overline{\textcolor{darkgreen}{\bf 99.92}}^{b}$ & $\overline{\textcolor{darkgreen}{\bf 99.92}}^{b}$ & $\overline{\textcolor{darkgreen}{\bf 99.92}}^{c,u,b}$ & $\overline{\textcolor{darkgreen}{\bf 99.92}}^{c,u,b}$ & $\uline{\textcolor{red}{\bf 99.72}}^{b}$ & $\overline{\textcolor{darkgreen}{\bf 99.92}}^{u,b}$ \\
 &  & 0.85 & $\overline{\textcolor{darkgreen}{\bf 100.0}}^{c}$ & $\overline{\textcolor{darkgreen}{\bf 100.0}}^{c}$ & $\overline{\textcolor{darkgreen}{\bf 100.0}}^{c,u}$ & $\overline{\textcolor{darkgreen}{\bf 100.0}}^{c,u}$ & $\uline{\textcolor{red}{\bf 99.99}}^{c}$ & $\overline{\textcolor{darkgreen}{\bf 100.0}}^{c,u}$ \\
\cline{2-9}
 & \multirow[t]{4}{*}{0.70} & 0.80 & $44.15^{b}$ & $44.62^{b}$ & $\overline{\textcolor{darkgreen}{\bf 45.06}}^{b}$ & $45.05^{b}$ & $\uline{\textcolor{red}{\bf 41.1}}^{b}$ & $44.64^{u}$ \\
 &  & 0.85 & $81.16^{c}$ & $81.65^{c}$ & $\overline{\textcolor{darkgreen}{\bf 82.25}}^{c}$ & $82.21^{c}$ & $\uline{\textcolor{red}{\bf 77.25}}^{c}$ & $81.69^{u}$ \\
 &  & 0.90 & $98.15^{b}$ & $98.27^{b}$ & $98.39^{c,b}$ & $\overline{\textcolor{darkgreen}{\bf 98.4}}^{b}$ & $\uline{\textcolor{red}{\bf 96.81}}^{b}$ & $98.26^{u}$ \\
 &  & 0.95 & $\overline{\textcolor{darkgreen}{\bf 99.99}}^{c,b}$ & $\overline{\textcolor{darkgreen}{\bf 99.99}}^{c,b}$ & $\overline{\textcolor{darkgreen}{\bf 99.99}}^{c,u,b}$ & $\overline{\textcolor{darkgreen}{\bf 99.99}}^{c,u,b}$ & $\uline{\textcolor{red}{\bf 99.93}}^{b}$ & $\overline{\textcolor{darkgreen}{\bf 99.99}}^{u,b}$ \\
\cline{2-9}
 & \multirow[t]{2}{*}{0.90} & 0.95 & $\uline{\textcolor{red}{\bf 28.65}}^{b}$ & $\overline{\textcolor{darkgreen}{\bf 32.05}}^{b}$ & $31.67^{b}$ & $31.76^{b}$ & $30.79^{b}$ & $31.02^{u}$ \\
 &  & 0.99 & $\uline{\textcolor{red}{\bf 89.27}}^{b}$ & $\overline{\textcolor{darkgreen}{\bf 93.4}}^{b}$ & $92.64^{c,b}$ & $92.74^{b}$ & $91.21^{b}$ & $92.97^{u}$ \\
\cline{1-9} \cline{2-9}
\bottomrule
\end{tabular}

\\{$^{c}$ conditional exact test, $^{u}$ unconditional exact  test, $^{b}$ Boschloo's test.\\
NA: doubly adaptive biased coin design Neyman allocation, TNA: tempered doubly adaptive biased coin design Neyman allocation, CMDP-P: constrained Markov decision process optimizing average power, CMDP-BP: protected constrained Markov decision process optimizing average power, Bayesian RAR: Bayesian response-adaptive randomization, EA: equal allocation}\label{summary_table}
\end{table}
\clearpage

of patient benefit from Bayesian RAR while showing power comparable to equal allocation for~$n=250$. \autoref{tab:RR_CMDP_EP} shows that CMDP-BP shows no substantial loss in power in comparison to the CMDP-P procedure.

\subsection{Summary table}\label{sect:summary_table}
\autoref{summary_table} summarizes the tables in the appendix and displays the highest power attained for each RAR procedure across exact tests. The respective test under which this power is attained is denoted using a superscript. 

 The table shows that the highest power values are attained for the CMDP-BP procedure, with a highest gain over equal allocation of around~$6.41\%$ for the parameter configuration~$(0.01,0.2)$ for~$\Iend=50$ under the unconditional exact test, and a highest gain around~0.57\%  for the parameter configuration~$(0.3,0.5)$ for~$\Iend=250$ under the unconditional exact test, also noting the loss of~-3.03\%  for the parameter configuration~$(0.01,0.05)$. We note that, next to power gains, this procedure was also able to increase patient benefit over equal allocation, sometimes reaching patient benefit values that are also seen under the Bayesian RAR design, only targeting high patient benefit. For this reason, one might prefer the CMDP-BP design over equal allocation. 
\FloatBarrier
\subsection{Results for randomization-based Wald test}
\autoref{tab:randomization_test} shows simulation-based unconditional type I error rates and power values for the randomization-based test~\citep{SIMON2011767} based on the Wald statistic. 
A substantially different relative performance of the designs as in~\autoref{summary_table} is seen. The CMDP-based procedures often provide lowest power, while the tempered  DBCD Neyman allocation procedure often provides highest power, followed by equal allocation and DBCD Neyman allocation, respectively. An explanation could be that the CMDP procedures have allocation probabilities in the set~$\{0.05, 0.5, 0.95\}$. Hence, possible allocation paths for CMDP-based procedures might concentrate probabilistically, leading to low power under the randomization test. Instead, the other RAR procedures have a lower degree of probabilistic clustering in the allocation paths as they have a wider range of potential allocation probabilities, or often have allocation probabilities close to 50\%.  Comparing power values under the randomization-based Wald test with the maximum power values in~\autoref{summary_table}, the randomization-based test often gives lower power for the CMDP-based and Bayesian RAR procedures, 
For the DBCD Neyman allocation procedures and equal allocation, the randomization-based Wald test mainly gives lower power for small or large control success rates under~$n = 50$.

\FloatBarrier
\section{Conclusion and discussion}\label{sect:conclusion_discussion}

We have investigated solutions for type I error rate inflation under power-maximizing {\it response-adaptive randomization}~(RAR) procedures, which is a problem recently identified in~\citet{pin2025revisitingoptimalallocationsbinary} for RAR procedures targeting maximum statistical power. This problem is of particular concern for the use of RAR designs in confirmatory trial settings. In this paper, we focused on two-arm clinical trials with binary outcomes.

As a first step, exact tests such as the conditional exact test~(based on total successes) and unconditional exact test, as well as a novel generalized version of Boschloo's test, were considered in combination with the {\it doubly adaptive biased coin design}~(DBCD) Neyman allocation. This idea alone did not lead to a satisfying result, as the unconditional exact test, which out of all considered exact tests almost uniformly provided highest power under equal allocation, showed a very conservative type I error rate profile under DBCD Neyman allocation, and low power as a result. This is mainly due to the type I error rate profile for the asymptotic Wald test under DBCD Neyman allocation being very peaked. An improvement was found in terms of power for the conditional exact test, but this is somewhat of an empty victory for DBCD Neyman allocation as the conditional exact test for equal allocation coincides with Fisher's exact test, which is known to be very conservative. For instance, power improvements for the conditional exact test were even~(often) seen for the Bayesian RAR procedure solely targeting high expected treatment outcomes. 
For DBCD Neyman allocation and Bayesian RAR procedures, our newly developed generalized version of Boschloo's test often showed highest power.

A second considered step is to directly incorporate the type I error rate constraints when maximizing power by the RAR procedure through the {\it constrained Markov decision process}~(CMDP) framework considered in~\citet{baas2024CMDP}. 
This approach yields uniformly higher power for the unconditional exact Wald test than under the equal allocation design for a trial size of~50 participants, whereas for a trial size of~250 participants, the power gains are smaller and less uniform. 
A benefit of the CMDP procedure is that it can be used for any test, not only the Wald test~\citep[e.g.,][considered Fisher's exact test instead]{baas2024CMDP}. Another benefit is that the extension to multiple treatment arms might be more readily available.

We also considered versions of DBCD Neyman allocation and CMDP
procedures that aimed to protect the expected proportion of participants allocated to the superior arm at a minimum of~50\%. As expected, these procedures led to a uniform increase in expected treatment outcomes compared to their non-protected counterparts, but curiously, they also showed comparable or higher power for most considered parameters. %

The evaluation in this paper shows that CMDP procedures show higher power than equal allocation for the unconditional exact Wald test under a trial size of~50 participants. For a trial size of~250 participants, the power gains under CMDP were smaller; however, patient benefit values under the CMDP procedure, maximizing power and constraining patient benefit, are often higher than under equal allocation. Hence, we suggest using this design in small to moderate trial sizes when the model of~\autoref{subsect:model} can reasonably be assumed. 
In case the model in~\autoref{subsect:model} cannot be assumed, one might use a randomization-based test, and our results show that the CMDP procedures showed lowest power under this test, while DBCD Neyman allocation and equal allocation showed higher power. As tempered DBCD Neyman allocation showed high power values as well as uniform gains in patient benefit over DBCD Neyman allocation and equal allocation, we suggest using tempered DBCD Neyman allocation when using a randomization-based test.

While we aimed for a general comparison of RAR procedures in this paper, a number of specific choices were made, listed at the start of~\autoref{sect:numerical_results}, and these choices could warrant further investigation in specific trial settings. For instance,~\citet{pin2025revisitingoptimalallocationsbinary} showed that the type I error rate inflation under DBCD Neyman allocation is exacerbated when decreasing the burn-in value. The best choice of the burn-in length may depend on~(a priori beliefs surrounding) the specific trial parameters at hand, and could possibly be determined using the evaluation procedure described in~\citet{Tang2025}. One further thing of note is that we mainly focused on the Wald test for simple differences in this paper. We note that our evaluation can also be directly applied to other treatment effect formulations~(e.g., the odds ratio) by changing the DBCD Neyman allocation procedure as explained in~\citet{Pin2024}, and by reformulating the CMDP procedures using the Wald test for the treatment effect formulation of interest.  Our conclusions likely do not carry over to multi-arm trials, e.g., in multi-arm settings with fixed control allocation proportions, the type I error rate inflation is likely not as substantial as in the two-arm setting, leaving more room for power gains.

As Algorithm~1 in~\citet{baas2024CMDP} does not solve the CMDP problem exactly~(but often yields a procedure having the desired properties) and can have a relatively long computation time, future research could focus on alternative and more computationally tractable solution methods for the CMDP problem. Another potential future research topic is obtaining exact instead of~(Bayesian) average guarantees for CMDP procedures, which could help provide uniform power improvements, even for higher trial sizes.
Future research could also investigate improving the performance of CMDP procedures for randomization-based tests.
An additional topic of future research is modeling the test outcomes as decisions made by the CMDP problem, in an approach similar to the one in~\citet{baas2025exactstatisticaltestsusing},  thereby removing the dependence on a specific test known from the literature and improving the average power gains of the CMDP-P procedure.  
Lastly, the curves showing the differences in expected proportion of participants allocated to the superior arm for the CMDP procedures did not show an interpretable pattern, and future research could investigate the asymptotic allocation proportions for CMDP procedures. To conclude, we found that the protected CMDP procedure often outperforms equal allocation under the unconditional exact test, while also showing benefits in terms of expected treatment outcomes, especially for small trial sizes~(e.g., $50$ participants). For larger trial sizes~(around~$250$), the power gains are smaller, but several promising future research directions could lead to CMDP procedures outperforming equal allocation in such settings.

\bibliographystyle{elsarticle-harv} 
\bibliography{main}

\begin{appendix}

\section{Allocation proportion plots for DBCD Neyman allocation}
\begin{figure}[h!]
    \centering
    \includegraphics[width=\linewidth]{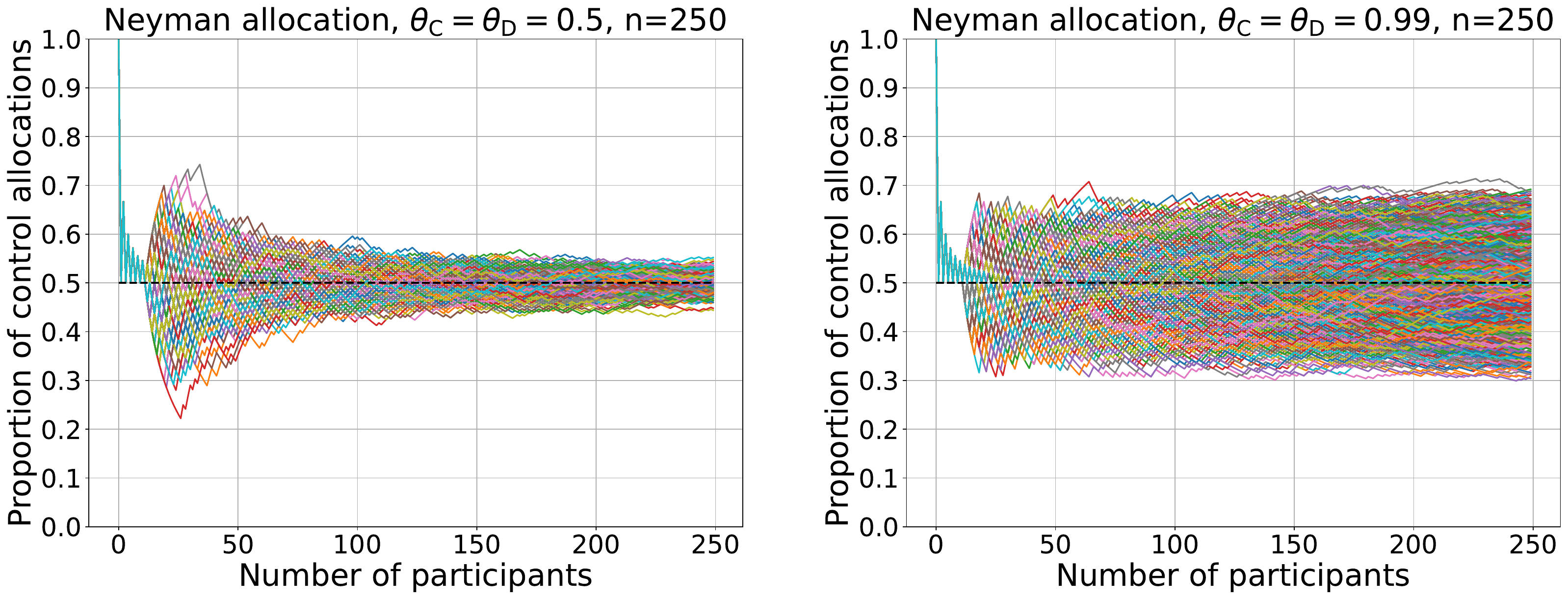}
    \caption{Simulation paths of the running proportion of participants allocated to the control treatment under doubly adaptive biased coin design Neyman allocation~(burn-in 6 participants per arm) in case~$\theta_{\C}=\theta_{\D}=0.5$~(left) and~$\theta_{\C}=\theta_{\D}=0.99$~(right). The initial period of 12 participants where the allocation proportion fluctuates between~$0.5$ and a value higher than~$0.5$ corresponds to the burn-in period~(with even/uneven allocation to control/developmental treatment). }
    \label{fig:simulation_paths}
\end{figure}

\FloatBarrier

\begin{table}[]
\vspace{-5mm}
\section{Rejection rate tables}
  \centering
  \caption{Rejection rates~$(\%)$ for  the doubly adaptive biased coin design Neyman allocation procedure~(burn-in equal to~6 participants per arm). The bold values indicate highest power or type I error inflation in the row, depending on the parameter configuration considered. The underlined values indicate higher power than under that test for the equal allocation procedure. The twice underlined values indicate higher power than the test with highest power for the equal allocation procedure. Both upper and lower significance levels were set to~$2.5\%$.}\label{tab_RR_neyman}
\scriptsize \begin{tabular}{lllrrrrrrrr}
\toprule
 &  &  & \multicolumn{4}{c}{DBCD Neyman allocation} & \multicolumn{4}{c}{Equal allocation} \\
 &  &  & CX-S & UX & Asymp. & GB & CX-S & UX & Asymp. & GB \\
$\Iend$ & $\theta_\text{C}$ & $\theta_{\text{D}}$ &  &  &  &  &  &  &  &  \\
\midrule
\multirow[t]{32}{*}{50} & \multirow[t]{6}{*}{0.01} & 0.01 & 0.74 & 0.00 & 0.05 & 1.80 & 0.00 & 0.02 & 0.02 & 0.00 \\
 &  & 0.10 & \uline{8.05} & 6.19 & 32.66 & \uline{\textcolor{darkgreen}{\bf8.93}} & 2.79 & \textcolor{darkgreen}{\bf18.58} & 19.05 & 7.81 \\
 &  & 0.20 & \uline{38.98} & \textcolor{darkgreen}{\bf51.34} & 80.32 & 42.14 & 34.25 & \textcolor{darkgreen}{\bf64.17} & 67.38 & 49.64 \\
 &  & 0.30 & \uline{79.85} & \textcolor{darkgreen}{\bf85.43} & 95.60 & 82.05 & 76.47 & \textcolor{darkgreen}{\bf89.35} & 92.26 & 84.89 \\
 &  & 0.40 & \uline{96.63} & 96.51 & 99.35 & \textcolor{darkgreen}{\bf97.02} & 95.50 & \textcolor{darkgreen}{\bf97.92} & 98.82 & 97.37 \\
 &  & 0.50 & \uline{99.68} & 99.47 & 99.95 & \textcolor{darkgreen}{\bf99.72} & 99.53 & \textcolor{darkgreen}{\bf99.78} & 99.89 & 99.75 \\
\cline{2-11}
 & \multirow[t]{6}{*}{0.10} & 0.10 & 2.51 & 2.63 & \textcolor{red}{\bf 13.71} & 2.75 & 0.88 & 3.90 & 4.87 & 1.91 \\
 &  & 0.25 & \uline{21.47} & 21.22 & 38.97 & \textcolor{darkgreen}{\bf22.73} & 18.04 & \textcolor{darkgreen}{\bf24.65} & 30.66 & 23.79 \\
 &  & 0.40 & \uline{65.79} & 56.03 & 78.24 & \textcolor{darkgreen}{\bf67.58} & 59.05 & \textcolor{darkgreen}{\bf69.03} & 75.60 & 68.98 \\
 &  & 0.50 & \uline{87.88} & 79.51 & 93.18 & \textcolor{darkgreen}{\bf88.80} & 83.26 & \textcolor{darkgreen}{\bf89.55} & 92.49 & 89.54 \\
 &  & 0.60 & \uline{97.41} & 93.77 & 98.68 & \textcolor{darkgreen}{\bf97.64} & 95.95 & \textcolor{darkgreen}{\bf97.88} & 98.52 & \textcolor{darkgreen}{\bf97.88} \\
 &  & 0.70 & \uline{99.71} & 98.98 & 99.87 & \textcolor{darkgreen}{\bf99.74} & 99.53 & \textcolor{darkgreen}{\bf99.76} & 99.85 & \textcolor{darkgreen}{\bf99.76} \\
\cline{2-11}
 & \multirow[t]{6}{*}{0.30} & 0.30 & 4.51 & 2.58 & \textcolor{red}{\bf 8.49} & 4.96 & 2.29 & 4.48 & \textcolor{red}{\bf 6.42} & 4.48 \\
 &  & 0.45 & \uline{17.61} & 9.53 & 24.25 & \uuline{\textcolor{darkgreen}{\bf18.52}} & 13.46 & \textcolor{darkgreen}{\bf18.20} & 22.15 & \textcolor{darkgreen}{\bf18.20} \\
 &  & 0.60 & \uuline{54.96} & 39.39 & 63.64 & \uuline{\textcolor{darkgreen}{\bf55.74}} & 50.59 & \textcolor{darkgreen}{\bf53.50} & 62.36 & \textcolor{darkgreen}{\bf53.50} \\
 &  & 0.70 & \uuline{80.78} & 68.49 & 86.53 & \uuline{\textcolor{darkgreen}{\bf81.13}} & 78.22 & \textcolor{darkgreen}{\bf79.30} & 85.96 & \textcolor{darkgreen}{\bf79.30} \\
 &  & 0.80 & \uuline{95.62} & 90.32 & 97.40 & \uuline{\textcolor{darkgreen}{\bf95.81}} & 94.61 & \textcolor{darkgreen}{\bf95.39} & 97.25 & \textcolor{darkgreen}{\bf95.39} \\
 &  & 0.90 & \uline{99.71} & 98.98 & 99.87 & \textcolor{darkgreen}{\bf99.74} & 99.53 & \textcolor{darkgreen}{\bf99.76} & 99.85 & \textcolor{darkgreen}{\bf99.76} \\
\cline{2-11}
 & \multirow[t]{6}{*}{0.50} & 0.50 & 4.49 & 1.70 & \textcolor{red}{\bf 7.21} & 4.65 & 3.28 & 3.73 & \textcolor{red}{\bf 6.51} & 3.73 \\
 &  & 0.65 & \uuline{17.18} & 9.05 & 23.34 & \uuline{\textcolor{darkgreen}{\bf17.82}} & 13.92 & \textcolor{darkgreen}{\bf16.56} & 21.82 & \textcolor{darkgreen}{\bf16.56} \\
 &  & 0.70 & \uline{28.00} & 16.47 & 35.96 & \uuline{\textcolor{darkgreen}{\bf29.03}} & 23.25 & \textcolor{darkgreen}{\bf28.18} & 33.97 & \textcolor{darkgreen}{\bf28.18} \\
 &  & 0.80 & \uline{58.31} & 42.93 & 67.67 & \textcolor{darkgreen}{\bf59.78} & 51.04 & \textcolor{darkgreen}{\bf60.64} & 65.60 & \textcolor{darkgreen}{\bf60.64} \\
 &  & 0.90 & \uline{87.88} & 79.51 & 93.18 & \textcolor{darkgreen}{\bf88.80} & 83.26 & \textcolor{darkgreen}{\bf89.55} & 92.49 & 89.54 \\
 &  & 0.95 & \uline{96.59} & 93.74 & 98.67 & \textcolor{darkgreen}{\bf96.98} & 95.15 & \textcolor{darkgreen}{\bf97.32} & 98.35 & 97.31 \\
\cline{2-11}
 & \multirow[t]{5}{*}{0.70} & 0.70 & 4.51 & 2.58 & \textcolor{red}{\bf 8.49} & 4.96 & 2.29 & 4.48 & \textcolor{red}{\bf 6.42} & 4.48 \\
 &  & 0.80 & \uline{10.76} & 7.76 & 18.86 & \uuline{\textcolor{darkgreen}{\bf11.59}} & 6.87 & \textcolor{darkgreen}{\bf11.17} & 15.26 & 11.15 \\
 &  & 0.90 & \uline{35.46} & 31.26 & 52.92 & \textcolor{darkgreen}{\bf37.05} & 30.49 & \textcolor{darkgreen}{\bf39.05} & 46.29 & 38.64 \\
 &  & 0.95 & \uline{57.82} & 56.60 & 77.34 & \textcolor{darkgreen}{\bf59.58} & 54.92 & \textcolor{darkgreen}{\bf64.19} & 71.05 & 62.60 \\
 &  & 0.99 & \uline{79.85} & \textcolor{darkgreen}{\bf85.43} & 95.60 & 82.05 & 76.47 & \textcolor{darkgreen}{\bf89.35} & 92.26 & 84.89 \\
\cline{2-11}
 & \multirow[t]{3}{*}{0.90} & 0.90 & 2.51 & 2.63 & \textcolor{red}{\bf 13.71} & 2.75 & 0.88 & 3.90 & 4.87 & 1.91 \\
 &  & 0.95 & \uline{3.90} & 3.21 & 17.77 & \uline{\textcolor{darkgreen}{\bf4.34}} & 1.30 & \textcolor{darkgreen}{\bf7.21} & 8.10 & 3.17 \\
 &  & 0.99 & \uline{8.05} & 6.19 & 32.66 & \uline{\textcolor{darkgreen}{\bf8.93}} & 2.79 & \textcolor{darkgreen}{\bf18.58} & 19.05 & 7.81 \\
\cline{1-11} \cline{2-11}
\multirow[t]{32}{*}{250} & \multirow[t]{6}{*}{0.01} & 0.01 & 4.13 & 0.10 & \textcolor{red}{\bf 6.40} & 4.13 & 0.10 & 2.15 & 2.15 & 0.50 \\
 &  & 0.05 & \uline{34.40} & \textcolor{darkgreen}{\bf36.04} & 59.66 & 34.52 & 29.83 & \textcolor{darkgreen}{\bf45.48} & 45.49 & 36.25 \\
 &  & 0.10 & \uline{89.02} & 83.57 & 95.25 & \textcolor{darkgreen}{\bf89.27} & 87.46 & \textcolor{darkgreen}{\bf92.97} & 93.05 & 90.49 \\
 &  & 0.13 & \uline{97.98} & 95.55 & 99.23 & \textcolor{darkgreen}{\bf98.03} & 97.42 & \textcolor{darkgreen}{\bf98.83} & 98.88 & 98.25 \\
 &  & 0.17 & \uline{\textcolor{darkgreen}{\bf99.88}} & 99.54 & 99.96 & \textcolor{darkgreen}{\bf99.88} & 99.82 & \textcolor{darkgreen}{\bf99.93} & 99.94 & 99.88 \\
 &  & 0.20 & \uline{\textcolor{darkgreen}{\bf99.99}} & 99.94 & 100.00 & \textcolor{darkgreen}{\bf99.99} & 99.98 & \textcolor{darkgreen}{\bf99.99} & 100.00 & \textcolor{darkgreen}{\bf99.99} \\
\cline{2-11}
 & \multirow[t]{6}{*}{0.10} & 0.10 & 4.90 & 1.94 & \textcolor{red}{\bf 7.30} & 4.96 & 2.93 & 4.61 & \textcolor{red}{\bf 5.43} & 3.92 \\
 &  & 0.20 & \uline{58.85} & 36.04 & 63.14 & \uline{\textcolor{darkgreen}{\bf58.97}} & 53.82 & \textcolor{darkgreen}{\bf60.13} & 61.93 & 58.40 \\
 &  & 0.25 & \uline{88.10} & 71.32 & 89.90 & \uline{\textcolor{darkgreen}{\bf88.16}} & 85.48 & \textcolor{darkgreen}{\bf88.73} & 89.46 & 88.00 \\
 &  & 0.30 & \uline{98.14} & 92.52 & 98.49 & \uline{\textcolor{darkgreen}{\bf98.15}} & 97.50 & \textcolor{darkgreen}{\bf98.26} & 98.40 & 98.12 \\
 &  & 0.35 & \uline{\textcolor{darkgreen}{\bf99.84}} & 98.90 & 99.88 & \textcolor{darkgreen}{\bf99.84} & 99.76 & \textcolor{darkgreen}{\bf99.85} & 99.87 & 99.84 \\
 &  & 0.40 & \textcolor{darkgreen}{\bf99.99} & 99.91 & 99.99 & \textcolor{darkgreen}{\bf99.99} & \textcolor{darkgreen}{\bf99.99} & \textcolor{darkgreen}{\bf99.99} & 99.99 & \textcolor{darkgreen}{\bf99.99} \\
\cline{2-11}
 & \multirow[t]{6}{*}{0.30} & 0.30 & 4.90 & 0.88 & \textcolor{red}{\bf 5.46} & 4.94 & 3.62 & 4.94 & \textcolor{red}{\bf 5.15} & 4.80 \\
 &  & 0.40 & \uline{37.68} & 16.25 & 39.33 & \uuline{\textcolor{darkgreen}{\bf37.85}} & 33.84 & \textcolor{darkgreen}{\bf37.78} & 39.29 & 37.74 \\
 &  & 0.50 & \uuline{89.99} & 72.56 & 90.70 & \uuline{\textcolor{darkgreen}{\bf90.05}} & 87.62 & \textcolor{darkgreen}{\bf89.70} & 90.42 & \textcolor{darkgreen}{\bf89.70} \\
 &  & 0.55 & \uuline{98.13} & 91.83 & 98.29 & \uuline{\textcolor{darkgreen}{\bf98.14}} & 97.43 & \textcolor{darkgreen}{\bf98.08} & 98.16 & \textcolor{darkgreen}{\bf98.08} \\
 &  & 0.60 & \uline{\textcolor{darkgreen}{\bf99.81}} & 98.61 & 99.83 & \textcolor{darkgreen}{\bf99.81} & 99.72 & \textcolor{darkgreen}{\bf99.81} & 99.81 & \textcolor{darkgreen}{\bf99.81} \\
 &  & 0.65 & \uline{\textcolor{darkgreen}{\bf99.99}} & 99.87 & 99.99 & \textcolor{darkgreen}{\bf99.99} & 99.98 & \textcolor{darkgreen}{\bf99.99} & 99.99 & \textcolor{darkgreen}{\bf99.99} \\
\cline{2-11}
 & \multirow[t]{6}{*}{0.50} & 0.50 & 4.95 & 0.79 & \textcolor{red}{\bf 5.16} & 4.99 & 3.67 & 4.97 & 4.97 & 4.97 \\
 &  & 0.60 & \uuline{35.30} & 14.38 & 36.38 & \uuline{\textcolor{darkgreen}{\bf35.40}} & 30.59 & \textcolor{darkgreen}{\bf35.17} & 35.31 & \textcolor{darkgreen}{\bf35.17} \\
 &  & 0.70 & \uuline{89.99} & 72.56 & 90.70 & \uuline{\textcolor{darkgreen}{\bf90.05}} & 87.62 & \textcolor{darkgreen}{\bf89.70} & 90.42 & \textcolor{darkgreen}{\bf89.70} \\
 &  & 0.75 & \uuline{98.51} & 93.20 & 98.67 & \uuline{\textcolor{darkgreen}{\bf98.53}} & 98.06 & \textcolor{darkgreen}{\bf98.48} & 98.67 & \textcolor{darkgreen}{\bf98.48} \\
 &  & 0.80 & \uline{99.91} & 99.29 & 99.93 & \textcolor{darkgreen}{\bf99.92} & 99.88 & \textcolor{darkgreen}{\bf99.92} & 99.93 & \textcolor{darkgreen}{\bf99.92} \\
 &  & 0.85 & \textcolor{darkgreen}{\bf100.00} & 99.98 & 100.00 & \textcolor{darkgreen}{\bf100.00} & \textcolor{darkgreen}{\bf100.00} & \textcolor{darkgreen}{\bf100.00} & 100.00 & \textcolor{darkgreen}{\bf100.00} \\
\cline{2-11}
 & \multirow[t]{5}{*}{0.70} & 0.70 & 4.90 & 0.88 & \textcolor{red}{\bf 5.46} & 4.94 & 3.62 & 4.94 & \textcolor{red}{\bf 5.15} & 4.80 \\
 &  & 0.80 & \uline{43.99} & 21.19 & 46.39 & \uline{\textcolor{darkgreen}{\bf44.15}} & 39.24 & \textcolor{darkgreen}{\bf44.64} & 45.44 & 43.87 \\
 &  & 0.85 & \uline{81.16} & 59.31 & 83.03 & \uline{\textcolor{darkgreen}{\bf81.23}} & 77.67 & \textcolor{darkgreen}{\bf81.69} & 82.38 & 81.02 \\
 &  & 0.90 & \uline{98.14} & 92.52 & 98.49 & \uline{\textcolor{darkgreen}{\bf98.15}} & 97.50 & \textcolor{darkgreen}{\bf98.26} & 98.40 & 98.12 \\
 &  & 0.95 & \uline{\textcolor{darkgreen}{\bf99.99}} & 99.87 & 99.99 & \textcolor{darkgreen}{\bf99.99} & 99.98 & \textcolor{darkgreen}{\bf99.99} & 99.99 & \textcolor{darkgreen}{\bf99.99} \\
\cline{2-11}
 & \multirow[t]{3}{*}{0.90} & 0.90 & 4.90 & 1.94 & \textcolor{red}{\bf 7.30} & 4.96 & 2.93 & 4.61 & \textcolor{red}{\bf 5.43} & 3.92 \\
 &  & 0.95 & \uline{28.49} & 17.76 & 37.59 & \uline{\textcolor{darkgreen}{\bf28.65}} & 23.49 & \textcolor{darkgreen}{\bf31.02} & 33.80 & 27.26 \\
 &  & 0.99 & \uline{89.02} & 83.57 & 95.25 & \textcolor{darkgreen}{\bf89.27} & 87.46 & \textcolor{darkgreen}{\bf92.97} & 93.05 & 90.49 \\
\cline{1-11} \cline{2-11}
\bottomrule
\end{tabular}

CX-S: conditional exact tests based on total successes, UX: unconditional exact test, Asymp.: asymptotic test, GB: generalized version of Boschloo's test
\end{table}

\begin{figure}
	\centering \includegraphics[width = .85\textwidth]{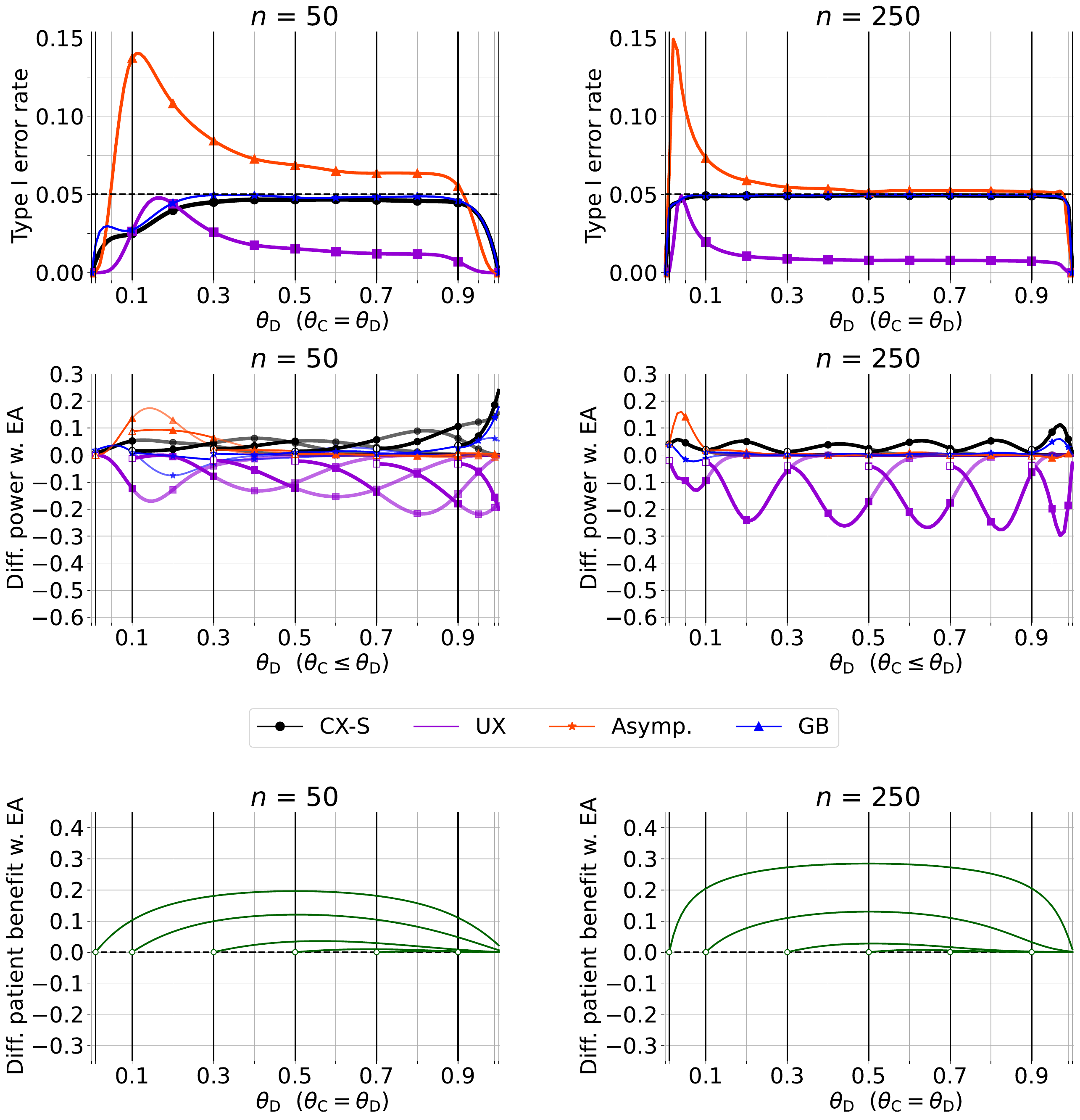}
	\caption{Tempered doubly adaptive biased coin design Neyman allocation; type I error rate profiles~(top), power difference curves compared to the same test under equal allocation~(subtracted, middle), and patient benefit difference curves with respect to equal allocation~(subtracted, bottom) for trial sizes~$n\in\{50,250\}$. We consider the conditional exact, unconditional exact~(UX), asymptotic~(Asymp.), and generalized version of Boschloo's~(GB) Wald tests. For all tests, the significance level was set to~$0.05$.}
	\label{fig:TNA_plots}
\end{figure}

\begin{table}[]
    \centering
 \caption{Rejection rates~$(\%)$ for the tempered doubly adaptive biased coin design Neyman allocation procedure~(burn-in equal to~6 participants per arm). The bold values indicate highest power or type I error inflation, depending on the parameter configuration considered. The underlined values indicate higher power than under that test for the equal allocation procedure. The twice underlined values indicate higher power than the test with highest power for the equal allocation procedure. Both upper and lower significance levels were set to~$2.5\%$.}\label{RR_TNA}
\scriptsize \begin{tabular}{lllrrrrrrrr}
\toprule
 &  &  & \multicolumn{4}{c}{Tempered DBCD Neyman allocation} & \multicolumn{4}{c}{Equal allocation} \\
 &  &  & CX-S & UX & Asymp. & GB & CX-S & UX & Asymp. & GB \\
$\Iend$ & $\theta_\text{C}$ & $\theta_{\text{D}}$ &  &  &  &  &  &  &  &  \\
\midrule
\multirow[t]{32}{*}{50} & \multirow[t]{6}{*}{0.01} & 0.01 & 0.74 & 0.00 & 0.05 & 1.80 & 0.00 & 0.02 & 0.02 & 0.00 \\
 &  & 0.10 & \uline{8.00} & 6.19 & 32.71 & \uline{\textcolor{darkgreen}{\bf8.88}} & 2.79 & \textcolor{darkgreen}{\bf18.58} & 19.05 & 7.81 \\
 &  & 0.20 & \uline{38.90} & \textcolor{darkgreen}{\bf51.33} & 80.30 & 42.07 & 34.25 & \textcolor{darkgreen}{\bf64.17} & 67.38 & 49.64 \\
 &  & 0.30 & \uline{79.79} & \textcolor{darkgreen}{\bf85.40} & 95.56 & 82.02 & 76.47 & \textcolor{darkgreen}{\bf89.35} & 92.26 & 84.89 \\
 &  & 0.40 & \uline{96.59} & 96.49 & 99.33 & \textcolor{darkgreen}{\bf97.01} & 95.50 & \textcolor{darkgreen}{\bf97.92} & 98.82 & 97.37 \\
 &  & 0.50 & \uline{99.67} & 99.46 & 99.94 & \textcolor{darkgreen}{\bf99.72} & 99.53 & \textcolor{darkgreen}{\bf99.78} & 99.89 & 99.75 \\
\cline{2-11}
 & \multirow[t]{6}{*}{0.10} & 0.10 & 2.50 & 2.63 & \textcolor{red}{\bf 13.73} & 2.74 & 0.88 & 3.90 & 4.87 & 1.91 \\
 &  & 0.25 & \uline{21.19} & 21.18 & 38.81 & \textcolor{darkgreen}{\bf22.68} & 18.04 & \textcolor{darkgreen}{\bf24.65} & 30.66 & 23.79 \\
 &  & 0.40 & \uline{65.26} & 55.82 & 77.99 & \textcolor{darkgreen}{\bf67.43} & 59.05 & \textcolor{darkgreen}{\bf69.03} & 75.60 & 68.98 \\
 &  & 0.50 & \uline{87.68} & 79.22 & 93.01 & \textcolor{darkgreen}{\bf88.66} & 83.26 & \textcolor{darkgreen}{\bf89.55} & 92.49 & 89.54 \\
 &  & 0.60 & \uline{97.40} & 93.56 & 98.62 & \textcolor{darkgreen}{\bf97.59} & 95.95 & \textcolor{darkgreen}{\bf97.88} & 98.52 & \textcolor{darkgreen}{\bf97.88} \\
 &  & 0.70 & \uline{99.72} & 98.91 & 99.86 & \textcolor{darkgreen}{\bf99.74} & 99.53 & \textcolor{darkgreen}{\bf99.76} & 99.85 & \textcolor{darkgreen}{\bf99.76} \\
\cline{2-11}
 & \multirow[t]{6}{*}{0.30} & 0.30 & 4.51 & 2.57 & \textcolor{red}{\bf 8.44} & 4.96 & 2.29 & 4.48 & \textcolor{red}{\bf 6.42} & 4.48 \\
 &  & 0.45 & \uline{17.86} & 9.30 & 23.91 & \uuline{\textcolor{darkgreen}{\bf18.59}} & 13.46 & \textcolor{darkgreen}{\bf18.20} & 22.15 & \textcolor{darkgreen}{\bf18.20} \\
 &  & 0.60 & \uuline{55.40} & 38.09 & 62.84 & \uuline{\textcolor{darkgreen}{\bf56.21}} & 50.59 & \textcolor{darkgreen}{\bf53.50} & 62.36 & \textcolor{darkgreen}{\bf53.50} \\
 &  & 0.70 & \uuline{81.10} & 66.64 & 85.85 & \uuline{\textcolor{darkgreen}{\bf81.42}} & 78.22 & \textcolor{darkgreen}{\bf79.30} & 85.96 & \textcolor{darkgreen}{\bf79.30} \\
 &  & 0.80 & \uuline{95.79} & 89.13 & 97.10 & \uuline{\textcolor{darkgreen}{\bf95.84}} & 94.61 & \textcolor{darkgreen}{\bf95.39} & 97.25 & \textcolor{darkgreen}{\bf95.39} \\
 &  & 0.90 & \uline{99.73} & 98.70 & 99.83 & \textcolor{darkgreen}{\bf99.74} & 99.53 & \textcolor{darkgreen}{\bf99.76} & 99.85 & \textcolor{darkgreen}{\bf99.76} \\
\cline{2-11}
 & \multirow[t]{6}{*}{0.50} & 0.50 & 4.66 & 1.52 & \textcolor{red}{\bf 6.88} & 4.80 & 3.28 & 3.73 & \textcolor{red}{\bf 6.51} & 3.73 \\
 &  & 0.65 & \uuline{17.74} & 7.95 & 22.06 & \uuline{\textcolor{darkgreen}{\bf17.94}} & 13.92 & \textcolor{darkgreen}{\bf16.56} & 21.82 & \textcolor{darkgreen}{\bf16.56} \\
 &  & 0.70 & \uuline{28.90} & 14.59 & 34.07 & \uuline{\textcolor{darkgreen}{\bf29.20}} & 23.25 & \textcolor{darkgreen}{\bf28.18} & 33.97 & \textcolor{darkgreen}{\bf28.18} \\
 &  & 0.80 & \uline{59.92} & 38.98 & 65.15 & \textcolor{darkgreen}{\bf60.50} & 51.04 & \textcolor{darkgreen}{\bf60.64} & 65.60 & \textcolor{darkgreen}{\bf60.64} \\
 &  & 0.90 & \uline{89.41} & 75.17 & 91.84 & \uuline{\textcolor{darkgreen}{\bf89.76}} & 83.26 & \textcolor{darkgreen}{\bf89.55} & 92.49 & 89.54 \\
 &  & 0.95 & \uuline{97.38} & 90.94 & 98.15 & \uuline{\textcolor{darkgreen}{\bf97.51}} & 95.15 & \textcolor{darkgreen}{\bf97.32} & 98.35 & 97.31 \\
\cline{2-11}
 & \multirow[t]{5}{*}{0.70} & 0.70 & 4.64 & 1.22 & \textcolor{red}{\bf 6.36} & 4.84 & 2.29 & 4.48 & \textcolor{red}{\bf 6.42} & 4.48 \\
 &  & 0.80 & \uuline{11.83} & 4.21 & 14.97 & \uuline{\textcolor{darkgreen}{\bf12.26}} & 6.87 & \textcolor{darkgreen}{\bf11.17} & 15.26 & 11.15 \\
 &  & 0.90 & \uuline{41.07} & 21.03 & 46.60 & \uuline{\textcolor{darkgreen}{\bf42.05}} & 30.49 & \textcolor{darkgreen}{\bf39.05} & 46.29 & 38.64 \\
 &  & 0.95 & \uuline{67.13} & 42.28 & 71.52 & \uuline{\textcolor{darkgreen}{\bf68.02}} & 54.92 & \textcolor{darkgreen}{\bf64.19} & 71.05 & 62.60 \\
 &  & 0.99 & \uuline{90.70} & 69.92 & 91.51 & \uuline{\textcolor{darkgreen}{\bf91.01}} & 76.47 & \textcolor{darkgreen}{\bf89.35} & 92.26 & 84.89 \\
\cline{2-11}
 & \multirow[t]{3}{*}{0.90} & 0.90 & 4.46 & 0.71 & \textcolor{red}{\bf 5.54} & 4.61 & 0.88 & 3.90 & 4.87 & 1.91 \\
 &  & 0.95 & \uuline{8.33} & 1.18 & 8.68 & \uuline{\textcolor{darkgreen}{\bf8.50}} & 1.30 & \textcolor{darkgreen}{\bf7.21} & 8.10 & 3.17 \\
 &  & 0.99 & \uuline{21.32} & 2.92 & 19.40 & \uuline{\textcolor{darkgreen}{\bf21.63}} & 2.79 & \textcolor{darkgreen}{\bf18.58} & 19.05 & 7.81 \\
\cline{1-11} \cline{2-11}
\multirow[t]{32}{*}{250} & \multirow[t]{6}{*}{0.01} & 0.01 & 4.13 & 0.10 & \textcolor{red}{\bf 6.40} & 4.13 & 0.10 & 2.15 & 2.15 & 0.50 \\
 &  & 0.05 & \uline{34.37} & \textcolor{darkgreen}{\bf36.03} & 59.61 & 34.62 & 29.83 & \textcolor{darkgreen}{\bf45.48} & 45.49 & 36.25 \\
 &  & 0.10 & \uline{88.99} & 83.53 & 95.22 & \textcolor{darkgreen}{\bf89.25} & 87.46 & \textcolor{darkgreen}{\bf92.97} & 93.05 & 90.49 \\
 &  & 0.13 & \uline{97.97} & 95.53 & 99.22 & \textcolor{darkgreen}{\bf98.01} & 97.42 & \textcolor{darkgreen}{\bf98.83} & 98.88 & 98.25 \\
 &  & 0.17 & \uline{\textcolor{darkgreen}{\bf99.88}} & 99.53 & 99.96 & \textcolor{darkgreen}{\bf99.88} & 99.82 & \textcolor{darkgreen}{\bf99.93} & 99.94 & 99.88 \\
 &  & 0.20 & \uline{\textcolor{darkgreen}{\bf99.99}} & 99.94 & 100.00 & \textcolor{darkgreen}{\bf99.99} & 99.98 & \textcolor{darkgreen}{\bf99.99} & 100.00 & \textcolor{darkgreen}{\bf99.99} \\
\cline{2-11}
 & \multirow[t]{6}{*}{0.10} & 0.10 & 4.90 & 1.95 & \textcolor{red}{\bf 7.34} & 4.94 & 2.93 & 4.61 & \textcolor{red}{\bf 5.43} & 3.92 \\
 &  & 0.20 & \uline{58.80} & 36.03 & 63.11 & \uline{\textcolor{darkgreen}{\bf58.89}} & 53.82 & \textcolor{darkgreen}{\bf60.13} & 61.93 & 58.40 \\
 &  & 0.25 & \uline{88.08} & 71.29 & 89.87 & \uline{\textcolor{darkgreen}{\bf88.12}} & 85.48 & \textcolor{darkgreen}{\bf88.73} & 89.46 & 88.00 \\
 &  & 0.30 & \uline{98.13} & 92.50 & 98.49 & \uline{\textcolor{darkgreen}{\bf98.14}} & 97.50 & \textcolor{darkgreen}{\bf98.26} & 98.40 & 98.12 \\
 &  & 0.35 & \uline{\textcolor{darkgreen}{\bf99.84}} & 98.90 & 99.88 & \textcolor{darkgreen}{\bf99.84} & 99.76 & \textcolor{darkgreen}{\bf99.85} & 99.87 & 99.84 \\
 &  & 0.40 & \textcolor{darkgreen}{\bf99.99} & 99.91 & 99.99 & \textcolor{darkgreen}{\bf99.99} & \textcolor{darkgreen}{\bf99.99} & \textcolor{darkgreen}{\bf99.99} & 99.99 & \textcolor{darkgreen}{\bf99.99} \\
\cline{2-11}
 & \multirow[t]{6}{*}{0.30} & 0.30 & 4.91 & 0.88 & \textcolor{red}{\bf 5.46} & 4.94 & 3.62 & 4.94 & \textcolor{red}{\bf 5.15} & 4.80 \\
 &  & 0.40 & \uline{37.64} & 16.20 & 39.24 & \uline{\textcolor{darkgreen}{\bf37.76}} & 33.84 & \textcolor{darkgreen}{\bf37.78} & 39.29 & 37.74 \\
 &  & 0.50 & \uuline{89.98} & 72.41 & 90.63 & \uuline{\textcolor{darkgreen}{\bf90.03}} & 87.62 & \textcolor{darkgreen}{\bf89.70} & 90.42 & \textcolor{darkgreen}{\bf89.70} \\
 &  & 0.55 & \uuline{98.12} & 91.75 & 98.27 & \uuline{\textcolor{darkgreen}{\bf98.13}} & 97.43 & \textcolor{darkgreen}{\bf98.08} & 98.16 & \textcolor{darkgreen}{\bf98.08} \\
 &  & 0.60 & \uline{\textcolor{darkgreen}{\bf99.81}} & 98.58 & 99.82 & \textcolor{darkgreen}{\bf99.81} & 99.72 & \textcolor{darkgreen}{\bf99.81} & 99.81 & \textcolor{darkgreen}{\bf99.81} \\
 &  & 0.65 & \uline{\textcolor{darkgreen}{\bf99.99}} & 99.87 & 99.99 & \textcolor{darkgreen}{\bf99.99} & 99.98 & \textcolor{darkgreen}{\bf99.99} & 99.99 & \textcolor{darkgreen}{\bf99.99} \\
\cline{2-11}
 & \multirow[t]{6}{*}{0.50} & 0.50 & 4.94 & 0.78 & \textcolor{red}{\bf 5.16} & 4.99 & 3.67 & 4.97 & 4.97 & 4.97 \\
 &  & 0.60 & \uuline{35.31} & 14.12 & 36.23 & \uuline{\textcolor{darkgreen}{\bf35.48}} & 30.59 & \textcolor{darkgreen}{\bf35.17} & 35.31 & \textcolor{darkgreen}{\bf35.17} \\
 &  & 0.70 & \uuline{90.05} & 72.08 & 90.52 & \uuline{\textcolor{darkgreen}{\bf90.10}} & 87.62 & \textcolor{darkgreen}{\bf89.70} & 90.42 & \textcolor{darkgreen}{\bf89.70} \\
 &  & 0.75 & \uuline{98.52} & 92.97 & 98.62 & \uuline{\textcolor{darkgreen}{\bf98.53}} & 98.06 & \textcolor{darkgreen}{\bf98.48} & 98.67 & \textcolor{darkgreen}{\bf98.48} \\
 &  & 0.80 & \uline{99.91} & 99.24 & 99.92 & \textcolor{darkgreen}{\bf99.92} & 99.88 & \textcolor{darkgreen}{\bf99.92} & 99.93 & \textcolor{darkgreen}{\bf99.92} \\
 &  & 0.85 & \textcolor{darkgreen}{\bf100.00} & 99.98 & 100.00 & \textcolor{darkgreen}{\bf100.00} & \textcolor{darkgreen}{\bf100.00} & \textcolor{darkgreen}{\bf100.00} & 100.00 & \textcolor{darkgreen}{\bf100.00} \\
\cline{2-11}
 & \multirow[t]{5}{*}{0.70} & 0.70 & 4.95 & 0.78 & \textcolor{red}{\bf 5.23} & 4.98 & 3.62 & 4.94 & \textcolor{red}{\bf 5.15} & 4.80 \\
 &  & 0.80 & \uline{44.46} & 20.01 & 45.52 & \uline{\textcolor{darkgreen}{\bf44.62}} & 39.24 & \textcolor{darkgreen}{\bf44.64} & 45.44 & 43.87 \\
 &  & 0.85 & \uline{81.65} & 57.64 & 82.33 & \uuline{\textcolor{darkgreen}{\bf81.74}} & 77.67 & \textcolor{darkgreen}{\bf81.69} & 82.38 & 81.02 \\
 &  & 0.90 & \uline{98.25} & 91.78 & 98.36 & \uuline{\textcolor{darkgreen}{\bf98.27}} & 97.50 & \textcolor{darkgreen}{\bf98.26} & 98.40 & 98.12 \\
 &  & 0.95 & \uline{\textcolor{darkgreen}{\bf99.99}} & 99.83 & 99.99 & \textcolor{darkgreen}{\bf99.99} & 99.98 & \textcolor{darkgreen}{\bf99.99} & 99.99 & \textcolor{darkgreen}{\bf99.99} \\
\cline{2-11}
 & \multirow[t]{3}{*}{0.90} & 0.90 & 4.92 & 0.73 & \textcolor{red}{\bf 5.17} & 4.96 & 2.93 & 4.61 & \textcolor{red}{\bf 5.43} & 3.92 \\
 &  & 0.95 & \uuline{31.96} & 11.14 & 32.77 & \uuline{\textcolor{darkgreen}{\bf32.05}} & 23.49 & \textcolor{darkgreen}{\bf31.02} & 33.80 & 27.26 \\
 &  & 0.99 & \uuline{93.34} & 74.43 & 93.54 & \uuline{\textcolor{darkgreen}{\bf93.40}} & 87.46 & \textcolor{darkgreen}{\bf92.97} & 93.05 & 90.49 \\
\cline{1-11} \cline{2-11}
\bottomrule
\end{tabular}

CX-S: conditional exact tests based on total successes, UX: unconditional exact test, Asymp.: asymptotic test, GB: generalized version of Boschloo's test
\end{table}

\begin{figure}
	
	\centering
	\includegraphics[width=.85\linewidth]{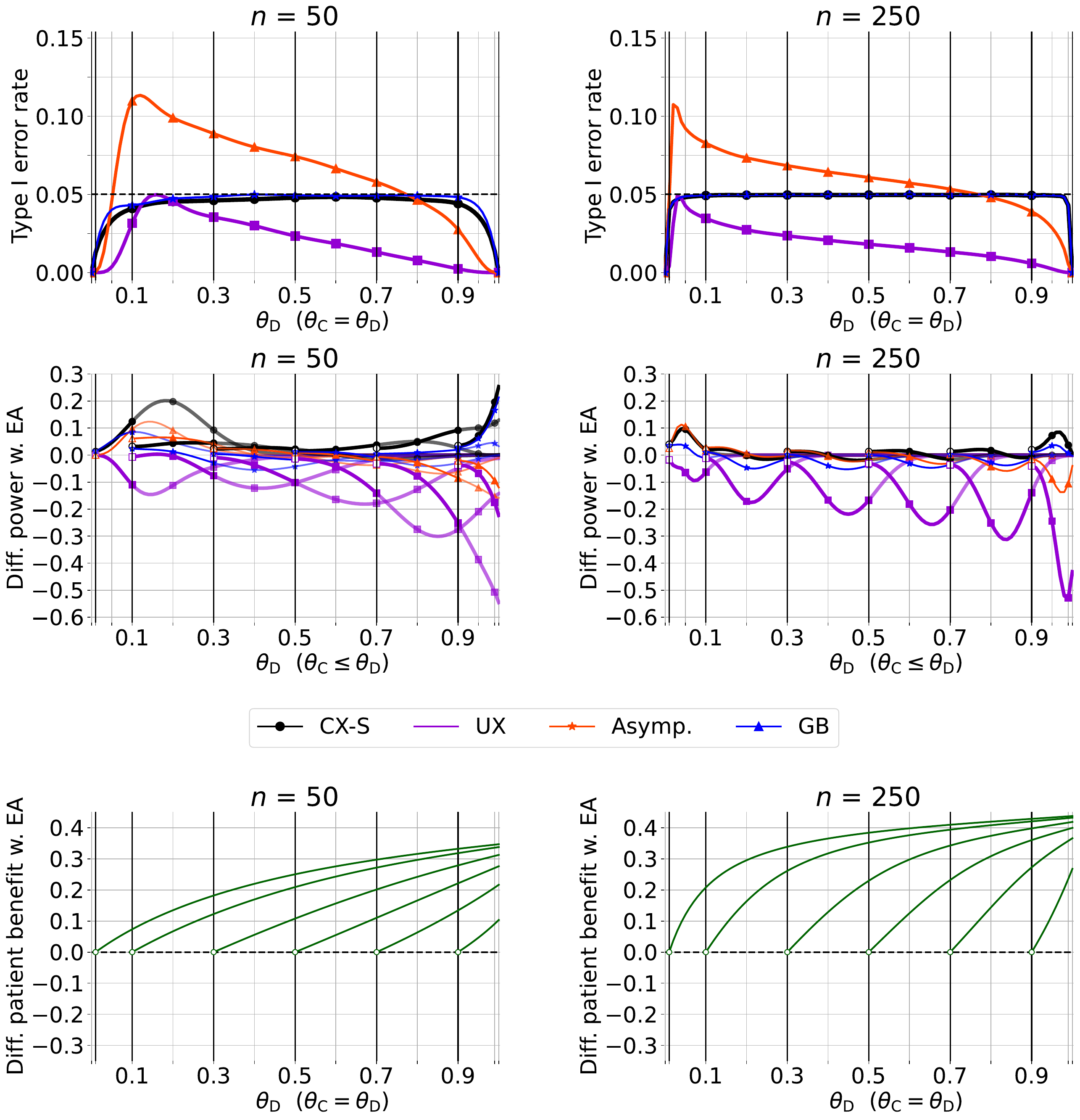}
	
	\caption{Bayesian response-adaptive design based on~\citet{THALL2007859}; type I error rate profiles~(top), power difference curves compared to the same test under equal allocation~(subtracted, middle), and patient benefit difference curves with respect to equal allocation~(subtracted, bottom) for trial sizes~$n\in\{50,250\}$. We consider the conditional exact, unconditional exact~(UX), asymptotic~(Asymp.), and generalized version of Boschloo's~(GB) Wald tests. For all tests, the significance level was set to~$0.05$.}
	\label{fig:TW_Bayesian RAR_plots}

\end{figure}

\begin{table}[]
    \centering
\caption{Rejection rates~$(\%)$ for the tuned Bayesian RAR procedure from~\citet{THALL2007859}~(burn-in equal to~6 participants per arm). The bold values indicate highest power or type I error inflation, depending on the parameter configuration considered. The underlined values indicate higher power than under that test for the equal allocation procedure. The twice underlined values indicate higher power than the test with highest power for the equal allocation procedure. Both upper and lower significance levels were set to~$2.5\%$.}\label{RR_BRAR}
\scriptsize \begin{tabular}{lllrrrrrrrr}
\toprule
 &  &  & \multicolumn{4}{c}{BRAR} & \multicolumn{4}{c}{Equal allocation} \\
 &  &  & CX-S & UX & Asymp. & GB & CX-S & UX & Asymp. & GB \\
$\Iend$ & $\theta_\text{C}$ & $\theta_{\text{D}}$ &  &  &  &  &  &  &  &  \\
\midrule
\multirow[t]{32}{*}{50} & \multirow[t]{6}{*}{0.01} & 0.01 & 1.28 & 0.00 & 0.04 & 1.48 & 0.00 & 0.02 & 0.02 & 0.00 \\
 &  & 0.10 & \uline{15.18} & 7.58 & 28.69 & \uline{\textcolor{darkgreen}{\bf16.32}} & 2.79 & \textcolor{darkgreen}{\bf18.58} & 19.05 & 7.81 \\
 &  & 0.20 & \uline{54.02} & 52.93 & 76.55 & \uline{\textcolor{darkgreen}{\bf54.44}} & 34.25 & \textcolor{darkgreen}{\bf64.17} & 67.38 & 49.64 \\
 &  & 0.30 & \uline{85.72} & 84.94 & 94.47 & \uline{\textcolor{darkgreen}{\bf86.08}} & 76.47 & \textcolor{darkgreen}{\bf89.35} & 92.26 & 84.89 \\
 &  & 0.40 & \uline{97.11} & 96.22 & 99.14 & \textcolor{darkgreen}{\bf97.29} & 95.50 & \textcolor{darkgreen}{\bf97.92} & 98.82 & 97.37 \\
 &  & 0.50 & \uline{99.64} & 99.41 & 99.92 & \textcolor{darkgreen}{\bf99.67} & 99.53 & \textcolor{darkgreen}{\bf99.78} & 99.89 & 99.75 \\
\cline{2-11}
 & \multirow[t]{6}{*}{0.10} & 0.10 & 4.10 & 3.15 & \textcolor{red}{\bf 10.99} & 4.31 & 0.88 & 3.90 & 4.87 & 1.91 \\
 &  & 0.25 & \uline{22.59} & 21.11 & 36.32 & \textcolor{darkgreen}{\bf23.21} & 18.04 & \textcolor{darkgreen}{\bf24.65} & 30.66 & 23.79 \\
 &  & 0.40 & \uline{62.57} & 56.83 & 75.82 & \textcolor{darkgreen}{\bf63.47} & 59.05 & \textcolor{darkgreen}{\bf69.03} & 75.60 & 68.98 \\
 &  & 0.50 & \uline{84.58} & 79.29 & 91.51 & \textcolor{darkgreen}{\bf85.30} & 83.26 & \textcolor{darkgreen}{\bf89.55} & 92.49 & 89.54 \\
 &  & 0.60 & 95.83 & 92.49 & 98.01 & \textcolor{darkgreen}{\bf96.16} & 95.95 & \textcolor{darkgreen}{\bf97.88} & 98.52 & \textcolor{darkgreen}{\bf97.88} \\
 &  & 0.70 & 99.35 & 98.07 & 99.71 & \textcolor{darkgreen}{\bf99.41} & 99.53 & \textcolor{darkgreen}{\bf99.76} & 99.85 & \textcolor{darkgreen}{\bf99.76} \\
\cline{2-11}
 & \multirow[t]{6}{*}{0.30} & 0.30 & 4.61 & 3.55 & \textcolor{red}{\bf 8.89} & 4.86 & 2.29 & 4.48 & \textcolor{red}{\bf 6.42} & 4.48 \\
 &  & 0.45 & \uline{16.14} & 11.53 & 23.61 & \textcolor{darkgreen}{\bf16.93} & 13.46 & \textcolor{darkgreen}{\bf18.20} & 22.15 & \textcolor{darkgreen}{\bf18.20} \\
 &  & 0.60 & \uline{50.65} & 37.08 & 59.75 & \textcolor{darkgreen}{\bf51.49} & 50.59 & \textcolor{darkgreen}{\bf53.50} & 62.36 & \textcolor{darkgreen}{\bf53.50} \\
 &  & 0.70 & 76.69 & 61.45 & 82.24 & \textcolor{darkgreen}{\bf77.01} & 78.22 & \textcolor{darkgreen}{\bf79.30} & 85.96 & \textcolor{darkgreen}{\bf79.30} \\
 &  & 0.80 & 93.49 & 82.65 & 95.08 & \textcolor{darkgreen}{\bf93.64} & 94.61 & \textcolor{darkgreen}{\bf95.39} & 97.25 & \textcolor{darkgreen}{\bf95.39} \\
 &  & 0.90 & 99.22 & 94.97 & 99.33 & \textcolor{darkgreen}{\bf99.29} & 99.53 & \textcolor{darkgreen}{\bf99.76} & 99.85 & \textcolor{darkgreen}{\bf99.76} \\
\cline{2-11}
 & \multirow[t]{6}{*}{0.50} & 0.50 & 4.79 & 2.34 & \textcolor{red}{\bf 7.42} & 4.95 & 3.28 & 3.73 & \textcolor{red}{\bf 6.51} & 3.73 \\
 &  & 0.65 & \uuline{16.71} & 8.34 & 20.87 & \uuline{\textcolor{darkgreen}{\bf16.89}} & 13.92 & \textcolor{darkgreen}{\bf16.56} & 21.82 & \textcolor{darkgreen}{\bf16.56} \\
 &  & 0.70 & \uline{26.96} & 14.06 & 31.56 & \textcolor{darkgreen}{\bf27.21} & 23.25 & \textcolor{darkgreen}{\bf28.18} & 33.97 & \textcolor{darkgreen}{\bf28.18} \\
 &  & 0.80 & \uline{56.08} & 33.16 & 59.77 & \textcolor{darkgreen}{\bf56.76} & 51.04 & \textcolor{darkgreen}{\bf60.64} & 65.60 & \textcolor{darkgreen}{\bf60.64} \\
 &  & 0.90 & \uline{85.92} & 62.02 & 86.26 & \textcolor{darkgreen}{\bf86.56} & 83.26 & \textcolor{darkgreen}{\bf89.55} & 92.49 & 89.54 \\
 &  & 0.95 & \uline{95.64} & 76.36 & 94.63 & \textcolor{darkgreen}{\bf95.92} & 95.15 & \textcolor{darkgreen}{\bf97.32} & 98.35 & 97.31 \\
\cline{2-11}
 & \multirow[t]{5}{*}{0.70} & 0.70 & 4.78 & 1.31 & \textcolor{red}{\bf 5.79} & 4.93 & 2.29 & 4.48 & \textcolor{red}{\bf 6.42} & 4.48 \\
 &  & 0.80 & \uuline{11.62} & 3.36 & 12.45 & \uuline{\textcolor{darkgreen}{\bf12.06}} & 6.87 & \textcolor{darkgreen}{\bf11.17} & 15.26 & 11.15 \\
 &  & 0.90 & \uuline{39.64} & 13.86 & 37.77 & \uuline{\textcolor{darkgreen}{\bf40.63}} & 30.49 & \textcolor{darkgreen}{\bf39.05} & 46.29 & 38.64 \\
 &  & 0.95 & \uuline{64.96} & 25.50 & 58.98 & \uuline{\textcolor{darkgreen}{\bf66.04}} & 54.92 & \textcolor{darkgreen}{\bf64.19} & 71.05 & 62.60 \\
 &  & 0.99 & \uline{88.28} & 38.64 & 76.93 & \uline{\textcolor{darkgreen}{\bf89.20}} & 76.47 & \textcolor{darkgreen}{\bf89.35} & 92.26 & 84.89 \\
\cline{2-11}
 & \multirow[t]{3}{*}{0.90} & 0.90 & 4.41 & 0.24 & 2.75 & 4.82 & 0.88 & 3.90 & 4.87 & 1.91 \\
 &  & 0.95 & \uuline{8.46} & 0.42 & 4.32 & \uuline{\textcolor{darkgreen}{\bf9.29}} & 1.30 & \textcolor{darkgreen}{\bf7.21} & 8.10 & 3.17 \\
 &  & 0.99 & \uuline{22.37} & 1.05 & 9.59 & \uuline{\textcolor{darkgreen}{\bf24.38}} & 2.79 & \textcolor{darkgreen}{\bf18.58} & 19.05 & 7.81 \\
\cline{1-11} \cline{2-11}
\multirow[t]{32}{*}{250} & \multirow[t]{6}{*}{0.01} & 0.01 & 4.06 & 0.40 & 4.69 & 4.06 & 0.10 & 2.15 & 2.15 & 0.50 \\
 &  & 0.05 & \uline{39.37} & 39.01 & 56.04 & \uline{\textcolor{darkgreen}{\bf39.57}} & 29.83 & \textcolor{darkgreen}{\bf45.48} & 45.49 & 36.25 \\
 &  & 0.10 & \uline{89.41} & 86.62 & 94.75 & \textcolor{darkgreen}{\bf89.52} & 87.46 & \textcolor{darkgreen}{\bf92.97} & 93.05 & 90.49 \\
 &  & 0.13 & \uline{97.78} & 96.70 & 99.12 & \textcolor{darkgreen}{\bf97.81} & 97.42 & \textcolor{darkgreen}{\bf98.83} & 98.88 & 98.25 \\
 &  & 0.17 & \uline{\textcolor{darkgreen}{\bf99.83}} & 99.68 & 99.95 & \textcolor{darkgreen}{\bf99.83} & 99.82 & \textcolor{darkgreen}{\bf99.93} & 99.94 & 99.88 \\
 &  & 0.20 & \textcolor{darkgreen}{\bf99.98} & 99.96 & 100.00 & \textcolor{darkgreen}{\bf99.98} & 99.98 & \textcolor{darkgreen}{\bf99.99} & 100.00 & \textcolor{darkgreen}{\bf99.99} \\
\cline{2-11}
 & \multirow[t]{6}{*}{0.10} & 0.10 & 4.92 & 3.46 & \textcolor{red}{\bf 8.27} & 4.95 & 2.93 & 4.61 & \textcolor{red}{\bf 5.43} & 3.92 \\
 &  & 0.20 & 53.64 & 43.00 & 62.44 & \textcolor{darkgreen}{\bf53.73} & 53.82 & \textcolor{darkgreen}{\bf60.13} & 61.93 & 58.40 \\
 &  & 0.25 & 83.93 & 75.07 & 88.75 & \textcolor{darkgreen}{\bf84.00} & 85.48 & \textcolor{darkgreen}{\bf88.73} & 89.46 & 88.00 \\
 &  & 0.30 & 96.73 & 93.11 & 98.03 & \textcolor{darkgreen}{\bf96.75} & 97.50 & \textcolor{darkgreen}{\bf98.26} & 98.40 & 98.12 \\
 &  & 0.35 & \textcolor{darkgreen}{\bf99.61} & 98.80 & 99.80 & \textcolor{darkgreen}{\bf99.61} & 99.76 & \textcolor{darkgreen}{\bf99.85} & 99.87 & 99.84 \\
 &  & 0.40 & \textcolor{darkgreen}{\bf99.97} & 99.86 & 99.99 & \textcolor{darkgreen}{\bf99.97} & \textcolor{darkgreen}{\bf99.99} & \textcolor{darkgreen}{\bf99.99} & 99.99 & \textcolor{darkgreen}{\bf99.99} \\
\cline{2-11}
 & \multirow[t]{6}{*}{0.30} & 0.30 & 4.96 & 2.35 & \textcolor{red}{\bf 6.84} & 4.98 & 3.62 & 4.94 & \textcolor{red}{\bf 5.15} & 4.80 \\
 &  & 0.40 & 33.69 & 21.12 & 38.83 & \textcolor{darkgreen}{\bf33.76} & 33.84 & \textcolor{darkgreen}{\bf37.78} & 39.29 & 37.74 \\
 &  & 0.50 & 85.86 & 73.00 & 88.55 & \textcolor{darkgreen}{\bf85.92} & 87.62 & \textcolor{darkgreen}{\bf89.70} & 90.42 & \textcolor{darkgreen}{\bf89.70} \\
 &  & 0.55 & 96.54 & 90.53 & 97.39 & \textcolor{darkgreen}{\bf96.57} & 97.43 & \textcolor{darkgreen}{\bf98.08} & 98.16 & \textcolor{darkgreen}{\bf98.08} \\
 &  & 0.60 & \textcolor{darkgreen}{\bf99.48} & 97.78 & 99.63 & \textcolor{darkgreen}{\bf99.48} & 99.72 & \textcolor{darkgreen}{\bf99.81} & 99.81 & \textcolor{darkgreen}{\bf99.81} \\
 &  & 0.65 & \textcolor{darkgreen}{\bf99.95} & 99.65 & 99.97 & \textcolor{darkgreen}{\bf99.95} & 99.98 & \textcolor{darkgreen}{\bf99.99} & 99.99 & \textcolor{darkgreen}{\bf99.99} \\
\cline{2-11}
 & \multirow[t]{6}{*}{0.50} & 0.50 & 4.96 & 1.81 & \textcolor{red}{\bf 6.08} & 4.98 & 3.67 & 4.97 & 4.97 & 4.97 \\
 &  & 0.60 & \uline{31.91} & 17.02 & 34.76 & \textcolor{darkgreen}{\bf32.00} & 30.59 & \textcolor{darkgreen}{\bf35.17} & 35.31 & \textcolor{darkgreen}{\bf35.17} \\
 &  & 0.70 & 86.12 & 69.38 & 87.43 & \textcolor{darkgreen}{\bf86.17} & 87.62 & \textcolor{darkgreen}{\bf89.70} & 90.42 & \textcolor{darkgreen}{\bf89.70} \\
 &  & 0.75 & 97.16 & 89.65 & 97.47 & \textcolor{darkgreen}{\bf97.17} & 98.06 & \textcolor{darkgreen}{\bf98.48} & 98.67 & \textcolor{darkgreen}{\bf98.48} \\
 &  & 0.80 & 99.71 & 97.93 & 99.74 & \textcolor{darkgreen}{\bf99.72} & 99.88 & \textcolor{darkgreen}{\bf99.92} & 99.93 & \textcolor{darkgreen}{\bf99.92} \\
 &  & 0.85 & \textcolor{darkgreen}{\bf99.99} & 99.77 & 99.99 & \textcolor{darkgreen}{\bf99.99} & \textcolor{darkgreen}{\bf100.00} & \textcolor{darkgreen}{\bf100.00} & 100.00 & \textcolor{darkgreen}{\bf100.00} \\
\cline{2-11}
 & \multirow[t]{5}{*}{0.70} & 0.70 & 4.96 & 1.32 & \textcolor{red}{\bf 5.34} & 4.98 & 3.62 & 4.94 & \textcolor{red}{\bf 5.15} & 4.80 \\
 &  & 0.80 & \uline{40.93} & 19.54 & 41.15 & \textcolor{darkgreen}{\bf41.04} & 39.24 & \textcolor{darkgreen}{\bf44.64} & 45.44 & 43.87 \\
 &  & 0.85 & 77.25 & 51.43 & 76.70 & \textcolor{darkgreen}{\bf77.32} & 77.67 & \textcolor{darkgreen}{\bf81.69} & 82.38 & 81.02 \\
 &  & 0.90 & 96.76 & 84.27 & 96.35 & \textcolor{darkgreen}{\bf96.78} & 97.50 & \textcolor{darkgreen}{\bf98.26} & 98.40 & 98.12 \\
 &  & 0.95 & \textcolor{darkgreen}{\bf99.92} & 97.95 & 99.87 & \textcolor{darkgreen}{\bf99.92} & 99.98 & \textcolor{darkgreen}{\bf99.99} & 99.99 & \textcolor{darkgreen}{\bf99.99} \\
\cline{2-11}
 & \multirow[t]{3}{*}{0.90} & 0.90 & 4.93 & 0.58 & 3.89 & 4.96 & 2.93 & 4.61 & \textcolor{red}{\bf 5.43} & 3.92 \\
 &  & 0.95 & \uline{30.72} & 6.59 & 24.94 & \uline{\textcolor{darkgreen}{\bf30.73}} & 23.49 & \textcolor{darkgreen}{\bf31.02} & 33.80 & 27.26 \\
 &  & 0.99 & \uline{91.11} & 40.18 & 82.44 & \uline{\textcolor{darkgreen}{\bf91.16}} & 87.46 & \textcolor{darkgreen}{\bf92.97} & 93.05 & 90.49 \\
\cline{1-11} \cline{2-11}
\bottomrule
\end{tabular}

CX-S: conditional exact tests based on total successes, UX: unconditional exact test, Asymp.: asymptotic test, GB: generalized version of Boschloo's test
\end{table}

\begin{table}[]
    \centering
\caption{Rejection rates~$(\%)$ for the constrained Markov decision process procedure targeting highest power~(CMDP-P, burn-in equal to~6 participants per arm). The bold values indicate highest power or type I error inflation, depending on the parameter configuration considered. The underlined values indicate higher power than under that test for the equal allocation procedure. The twice underlined values indicate higher power than the test with highest power for the equal allocation procedure. Both upper and lower significance levels were set to~$2.5\%$.}
\label{tab:RR_CMDPP}
\scriptsize \begin{tabular}{lllrrrrrrrr}
\toprule
 &  &  & \multicolumn{4}{c}{CMDP-P} & \multicolumn{4}{c}{Equal allocation} \\
 &  &  & CX-S & UX & Asymp. & GB & CX-S & UX & Asymp. & GB \\
$\Iend$ & $\theta_\text{C}$ & $\theta_{\text{D}}$ &  &  &  &  &  &  &  &  \\
\midrule
\multirow[t]{32}{*}{50} & \multirow[t]{6}{*}{0.01} & 0.01 & 0.34 & 0.03 & 0.03 & 0.34 & 0.00 & 0.02 & 0.02 & 0.00 \\
 &  & 0.10 & \uline{15.56} & \uuline{\textcolor{darkgreen}{\bf22.57}} & 22.90 & \uline{17.25} & 2.79 & \textcolor{darkgreen}{\bf18.58} & 19.05 & 7.81 \\
 &  & 0.20 & \uline{58.02} & \uuline{\textcolor{darkgreen}{\bf70.14}} & 71.42 & \uline{62.09} & 34.25 & \textcolor{darkgreen}{\bf64.17} & 67.38 & 49.64 \\
 &  & 0.30 & \uline{87.45} & \uuline{\textcolor{darkgreen}{\bf91.98}} & 93.29 & \uuline{90.13} & 76.47 & \textcolor{darkgreen}{\bf89.35} & 92.26 & 84.89 \\
 &  & 0.40 & \uline{97.38} & \uuline{98.25} & 98.90 & \uuline{\textcolor{darkgreen}{\bf98.28}} & 95.50 & \textcolor{darkgreen}{\bf97.92} & 98.82 & 97.37 \\
 &  & 0.50 & 99.47 & \textcolor{darkgreen}{\bf99.70} & 99.88 & 99.64 & 99.53 & \textcolor{darkgreen}{\bf99.78} & 99.89 & 99.75 \\
\cline{2-11}
 & \multirow[t]{6}{*}{0.10} & 0.10 & 3.64 & 4.68 & 4.96 & 4.36 & 0.88 & 3.90 & 4.87 & 1.91 \\
 &  & 0.25 & \uuline{25.26} & \uuline{\textcolor{darkgreen}{\bf28.78}} & 29.45 & \uuline{28.66} & 18.04 & \textcolor{darkgreen}{\bf24.65} & 30.66 & 23.79 \\
 &  & 0.40 & \uline{68.24} & \uuline{\textcolor{darkgreen}{\bf72.65}} & 73.04 & \uuline{71.52} & 59.05 & \textcolor{darkgreen}{\bf69.03} & 75.60 & 68.98 \\
 &  & 0.50 & \uline{88.73} & \uuline{\textcolor{darkgreen}{\bf90.88}} & 91.02 & \uuline{90.06} & 83.26 & \textcolor{darkgreen}{\bf89.55} & 92.49 & 89.54 \\
 &  & 0.60 & \uline{97.53} & \uuline{\textcolor{darkgreen}{\bf98.09}} & 98.12 & 97.80 & 95.95 & \textcolor{darkgreen}{\bf97.88} & 98.52 & \textcolor{darkgreen}{\bf97.88} \\
 &  & 0.70 & \uline{99.70} & \uuline{\textcolor{darkgreen}{\bf99.79}} & 99.79 & 99.73 & 99.53 & \textcolor{darkgreen}{\bf99.76} & 99.85 & \textcolor{darkgreen}{\bf99.76} \\
\cline{2-11}
 & \multirow[t]{6}{*}{0.30} & 0.30 & 4.34 & 4.94 & 4.97 & 4.98 & 2.29 & 4.48 & \textcolor{red}{\bf 6.42} & 4.48 \\
 &  & 0.45 & \uline{17.66} & \uuline{\textcolor{darkgreen}{\bf19.28}} & 19.32 & \uuline{18.59} & 13.46 & \textcolor{darkgreen}{\bf18.20} & 22.15 & \textcolor{darkgreen}{\bf18.20} \\
 &  & 0.60 & \uuline{53.78} & \uuline{\textcolor{darkgreen}{\bf57.13}} & 57.19 & \uuline{55.69} & 50.59 & \textcolor{darkgreen}{\bf53.50} & 62.36 & \textcolor{darkgreen}{\bf53.50} \\
 &  & 0.70 & \uuline{79.36} & \uuline{\textcolor{darkgreen}{\bf82.11}} & 82.15 & \uuline{81.12} & 78.22 & \textcolor{darkgreen}{\bf79.30} & 85.96 & \textcolor{darkgreen}{\bf79.30} \\
 &  & 0.80 & \uline{95.19} & \uuline{\textcolor{darkgreen}{\bf96.17}} & 96.18 & \uuline{95.81} & 94.61 & \textcolor{darkgreen}{\bf95.39} & 97.25 & \textcolor{darkgreen}{\bf95.39} \\
 &  & 0.90 & \uline{99.69} & \uuline{\textcolor{darkgreen}{\bf99.79}} & 99.79 & 99.74 & 99.53 & \textcolor{darkgreen}{\bf99.76} & 99.85 & \textcolor{darkgreen}{\bf99.76} \\
\cline{2-11}
 & \multirow[t]{6}{*}{0.50} & 0.50 & 3.80 & 4.92 & 4.93 & 4.56 & 3.28 & 3.73 & \textcolor{red}{\bf 6.51} & 3.73 \\
 &  & 0.65 & \uline{15.97} & \uuline{\textcolor{darkgreen}{\bf18.56}} & 18.60 & \uuline{17.70} & 13.92 & \textcolor{darkgreen}{\bf16.56} & 21.82 & \textcolor{darkgreen}{\bf16.56} \\
 &  & 0.70 & \uline{27.03} & \uuline{\textcolor{darkgreen}{\bf30.15}} & 30.20 & \uuline{29.07} & 23.25 & \textcolor{darkgreen}{\bf28.18} & 33.97 & \textcolor{darkgreen}{\bf28.18} \\
 &  & 0.80 & \uline{58.92} & \uuline{\textcolor{darkgreen}{\bf61.99}} & 62.07 & \uuline{60.88} & 51.04 & \textcolor{darkgreen}{\bf60.64} & 65.60 & \textcolor{darkgreen}{\bf60.64} \\
 &  & 0.90 & \uline{88.69} & \uuline{\textcolor{darkgreen}{\bf90.87}} & 91.01 & \uuline{90.06} & 83.26 & \textcolor{darkgreen}{\bf89.55} & 92.49 & 89.54 \\
 &  & 0.95 & \uline{96.62} & \uuline{\textcolor{darkgreen}{\bf97.80}} & 98.03 & \uuline{97.46} & 95.15 & \textcolor{darkgreen}{\bf97.32} & 98.35 & 97.31 \\
\cline{2-11}
 & \multirow[t]{5}{*}{0.70} & 0.70 & 4.30 & 4.94 & 4.97 & 4.98 & 2.29 & 4.48 & \textcolor{red}{\bf 6.42} & 4.48 \\
 &  & 0.80 & \uline{11.14} & \uuline{\textcolor{darkgreen}{\bf12.80}} & 12.90 & \uuline{12.62} & 6.87 & \textcolor{darkgreen}{\bf11.17} & 15.26 & 11.15 \\
 &  & 0.90 & \uuline{39.17} & \uuline{\textcolor{darkgreen}{\bf43.76}} & 44.38 & \uuline{43.24} & 30.49 & \textcolor{darkgreen}{\bf39.05} & 46.29 & 38.64 \\
 &  & 0.95 & \uuline{64.77} & \uuline{69.21} & 70.84 & \uuline{\textcolor{darkgreen}{\bf69.61}} & 54.92 & \textcolor{darkgreen}{\bf64.19} & 71.05 & 62.60 \\
 &  & 0.99 & \uline{87.44} & \uuline{\textcolor{darkgreen}{\bf92.01}} & 93.27 & \uuline{90.13} & 76.47 & \textcolor{darkgreen}{\bf89.35} & 92.26 & 84.89 \\
\cline{2-11}
 & \multirow[t]{3}{*}{0.90} & 0.90 & 3.63 & 4.67 & 4.95 & 4.36 & 0.88 & 3.90 & 4.87 & 1.91 \\
 &  & 0.95 & \uline{6.48} & \uuline{\textcolor{darkgreen}{\bf8.66}} & 9.08 & \uuline{7.50} & 1.30 & \textcolor{darkgreen}{\bf7.21} & 8.10 & 3.17 \\
 &  & 0.99 & \uline{15.57} & \uuline{\textcolor{darkgreen}{\bf22.57}} & 22.88 & \uline{17.26} & 2.79 & \textcolor{darkgreen}{\bf18.58} & 19.05 & 7.81 \\
\cline{1-11} \cline{2-11}
\multirow[t]{32}{*}{250} & \multirow[t]{6}{*}{0.01} & 0.01 & 1.32 & 0.67 & 0.67 & 1.32 & 0.10 & 2.15 & 2.15 & 0.50 \\
 &  & 0.05 & \uline{\textcolor{darkgreen}{\bf42.23}} & 35.87 & 35.92 & \uline{\textcolor{darkgreen}{\bf42.23}} & 29.83 & \textcolor{darkgreen}{\bf45.48} & 45.49 & 36.25 \\
 &  & 0.10 & \uline{\textcolor{darkgreen}{\bf92.64}} & 90.37 & 90.41 & \uline{\textcolor{darkgreen}{\bf92.64}} & 87.46 & \textcolor{darkgreen}{\bf92.97} & 93.05 & 90.49 \\
 &  & 0.13 & \uline{\textcolor{darkgreen}{\bf98.79}} & 98.27 & 98.29 & \uline{\textcolor{darkgreen}{\bf98.79}} & 97.42 & \textcolor{darkgreen}{\bf98.83} & 98.88 & 98.25 \\
 &  & 0.17 & \uline{\textcolor{darkgreen}{\bf99.93}} & 99.89 & 99.90 & \uline{\textcolor{darkgreen}{\bf99.93}} & 99.82 & \textcolor{darkgreen}{\bf99.93} & 99.94 & 99.88 \\
 &  & 0.20 & \uline{\textcolor{darkgreen}{\bf99.99}} & \textcolor{darkgreen}{\bf99.99} & 99.99 & \textcolor{darkgreen}{\bf99.99} & 99.98 & \textcolor{darkgreen}{\bf99.99} & 100.00 & \textcolor{darkgreen}{\bf99.99} \\
\cline{2-11}
 & \multirow[t]{6}{*}{0.10} & 0.10 & 4.87 & 4.21 & 4.24 & 4.91 & 2.93 & 4.61 & \textcolor{red}{\bf 5.43} & 3.92 \\
 &  & 0.20 & \uuline{61.29} & \uuline{61.20} & 61.43 & \uuline{\textcolor{darkgreen}{\bf61.35}} & 53.82 & \textcolor{darkgreen}{\bf60.13} & 61.93 & 58.40 \\
 &  & 0.25 & \uuline{89.33} & \uuline{\textcolor{darkgreen}{\bf89.38}} & 89.52 & \uuline{89.36} & 85.48 & \textcolor{darkgreen}{\bf88.73} & 89.46 & 88.00 \\
 &  & 0.30 & \uuline{98.38} & \uuline{98.38} & 98.42 & \uuline{\textcolor{darkgreen}{\bf98.39}} & 97.50 & \textcolor{darkgreen}{\bf98.26} & 98.40 & 98.12 \\
 &  & 0.35 & \uuline{99.86} & \uuline{99.86} & 99.87 & \uuline{\textcolor{darkgreen}{\bf99.87}} & 99.76 & \textcolor{darkgreen}{\bf99.85} & 99.87 & 99.84 \\
 &  & 0.40 & \textcolor{darkgreen}{\bf99.99} & \textcolor{darkgreen}{\bf99.99} & 99.99 & \textcolor{darkgreen}{\bf99.99} & \textcolor{darkgreen}{\bf99.99} & \textcolor{darkgreen}{\bf99.99} & 99.99 & \textcolor{darkgreen}{\bf99.99} \\
\cline{2-11}
 & \multirow[t]{6}{*}{0.30} & 0.30 & 4.96 & 4.93 & 4.99 & 4.97 & 3.62 & 4.94 & \textcolor{red}{\bf 5.15} & 4.80 \\
 &  & 0.40 & \uuline{38.02} & \uuline{\textcolor{darkgreen}{\bf38.07}} & 38.21 & \uuline{38.04} & 33.84 & \textcolor{darkgreen}{\bf37.78} & 39.29 & 37.74 \\
 &  & 0.50 & \uline{89.69} & \uuline{\textcolor{darkgreen}{\bf90.24}} & 90.28 & 89.69 & 87.62 & \textcolor{darkgreen}{\bf89.70} & 90.42 & \textcolor{darkgreen}{\bf89.70} \\
 &  & 0.55 & \uline{97.91} & \uuline{\textcolor{darkgreen}{\bf98.19}} & 98.19 & 97.91 & 97.43 & \textcolor{darkgreen}{\bf98.08} & 98.16 & \textcolor{darkgreen}{\bf98.08} \\
 &  & 0.60 & \uline{99.76} & \uuline{\textcolor{darkgreen}{\bf99.82}} & 99.82 & 99.76 & 99.72 & \textcolor{darkgreen}{\bf99.81} & 99.81 & \textcolor{darkgreen}{\bf99.81} \\
 &  & 0.65 & \uline{\textcolor{darkgreen}{\bf99.99}} & \textcolor{darkgreen}{\bf99.99} & 99.99 & \textcolor{darkgreen}{\bf99.99} & 99.98 & \textcolor{darkgreen}{\bf99.99} & 99.99 & \textcolor{darkgreen}{\bf99.99} \\
\cline{2-11}
 & \multirow[t]{6}{*}{0.50} & 0.50 & 4.08 & 4.99 & 4.99 & 4.08 & 3.67 & 4.97 & 4.97 & 4.97 \\
 &  & 0.60 & \uline{32.67} & \uuline{\textcolor{darkgreen}{\bf35.48}} & 35.51 & 32.67 & 30.59 & \textcolor{darkgreen}{\bf35.17} & 35.31 & \textcolor{darkgreen}{\bf35.17} \\
 &  & 0.70 & \uline{89.61} & \uuline{\textcolor{darkgreen}{\bf90.23}} & 90.27 & 89.61 & 87.62 & \textcolor{darkgreen}{\bf89.70} & 90.42 & \textcolor{darkgreen}{\bf89.70} \\
 &  & 0.75 & \uuline{98.52} & \uuline{\textcolor{darkgreen}{\bf98.57}} & 98.58 & \uuline{98.52} & 98.06 & \textcolor{darkgreen}{\bf98.48} & 98.67 & \textcolor{darkgreen}{\bf98.48} \\
 &  & 0.80 & \uline{\textcolor{darkgreen}{\bf99.92}} & \textcolor{darkgreen}{\bf99.92} & 99.92 & \textcolor{darkgreen}{\bf99.92} & 99.88 & \textcolor{darkgreen}{\bf99.92} & 99.93 & \textcolor{darkgreen}{\bf99.92} \\
 &  & 0.85 & \textcolor{darkgreen}{\bf100.00} & \textcolor{darkgreen}{\bf100.00} & 100.00 & \textcolor{darkgreen}{\bf100.00} & \textcolor{darkgreen}{\bf100.00} & \textcolor{darkgreen}{\bf100.00} & 100.00 & \textcolor{darkgreen}{\bf100.00} \\
\cline{2-11}
 & \multirow[t]{5}{*}{0.70} & 0.70 & 4.95 & 4.93 & 5.00 & 4.97 & 3.62 & 4.94 & \textcolor{red}{\bf 5.15} & 4.80 \\
 &  & 0.80 & \uuline{45.01} & \uuline{44.76} & 45.09 & \uuline{\textcolor{darkgreen}{\bf45.06}} & 39.24 & \textcolor{darkgreen}{\bf44.64} & 45.44 & 43.87 \\
 &  & 0.85 & \uuline{82.25} & \uuline{82.04} & 82.25 & \uuline{\textcolor{darkgreen}{\bf82.26}} & 77.67 & \textcolor{darkgreen}{\bf81.69} & 82.38 & 81.02 \\
 &  & 0.90 & \uuline{\textcolor{darkgreen}{\bf98.39}} & \uuline{98.38} & 98.41 & \uuline{\textcolor{darkgreen}{\bf98.39}} & 97.50 & \textcolor{darkgreen}{\bf98.26} & 98.40 & 98.12 \\
 &  & 0.95 & \uline{\textcolor{darkgreen}{\bf99.99}} & \textcolor{darkgreen}{\bf99.99} & 99.99 & \textcolor{darkgreen}{\bf99.99} & 99.98 & \textcolor{darkgreen}{\bf99.99} & 99.99 & \textcolor{darkgreen}{\bf99.99} \\
\cline{2-11}
 & \multirow[t]{3}{*}{0.90} & 0.90 & 4.90 & 4.21 & 4.24 & 4.90 & 2.93 & 4.61 & \textcolor{red}{\bf 5.43} & 3.92 \\
 &  & 0.95 & \uuline{\textcolor{darkgreen}{\bf31.65}} & 28.63 & 28.70 & \uuline{\textcolor{darkgreen}{\bf31.65}} & 23.49 & \textcolor{darkgreen}{\bf31.02} & 33.80 & 27.26 \\
 &  & 0.99 & \uline{\textcolor{darkgreen}{\bf92.64}} & 90.37 & 90.41 & \uline{\textcolor{darkgreen}{\bf92.64}} & 87.46 & \textcolor{darkgreen}{\bf92.97} & 93.05 & 90.49 \\
\cline{1-11} \cline{2-11}
\bottomrule
\end{tabular}

CX-S: conditional exact tests based on total successes, UX: unconditional exact test, Asymp.: asymptotic test, GB: generalized version of Boschloo's test
\end{table}

\begin{figure}[h!]
	
	\centering
	\includegraphics[width=.85\linewidth]{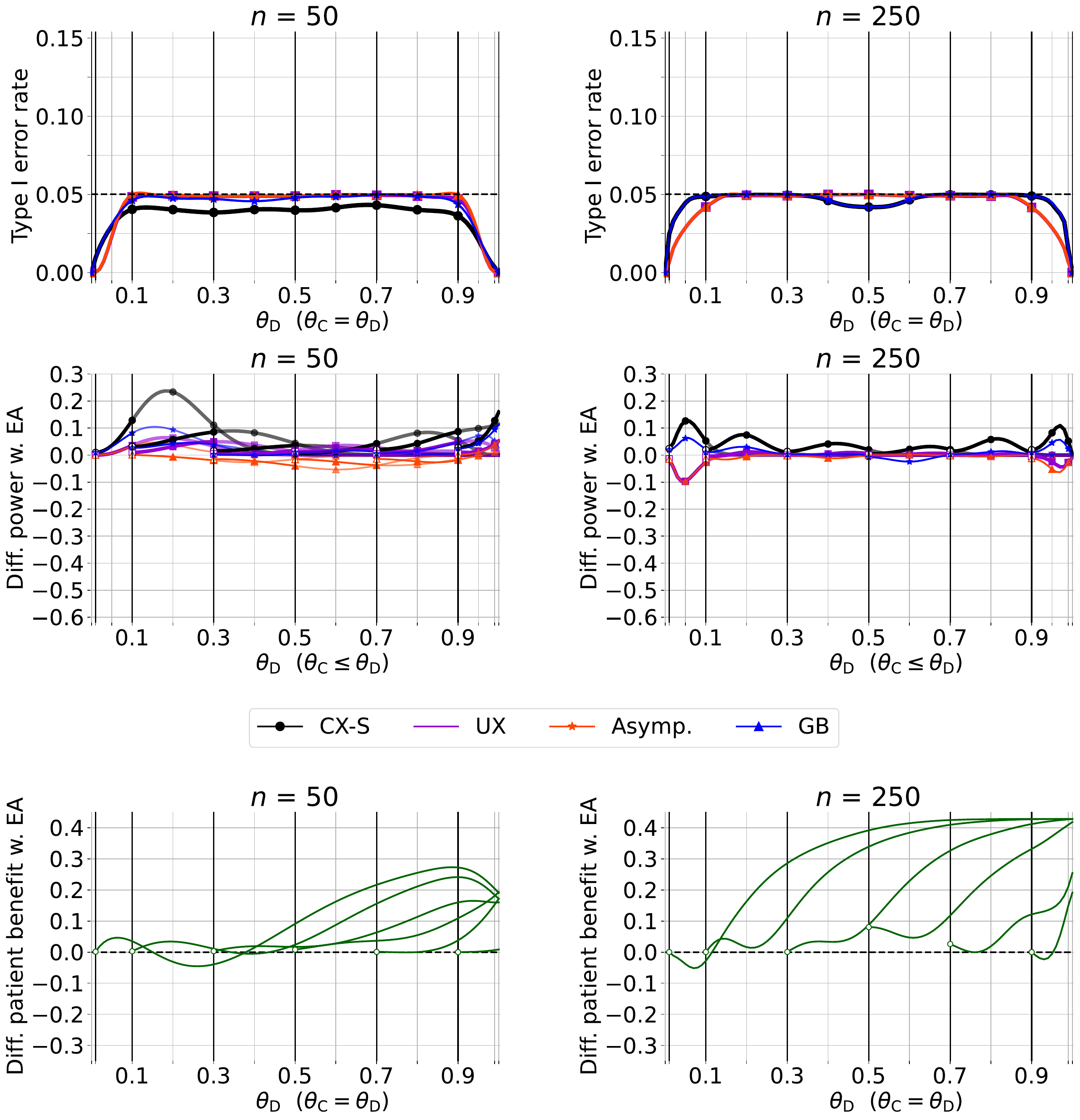}
	
	\caption{Constrained Markov decision process constraining patient benefit and targetting high power; type I error rate profiles~(top), power difference curves compared to the same test under equal allocation~(subtracted, middle), and patient benefit difference curves with respect to equal allocation~(subtracted, bottom) for trial sizes~$n\in\{50,250\}$. We consider the conditional exact, unconditional exact~(UX), asymptotic~(Asymp.), and generalized version of Boschloo's~(GB) Wald tests. For all tests, the significance level was set to~$0.05$.}
	\label{fig:CMDP-EP_plots}
\end{figure}

\begin{table}[h!]
    \centering
 \caption{Rejection rates~$(\%)$ for the constrained Markov decision process procedure protecting expected outcomes~(CMDP-BP, burn-in equal to~6 participants per arm). The bold values indicate highest power or type I error inflation, depending on the parameter configuration considered. The underlined values indicate higher power than under that test for the equal allocation procedure. The twice underlined values indicate higher power than the test with highest power for the equal allocation procedure. Both upper and lower significance levels were set to~$2.5\%$.}\label{tab:RR_CMDP_EP}
\scriptsize \begin{tabular}{lllrrrrrrrr}
\toprule
 &  &  & \multicolumn{4}{c}{CMDP-BP} & \multicolumn{4}{c}{Equal allocation} \\
 &  &  & CX-S & UX & Asymp. & GB & CX-S & UX & Asymp. & GB \\
$\Iend$ & $\theta_\text{C}$ & $\theta_{\text{D}}$ &  &  &  &  &  &  &  &  \\
\midrule
\multirow[t]{32}{*}{50} & \multirow[t]{6}{*}{0.01} & 0.01 & 0.95 & 0.02 & 0.02 & 0.95 & 0.00 & 0.02 & 0.02 & 0.00 \\
 &  & 0.10 & \uline{15.63} & \uuline{\textcolor{darkgreen}{\bf22.14}} & 22.32 & \uline{15.96} & 2.79 & \textcolor{darkgreen}{\bf18.58} & 19.05 & 7.81 \\
 &  & 0.20 & \uline{57.59} & \uuline{\textcolor{darkgreen}{\bf70.52}} & 71.23 & \uline{58.99} & 34.25 & \textcolor{darkgreen}{\bf64.17} & 67.38 & 49.64 \\
 &  & 0.30 & \uline{87.54} & \uuline{\textcolor{darkgreen}{\bf92.65}} & 93.37 & \uline{88.95} & 76.47 & \textcolor{darkgreen}{\bf89.35} & 92.26 & 84.89 \\
 &  & 0.40 & \uline{97.65} & \uuline{\textcolor{darkgreen}{\bf98.61}} & 98.97 & \uuline{98.28} & 95.50 & \textcolor{darkgreen}{\bf97.92} & 98.82 & 97.37 \\
 &  & 0.50 & \uline{99.70} & \uuline{99.80} & 99.90 & \uuline{\textcolor{darkgreen}{\bf99.84}} & 99.53 & \textcolor{darkgreen}{\bf99.78} & 99.89 & 99.75 \\
\cline{2-11}
 & \multirow[t]{6}{*}{0.10} & 0.10 & 4.04 & 4.84 & 4.99 & 4.61 & 0.88 & 3.90 & 4.87 & 1.91 \\
 &  & 0.25 & \uuline{25.39} & \uuline{\textcolor{darkgreen}{\bf29.04}} & 29.40 & \uuline{28.29} & 18.04 & \textcolor{darkgreen}{\bf24.65} & 30.66 & 23.79 \\
 &  & 0.40 & \uline{67.30} & \uuline{\textcolor{darkgreen}{\bf72.73}} & 72.94 & \uuline{70.73} & 59.05 & \textcolor{darkgreen}{\bf69.03} & 75.60 & 68.98 \\
 &  & 0.50 & \uline{87.61} & \uuline{\textcolor{darkgreen}{\bf90.87}} & 90.94 & \uuline{89.69} & 83.26 & \textcolor{darkgreen}{\bf89.55} & 92.49 & 89.54 \\
 &  & 0.60 & \uline{96.99} & \uuline{\textcolor{darkgreen}{\bf98.08}} & 98.09 & 97.73 & 95.95 & \textcolor{darkgreen}{\bf97.88} & 98.52 & \textcolor{darkgreen}{\bf97.88} \\
 &  & 0.70 & \uline{99.62} & \uuline{\textcolor{darkgreen}{\bf99.79}} & 99.79 & 99.74 & 99.53 & \textcolor{darkgreen}{\bf99.76} & 99.85 & \textcolor{darkgreen}{\bf99.76} \\
\cline{2-11}
 & \multirow[t]{6}{*}{0.30} & 0.30 & 3.84 & 4.88 & 4.90 & 4.71 & 2.29 & 4.48 & \textcolor{red}{\bf 6.42} & 4.48 \\
 &  & 0.45 & \uline{16.52} & \uuline{\textcolor{darkgreen}{\bf19.10}} & 19.13 & 18.20 & 13.46 & \textcolor{darkgreen}{\bf18.20} & 22.15 & \textcolor{darkgreen}{\bf18.20} \\
 &  & 0.60 & \uuline{53.52} & \uuline{\textcolor{darkgreen}{\bf56.97}} & 57.02 & \uuline{55.66} & 50.59 & \textcolor{darkgreen}{\bf53.50} & 62.36 & \textcolor{darkgreen}{\bf53.50} \\
 &  & 0.70 & \uuline{79.71} & \uuline{\textcolor{darkgreen}{\bf82.06}} & 82.09 & \uuline{81.30} & 78.22 & \textcolor{darkgreen}{\bf79.30} & 85.96 & \textcolor{darkgreen}{\bf79.30} \\
 &  & 0.80 & \uline{95.37} & \uuline{\textcolor{darkgreen}{\bf96.18}} & 96.19 & \uuline{95.89} & 94.61 & \textcolor{darkgreen}{\bf95.39} & 97.25 & \textcolor{darkgreen}{\bf95.39} \\
 &  & 0.90 & \uline{99.70} & \uuline{\textcolor{darkgreen}{\bf99.79}} & 99.79 & 99.74 & 99.53 & \textcolor{darkgreen}{\bf99.76} & 99.85 & \textcolor{darkgreen}{\bf99.76} \\
\cline{2-11}
 & \multirow[t]{6}{*}{0.50} & 0.50 & 3.99 & 4.89 & 4.91 & 4.79 & 3.28 & 3.73 & \textcolor{red}{\bf 6.51} & 3.73 \\
 &  & 0.65 & \uline{16.39} & \uuline{\textcolor{darkgreen}{\bf18.58}} & 18.62 & \uuline{18.15} & 13.92 & \textcolor{darkgreen}{\bf16.56} & 21.82 & \textcolor{darkgreen}{\bf16.56} \\
 &  & 0.70 & \uline{27.47} & \uuline{\textcolor{darkgreen}{\bf30.21}} & 30.26 & \uuline{29.48} & 23.25 & \textcolor{darkgreen}{\bf28.18} & 33.97 & \textcolor{darkgreen}{\bf28.18} \\
 &  & 0.80 & \uline{59.14} & \uuline{\textcolor{darkgreen}{\bf62.05}} & 62.13 & \uuline{61.04} & 51.04 & \textcolor{darkgreen}{\bf60.64} & 65.60 & \textcolor{darkgreen}{\bf60.64} \\
 &  & 0.90 & \uline{88.70} & \uuline{\textcolor{darkgreen}{\bf90.88}} & 91.02 & \uuline{90.16} & 83.26 & \textcolor{darkgreen}{\bf89.55} & 92.49 & 89.54 \\
 &  & 0.95 & \uline{96.62} & \uuline{\textcolor{darkgreen}{\bf97.79}} & 98.03 & \uuline{97.51} & 95.15 & \textcolor{darkgreen}{\bf97.32} & 98.35 & 97.31 \\
\cline{2-11}
 & \multirow[t]{5}{*}{0.70} & 0.70 & 4.31 & 4.95 & 4.98 & 4.98 & 2.29 & 4.48 & \textcolor{red}{\bf 6.42} & 4.48 \\
 &  & 0.80 & \uline{11.14} & \uuline{\textcolor{darkgreen}{\bf12.80}} & 12.90 & \uuline{12.66} & 6.87 & \textcolor{darkgreen}{\bf11.17} & 15.26 & 11.15 \\
 &  & 0.90 & \uuline{39.15} & \uuline{\textcolor{darkgreen}{\bf43.75}} & 44.40 & \uuline{43.39} & 30.49 & \textcolor{darkgreen}{\bf39.05} & 46.29 & 38.64 \\
 &  & 0.95 & \uuline{64.77} & \uuline{69.19} & 70.88 & \uuline{\textcolor{darkgreen}{\bf69.70}} & 54.92 & \textcolor{darkgreen}{\bf64.19} & 71.05 & 62.60 \\
 &  & 0.99 & \uline{87.44} & \uuline{\textcolor{darkgreen}{\bf92.00}} & 93.30 & \uuline{90.13} & 76.47 & \textcolor{darkgreen}{\bf89.35} & 92.26 & 84.89 \\
\cline{2-11}
 & \multirow[t]{3}{*}{0.90} & 0.90 & 3.63 & 4.69 & 4.97 & 4.36 & 0.88 & 3.90 & 4.87 & 1.91 \\
 &  & 0.95 & \uline{6.47} & \uuline{\textcolor{darkgreen}{\bf8.67}} & 9.11 & \uuline{7.49} & 1.30 & \textcolor{darkgreen}{\bf7.21} & 8.10 & 3.17 \\
 &  & 0.99 & \uline{15.52} & \uuline{\textcolor{darkgreen}{\bf22.57}} & 22.90 & \uline{17.21} & 2.79 & \textcolor{darkgreen}{\bf18.58} & 19.05 & 7.81 \\
\cline{1-11} \cline{2-11}
\multirow[t]{32}{*}{250} & \multirow[t]{6}{*}{0.01} & 0.01 & 2.46 & 0.66 & 0.66 & 2.46 & 0.10 & 2.15 & 2.15 & 0.50 \\
 &  & 0.05 & \uline{\textcolor{darkgreen}{\bf42.45}} & 35.80 & 35.82 & \uline{\textcolor{darkgreen}{\bf42.45}} & 29.83 & \textcolor{darkgreen}{\bf45.48} & 45.49 & 36.25 \\
 &  & 0.10 & \uline{\textcolor{darkgreen}{\bf92.72}} & 90.34 & 90.36 & \uline{\textcolor{darkgreen}{\bf92.72}} & 87.46 & \textcolor{darkgreen}{\bf92.97} & 93.05 & 90.49 \\
 &  & 0.13 & \uline{\textcolor{darkgreen}{\bf98.82}} & 98.28 & 98.28 & \uline{\textcolor{darkgreen}{\bf98.82}} & 97.42 & \textcolor{darkgreen}{\bf98.83} & 98.88 & 98.25 \\
 &  & 0.17 & \uline{\textcolor{darkgreen}{\bf99.93}} & 99.90 & 99.90 & \uline{\textcolor{darkgreen}{\bf99.93}} & 99.82 & \textcolor{darkgreen}{\bf99.93} & 99.94 & 99.88 \\
 &  & 0.20 & \uline{\textcolor{darkgreen}{\bf99.99}} & \textcolor{darkgreen}{\bf99.99} & 99.99 & \textcolor{darkgreen}{\bf99.99} & 99.98 & \textcolor{darkgreen}{\bf99.99} & 100.00 & \textcolor{darkgreen}{\bf99.99} \\
\cline{2-11}
 & \multirow[t]{6}{*}{0.10} & 0.10 & 4.87 & 4.19 & 4.20 & 4.92 & 2.93 & 4.61 & \textcolor{red}{\bf 5.43} & 3.92 \\
 &  & 0.20 & \uuline{61.23} & \uuline{61.27} & 61.38 & \uuline{\textcolor{darkgreen}{\bf61.33}} & 53.82 & \textcolor{darkgreen}{\bf60.13} & 61.93 & 58.40 \\
 &  & 0.25 & \uuline{89.31} & \uuline{\textcolor{darkgreen}{\bf89.42}} & 89.48 & \uuline{89.37} & 85.48 & \textcolor{darkgreen}{\bf88.73} & 89.46 & 88.00 \\
 &  & 0.30 & \uuline{98.38} & \uuline{\textcolor{darkgreen}{\bf98.39}} & 98.41 & \uuline{\textcolor{darkgreen}{\bf98.39}} & 97.50 & \textcolor{darkgreen}{\bf98.26} & 98.40 & 98.12 \\
 &  & 0.35 & \uuline{99.86} & \uuline{99.86} & 99.87 & \uuline{\textcolor{darkgreen}{\bf99.87}} & 99.76 & \textcolor{darkgreen}{\bf99.85} & 99.87 & 99.84 \\
 &  & 0.40 & \textcolor{darkgreen}{\bf99.99} & \textcolor{darkgreen}{\bf99.99} & 99.99 & \textcolor{darkgreen}{\bf99.99} & \textcolor{darkgreen}{\bf99.99} & \textcolor{darkgreen}{\bf99.99} & 99.99 & \textcolor{darkgreen}{\bf99.99} \\
\cline{2-11}
 & \multirow[t]{6}{*}{0.30} & 0.30 & 4.95 & 4.92 & 4.94 & 4.99 & 3.62 & 4.94 & \textcolor{red}{\bf 5.15} & 4.80 \\
 &  & 0.40 & \uuline{37.93} & \uuline{\textcolor{darkgreen}{\bf38.14}} & 38.19 & \uuline{38.06} & 33.84 & \textcolor{darkgreen}{\bf37.78} & 39.29 & 37.74 \\
 &  & 0.50 & \uline{89.58} & \uuline{\textcolor{darkgreen}{\bf90.27}} & 90.28 & 89.65 & 87.62 & \textcolor{darkgreen}{\bf89.70} & 90.42 & \textcolor{darkgreen}{\bf89.70} \\
 &  & 0.55 & \uline{97.89} & \uuline{\textcolor{darkgreen}{\bf98.19}} & 98.19 & 97.90 & 97.43 & \textcolor{darkgreen}{\bf98.08} & 98.16 & \textcolor{darkgreen}{\bf98.08} \\
 &  & 0.60 & \uline{99.76} & \uuline{\textcolor{darkgreen}{\bf99.82}} & 99.82 & 99.76 & 99.72 & \textcolor{darkgreen}{\bf99.81} & 99.81 & \textcolor{darkgreen}{\bf99.81} \\
 &  & 0.65 & \uline{\textcolor{darkgreen}{\bf99.99}} & \textcolor{darkgreen}{\bf99.99} & 99.99 & \textcolor{darkgreen}{\bf99.99} & 99.98 & \textcolor{darkgreen}{\bf99.99} & 99.99 & \textcolor{darkgreen}{\bf99.99} \\
\cline{2-11}
 & \multirow[t]{6}{*}{0.50} & 0.50 & 4.19 & 4.99 & 4.99 & 4.19 & 3.67 & 4.97 & 4.97 & 4.97 \\
 &  & 0.60 & \uline{32.72} & \uuline{\textcolor{darkgreen}{\bf35.45}} & 35.47 & 32.73 & 30.59 & \textcolor{darkgreen}{\bf35.17} & 35.31 & \textcolor{darkgreen}{\bf35.17} \\
 &  & 0.70 & \uline{89.64} & \uuline{\textcolor{darkgreen}{\bf90.17}} & 90.20 & 89.67 & 87.62 & \textcolor{darkgreen}{\bf89.70} & 90.42 & \textcolor{darkgreen}{\bf89.70} \\
 &  & 0.75 & \uuline{98.52} & \uuline{\textcolor{darkgreen}{\bf98.56}} & 98.57 & \uuline{98.53} & 98.06 & \textcolor{darkgreen}{\bf98.48} & 98.67 & \textcolor{darkgreen}{\bf98.48} \\
 &  & 0.80 & \uline{\textcolor{darkgreen}{\bf99.92}} & \textcolor{darkgreen}{\bf99.92} & 99.92 & \textcolor{darkgreen}{\bf99.92} & 99.88 & \textcolor{darkgreen}{\bf99.92} & 99.93 & \textcolor{darkgreen}{\bf99.92} \\
 &  & 0.85 & \textcolor{darkgreen}{\bf100.00} & \textcolor{darkgreen}{\bf100.00} & 100.00 & \textcolor{darkgreen}{\bf100.00} & \textcolor{darkgreen}{\bf100.00} & \textcolor{darkgreen}{\bf100.00} & 100.00 & \textcolor{darkgreen}{\bf100.00} \\
\cline{2-11}
 & \multirow[t]{5}{*}{0.70} & 0.70 & 4.97 & 4.90 & 4.95 & 4.99 & 3.62 & 4.94 & \textcolor{red}{\bf 5.15} & 4.80 \\
 &  & 0.80 & \uuline{44.99} & \uuline{44.70} & 44.97 & \uuline{\textcolor{darkgreen}{\bf45.05}} & 39.24 & \textcolor{darkgreen}{\bf44.64} & 45.44 & 43.87 \\
 &  & 0.85 & \uuline{82.21} & \uuline{82.05} & 82.19 & \uuline{\textcolor{darkgreen}{\bf82.25}} & 77.67 & \textcolor{darkgreen}{\bf81.69} & 82.38 & 81.02 \\
 &  & 0.90 & \uuline{98.39} & \uuline{98.38} & 98.40 & \uuline{\textcolor{darkgreen}{\bf98.40}} & 97.50 & \textcolor{darkgreen}{\bf98.26} & 98.40 & 98.12 \\
 &  & 0.95 & \uline{\textcolor{darkgreen}{\bf99.99}} & \textcolor{darkgreen}{\bf99.99} & 99.99 & \textcolor{darkgreen}{\bf99.99} & 99.98 & \textcolor{darkgreen}{\bf99.99} & 99.99 & \textcolor{darkgreen}{\bf99.99} \\
\cline{2-11}
 & \multirow[t]{3}{*}{0.90} & 0.90 & 4.89 & 4.17 & 4.19 & 4.90 & 2.93 & 4.61 & \textcolor{red}{\bf 5.43} & 3.92 \\
 &  & 0.95 & \uuline{31.70} & 28.51 & 28.56 & \uuline{\textcolor{darkgreen}{\bf31.76}} & 23.49 & \textcolor{darkgreen}{\bf31.02} & 33.80 & 27.26 \\
 &  & 0.99 & \uline{92.63} & 90.30 & 90.33 & \uline{\textcolor{darkgreen}{\bf92.74}} & 87.46 & \textcolor{darkgreen}{\bf92.97} & 93.05 & 90.49 \\
\cline{1-11} \cline{2-11}
\bottomrule
\end{tabular}

CX-S: conditional exact tests based on total successes, UX: unconditional exact test, Asymp.: asymptotic test, GB: generalized version of Boschloo's test
\end{table}
\FloatBarrier
\section{Rejection rates  for randomization test}

\autoref{tab:randomization_test} shows simulation-based unconditional type I error rate and power for the randomization-based Wald test~\citep{SIMON2011767} for several trial sizes and parameter values. 
The rejection rates were obtained by sampling a trial realisation from the specific trial design and parameter values, and then re-randomizing the allocations according to the specific designs to determine the critical value of the Wald statistic given the simulated sequence of outcome variables in the trial. The number of simulations was set to 1,000, and the 
 uncertainty surrounding the estimated rejection rates is indicated by a 95\% confidence interval based on a normal approximation. Highest and lowest power values (determined by overlapping confidence intervals) are indicated by overlined bold~(green) and underlined bold~(red) values, respectively. 

As the randomization-based Wald test depends on the way trial participants are sequentially randomized to treatment, an assumption had to be made about the randomization procedure for equal allocation. For this comparison, we use a permuted block design with blocks of 10 for equal allocation and with blocks of size 12 for the burn-in period. We note here that the permuted block design is not the only choice of randomization procedure for equal allocation, and other randomization methods might lead to a lower degree of predictability~\citep{berger2021roadmap}.

The randomization-based Wald test is an exact test by default; hence, any type I error rate inflation past the 5\% level in \autoref{tab:randomization_test} is due to simulation noise. The randomization-based Wald test shows a relative performance of the designs that looks quite different from the comparison in~\autoref{summary_table}. Mainly, the CMDP-EP and CMDP-P procedures often give the worst performance in terms of power, while curiously, the tempered DBCD Neyman allocation procedure often gives the best performance in terms of power, followed by the DBCD Neyman allocation and equal allocation procedures. An explanation could be that the CMDP procedures have allocation probabilities in the set~$\{0.05, 0.5, 0.95\}$. Hence, while the CMDP procedures are not deterministic, their paths might concentrate probabilistically, leading to low power under the randomization test. 
Instead, DBCD Neyman allocation procedures have a larger range of allocation probabilities, and hence, the aforementioned clustering of allocation paths might not occur for such procedures. The low power for randomization tests under aggressive and close to deterministic RA procedures was already recognized in~\citet{villar2018}. Comparing the maximum power values under the randomization-based Wald test with the power values in \autoref{summary_table}, the randomization-based test gives lower power for~$n=50$ and small or large control success rates, while roughly giving a similar performance for control success rates closer to~$0.5$. For~$n=250$, the randomization-based Wald test often gives comparable maximum power to the maximum values in~\autoref{summary_table}.
\begin{table}[]
    \centering

    \caption{Rejection rate of the randomization-based Wald test under several allocation procedures, parameter values, and trial sizes~$n.$ The margins of error are based on a 95\% confidence interval based on a normal approximation.
    }
    \label{tab:randomization_test}
    \scriptsize\begin{tabular}{lllrrrrrr}
\toprule
 &  &  & NA & TNA & CMDP-P & CMDP-BP & BRAR & EA \\
$\Iend$ & $\theta_\text{C}$ & $\theta_{\text{D}}$ &  &  &  &  &  &  \\
\midrule
\multirow[t]{32}{*}{50} & \multirow[t]{6}{*}{0.01} & 0.01 & 0.8 +/- 0.6 & 1.3 +/- 0.7 & 0.0 +/- 0.0 & 0.5 +/- 0.4 & 1.6 +/- 0.8 & 0.0 +/- 0.0 \\
 &  & 0.1 & 8.8 +/- 1.8 & $\uline{\textcolor{red}{\bf8.0}}$ +/- 1.7 & $\uline{\textcolor{red}{\bf5.6}}$ +/- 1.4 & $\uline{\textcolor{red}{\bf7.1}}$ +/- 1.6 & $\overline{\textcolor{darkgreen}{\bf14.4}}$ +/- 2.2 & $\overline{\textcolor{darkgreen}{\bf11.8}}$ +/- 2.0 \\
 &  & 0.2 & 39.8 +/- 3.0 & 40.1 +/- 3.0 & $\uline{\textcolor{red}{\bf25.2}}$ +/- 2.7 & $\uline{\textcolor{red}{\bf21.2}}$ +/- 2.5 & $\overline{\textcolor{darkgreen}{\bf49.3}}$ +/- 3.1 & $\overline{\textcolor{darkgreen}{\bf51.5}}$ +/- 3.1 \\
 &  & 0.3 & 77.4 +/- 2.6 & 77.5 +/- 2.6 & $\uline{\textcolor{red}{\bf50.3}}$ +/- 3.1 & $\uline{\textcolor{red}{\bf46.0}}$ +/- 3.1 & $\overline{\textcolor{darkgreen}{\bf81.1}}$ +/- 2.4 & $\overline{\textcolor{darkgreen}{\bf83.8}}$ +/- 2.3 \\
 &  & 0.4 & 95.0 +/- 1.4 & $\overline{\textcolor{darkgreen}{\bf95.7}}$ +/- 1.3 & $\uline{\textcolor{red}{\bf71.4}}$ +/- 2.8 & $\uline{\textcolor{red}{\bf66.0}}$ +/- 2.9 & 94.1 +/- 1.5 & $\overline{\textcolor{darkgreen}{\bf97.5}}$ +/- 1.0 \\
 &  & 0.5 & $\overline{\textcolor{darkgreen}{\bf99.5}}$ +/- 0.4 & $\overline{\textcolor{darkgreen}{\bf99.4}}$ +/- 0.5 & $\uline{\textcolor{red}{\bf86.2}}$ +/- 2.1 & $\uline{\textcolor{red}{\bf83.4}}$ +/- 2.3 & 99.0 +/- 0.6 & $\overline{\textcolor{darkgreen}{\bf99.9}}$ +/- 0.2 \\
\cline{2-9}
 & \multirow[t]{6}{*}{0.1} & 0.1 & 3.7 +/- 1.2 & 3.1 +/- 1.1 & 2.4 +/- 0.9 & 2.6 +/- 1.0 & 4.6 +/- 1.3 & 3.2 +/- 1.1 \\
 &  & 0.25 & $\overline{\textcolor{darkgreen}{\bf20.7}}$ +/- 2.5 & $\overline{\textcolor{darkgreen}{\bf22.1}}$ +/- 2.6 & $\uline{\textcolor{red}{\bf14.2}}$ +/- 2.2 & $\uline{\textcolor{red}{\bf12.4}}$ +/- 2.0 & $\overline{\textcolor{darkgreen}{\bf22.0}}$ +/- 2.6 & $\overline{\textcolor{darkgreen}{\bf25.2}}$ +/- 2.7 \\
 &  & 0.4 & $\overline{\textcolor{darkgreen}{\bf64.0}}$ +/- 3.0 & $\overline{\textcolor{darkgreen}{\bf64.2}}$ +/- 3.0 & $\uline{\textcolor{red}{\bf40.7}}$ +/- 3.0 & $\uline{\textcolor{red}{\bf38.5}}$ +/- 3.0 & 60.8 +/- 3.0 & $\overline{\textcolor{darkgreen}{\bf67.5}}$ +/- 2.9 \\
 &  & 0.5 & $\overline{\textcolor{darkgreen}{\bf86.6}}$ +/- 2.1 & $\overline{\textcolor{darkgreen}{\bf87.0}}$ +/- 2.1 & $\uline{\textcolor{red}{\bf66.7}}$ +/- 2.9 & $\uline{\textcolor{red}{\bf62.0}}$ +/- 3.0 & 82.5 +/- 2.4 & $\overline{\textcolor{darkgreen}{\bf87.7}}$ +/- 2.0 \\
 &  & 0.6 & $\overline{\textcolor{darkgreen}{\bf97.3}}$ +/- 1.0 & $\overline{\textcolor{darkgreen}{\bf97.0}}$ +/- 1.1 & 84.0 +/- 2.3 & $\uline{\textcolor{red}{\bf74.3}}$ +/- 2.7 & 94.1 +/- 1.5 & $\overline{\textcolor{darkgreen}{\bf97.3}}$ +/- 1.0 \\
 &  & 0.7 & $\overline{\textcolor{darkgreen}{\bf99.9}}$ +/- 0.2 & $\overline{\textcolor{darkgreen}{\bf99.8}}$ +/- 0.3 & 96.2 +/- 1.2 & $\uline{\textcolor{red}{\bf91.6}}$ +/- 1.7 & 99.0 +/- 0.6 & $\overline{\textcolor{darkgreen}{\bf99.7}}$ +/- 0.3 \\
\cline{2-9}
 & \multirow[t]{6}{*}{0.3} & 0.3 & 4.5 +/- 1.3 & 5.3 +/- 1.4 & 4.0 +/- 1.2 & 3.8 +/- 1.2 & 5.6 +/- 1.4 & 4.7 +/- 1.3 \\
 &  & 0.45 & $\overline{\textcolor{darkgreen}{\bf19.0}}$ +/- 2.4 & $\overline{\textcolor{darkgreen}{\bf17.2}}$ +/- 2.3 & $\uline{\textcolor{red}{\bf10.9}}$ +/- 1.9 & $\uline{\textcolor{red}{\bf9.9}}$ +/- 1.9 & $\uline{\textcolor{red}{\bf13.6}}$ +/- 2.1 & $\overline{\textcolor{darkgreen}{\bf16.5}}$ +/- 2.3 \\
 &  & 0.6 & $\overline{\textcolor{darkgreen}{\bf56.4}}$ +/- 3.1 & $\overline{\textcolor{darkgreen}{\bf53.1}}$ +/- 3.1 & 36.0 +/- 3.0 & $\uline{\textcolor{red}{\bf29.6}}$ +/- 2.8 & 44.8 +/- 3.1 & $\overline{\textcolor{darkgreen}{\bf52.4}}$ +/- 3.1 \\
 &  & 0.7 & $\overline{\textcolor{darkgreen}{\bf79.4}}$ +/- 2.5 & $\overline{\textcolor{darkgreen}{\bf83.3}}$ +/- 2.3 & 62.9 +/- 3.0 & $\uline{\textcolor{red}{\bf53.9}}$ +/- 3.1 & 71.0 +/- 2.8 & $\overline{\textcolor{darkgreen}{\bf81.3}}$ +/- 2.4 \\
 &  & 0.8 & $\overline{\textcolor{darkgreen}{\bf95.6}}$ +/- 1.3 & $\overline{\textcolor{darkgreen}{\bf95.2}}$ +/- 1.3 & 81.4 +/- 2.4 & $\uline{\textcolor{red}{\bf74.7}}$ +/- 2.7 & 90.7 +/- 1.8 & $\overline{\textcolor{darkgreen}{\bf93.7}}$ +/- 1.5 \\
 &  & 0.9 & $\overline{\textcolor{darkgreen}{\bf99.7}}$ +/- 0.3 & $\overline{\textcolor{darkgreen}{\bf99.8}}$ +/- 0.3 & 95.7 +/- 1.3 & $\uline{\textcolor{red}{\bf92.2}}$ +/- 1.7 & 98.2 +/- 0.8 & $\overline{\textcolor{darkgreen}{\bf99.8}}$ +/- 0.3 \\
\cline{2-9}
 & \multirow[t]{6}{*}{0.5} & 0.5 & 4.1 +/- 1.2 & 5.8 +/- 1.5 & 3.5 +/- 1.1 & 2.6 +/- 1.0 & 5.4 +/- 1.4 & 4.4 +/- 1.3 \\
 &  & 0.65 & $\overline{\textcolor{darkgreen}{\bf17.2}}$ +/- 2.3 & $\overline{\textcolor{darkgreen}{\bf17.9}}$ +/- 2.4 & $\uline{\textcolor{red}{\bf12.2}}$ +/- 2.0 & $\uline{\textcolor{red}{\bf9.4}}$ +/- 1.8 & $\overline{\textcolor{darkgreen}{\bf16.5}}$ +/- 2.3 & $\overline{\textcolor{darkgreen}{\bf16.9}}$ +/- 2.3 \\
 &  & 0.7 & $\overline{\textcolor{darkgreen}{\bf29.5}}$ +/- 2.8 & $\overline{\textcolor{darkgreen}{\bf31.0}}$ +/- 2.9 & 18.6 +/- 2.4 & $\uline{\textcolor{red}{\bf13.5}}$ +/- 2.1 & $\overline{\textcolor{darkgreen}{\bf27.0}}$ +/- 2.8 & $\overline{\textcolor{darkgreen}{\bf26.6}}$ +/- 2.7 \\
 &  & 0.8 & $\overline{\textcolor{darkgreen}{\bf58.5}}$ +/- 3.1 & $\overline{\textcolor{darkgreen}{\bf57.6}}$ +/- 3.1 & $\uline{\textcolor{red}{\bf38.8}}$ +/- 3.0 & $\uline{\textcolor{red}{\bf37.5}}$ +/- 3.0 & 50.4 +/- 3.1 & $\overline{\textcolor{darkgreen}{\bf54.5}}$ +/- 3.1 \\
 &  & 0.9 & $\overline{\textcolor{darkgreen}{\bf90.3}}$ +/- 1.8 & $\overline{\textcolor{darkgreen}{\bf89.6}}$ +/- 1.9 & $\uline{\textcolor{red}{\bf65.6}}$ +/- 2.9 & $\uline{\textcolor{red}{\bf62.8}}$ +/- 3.0 & 80.6 +/- 2.5 & $\overline{\textcolor{darkgreen}{\bf87.6}}$ +/- 2.0 \\
 &  & 0.95 & $\overline{\textcolor{darkgreen}{\bf96.4}}$ +/- 1.2 & $\overline{\textcolor{darkgreen}{\bf96.9}}$ +/- 1.1 & $\uline{\textcolor{red}{\bf77.7}}$ +/- 2.6 & $\uline{\textcolor{red}{\bf77.7}}$ +/- 2.6 & 93.4 +/- 1.5 & $\overline{\textcolor{darkgreen}{\bf97.1}}$ +/- 1.0 \\
\cline{2-9}
 & \multirow[t]{5}{*}{0.7} & 0.7 & 5.5 +/- 1.4 & 3.5 +/- 1.1 & 4.2 +/- 1.2 & 3.9 +/- 1.2 & 3.7 +/- 1.2 & 3.9 +/- 1.2 \\
 &  & 0.8 & $\overline{\textcolor{darkgreen}{\bf10.4}}$ +/- 1.9 & $\overline{\textcolor{darkgreen}{\bf10.9}}$ +/- 1.9 & $\overline{\textcolor{darkgreen}{\bf7.7}}$ +/- 1.7 & $\uline{\textcolor{red}{\bf6.2}}$ +/- 1.5 & $\overline{\textcolor{darkgreen}{\bf10.0}}$ +/- 1.9 & $\overline{\textcolor{darkgreen}{\bf9.6}}$ +/- 1.8 \\
 &  & 0.9 & 32.9 +/- 2.9 & $\overline{\textcolor{darkgreen}{\bf39.8}}$ +/- 3.0 & $\uline{\textcolor{red}{\bf20.8}}$ +/- 2.5 & $\uline{\textcolor{red}{\bf21.7}}$ +/- 2.6 & $\overline{\textcolor{darkgreen}{\bf34.8}}$ +/- 3.0 & $\overline{\textcolor{darkgreen}{\bf37.6}}$ +/- 3.0 \\
 &  & 0.95 & 55.7 +/- 3.1 & $\overline{\textcolor{darkgreen}{\bf67.1}}$ +/- 2.9 & $\uline{\textcolor{red}{\bf34.9}}$ +/- 3.0 & $\uline{\textcolor{red}{\bf33.3}}$ +/- 2.9 & 60.9 +/- 3.0 & 57.9 +/- 3.1 \\
 &  & 0.99 & 76.9 +/- 2.6 & $\overline{\textcolor{darkgreen}{\bf91.2}}$ +/- 1.8 & $\uline{\textcolor{red}{\bf51.7}}$ +/- 3.1 & $\uline{\textcolor{red}{\bf47.6}}$ +/- 3.1 & 85.1 +/- 2.2 & 86.4 +/- 2.1 \\
\cline{2-9}
 & \multirow[t]{3}{*}{0.9} & 0.9 & 2.6 +/- 1.0 & 4.7 +/- 1.3 & 2.4 +/- 1.0 & 2.7 +/- 1.0 & 4.1 +/- 1.2 & 3.2 +/- 1.1 \\
 &  & 0.95 & $\uline{\textcolor{red}{\bf4.5}}$ +/- 1.3 & $\overline{\textcolor{darkgreen}{\bf8.7}}$ +/- 1.8 & $\uline{\textcolor{red}{\bf3.6}}$ +/- 1.2 & $\uline{\textcolor{red}{\bf3.8}}$ +/- 1.2 & $\overline{\textcolor{darkgreen}{\bf6.6}}$ +/- 1.5 & $\uline{\textcolor{red}{\bf3.7}}$ +/- 1.2 \\
 &  & 0.99 & 9.8 +/- 1.8 & $\overline{\textcolor{darkgreen}{\bf23.0}}$ +/- 2.6 & $\uline{\textcolor{red}{\bf5.1}}$ +/- 1.4 & $\uline{\textcolor{red}{\bf5.6}}$ +/- 1.4 & $\overline{\textcolor{darkgreen}{\bf22.3}}$ +/- 2.6 & 12.2 +/- 2.0 \\
\cline{1-9} \cline{2-9}
\multirow[t]{32}{*}{250} & \multirow[t]{6}{*}{0.01} & 0.01 & 2.9 +/- 1.0 & 2.9 +/- 1.0 & 2.3 +/- 0.9 & 3.7 +/- 1.2 & 3.6 +/- 1.2 & 1.3 +/- 0.7 \\
 &  & 0.05 & $\overline{\textcolor{darkgreen}{\bf34.1}}$ +/- 2.9 & 32.0 +/- 2.9 & $\uline{\textcolor{red}{\bf21.2}}$ +/- 2.5 & $\uline{\textcolor{red}{\bf22.6}}$ +/- 2.6 & $\overline{\textcolor{darkgreen}{\bf37.5}}$ +/- 3.0 & $\overline{\textcolor{darkgreen}{\bf39.0}}$ +/- 3.0 \\
 &  & 0.1 & $\overline{\textcolor{darkgreen}{\bf88.9}}$ +/- 1.9 & $\overline{\textcolor{darkgreen}{\bf86.3}}$ +/- 2.1 & $\uline{\textcolor{red}{\bf67.2}}$ +/- 2.9 & $\uline{\textcolor{red}{\bf69.7}}$ +/- 2.8 & $\overline{\textcolor{darkgreen}{\bf87.3}}$ +/- 2.1 & $\overline{\textcolor{darkgreen}{\bf89.6}}$ +/- 1.9 \\
 &  & 0.13 & $\overline{\textcolor{darkgreen}{\bf97.7}}$ +/- 0.9 & $\overline{\textcolor{darkgreen}{\bf96.5}}$ +/- 1.1 & $\uline{\textcolor{red}{\bf86.7}}$ +/- 2.1 & $\uline{\textcolor{red}{\bf87.1}}$ +/- 2.1 & $\overline{\textcolor{darkgreen}{\bf96.5}}$ +/- 1.1 & $\overline{\textcolor{darkgreen}{\bf98.2}}$ +/- 0.8 \\
 &  & 0.17 & $\overline{\textcolor{darkgreen}{\bf100.0}}$ +/- 0.0 & $\overline{\textcolor{darkgreen}{\bf99.9}}$ +/- 0.2 & $\uline{\textcolor{red}{\bf97.5}}$ +/- 1.0 & $\uline{\textcolor{red}{\bf96.7}}$ +/- 1.1 & $\overline{\textcolor{darkgreen}{\bf99.7}}$ +/- 0.3 & $\overline{\textcolor{darkgreen}{\bf99.8}}$ +/- 0.3 \\
 &  & 0.2 & $\overline{\textcolor{darkgreen}{\bf100.0}}$ +/- 0.0 & $\overline{\textcolor{darkgreen}{\bf100.0}}$ +/- 0.0 & $\uline{\textcolor{red}{\bf99.0}}$ +/- 0.6 & $\uline{\textcolor{red}{\bf98.5}}$ +/- 0.8 & $\overline{\textcolor{darkgreen}{\bf100.0}}$ +/- 0.0 & $\overline{\textcolor{darkgreen}{\bf100.0}}$ +/- 0.0 \\
\cline{2-9}
 & \multirow[t]{6}{*}{0.1} & 0.1 & 4.9 +/- 1.3 & 5.9 +/- 1.5 & 3.7 +/- 1.2 & 3.9 +/- 1.2 & 4.1 +/- 1.2 & 4.2 +/- 1.2 \\
 &  & 0.2 & $\overline{\textcolor{darkgreen}{\bf57.8}}$ +/- 3.1 & $\overline{\textcolor{darkgreen}{\bf57.3}}$ +/- 3.1 & $\overline{\textcolor{darkgreen}{\bf54.4}}$ +/- 3.1 & $\uline{\textcolor{red}{\bf49.3}}$ +/- 3.1 & $\uline{\textcolor{red}{\bf51.7}}$ +/- 3.1 & $\overline{\textcolor{darkgreen}{\bf58.5}}$ +/- 3.1 \\
 &  & 0.25 & $\overline{\textcolor{darkgreen}{\bf89.4}}$ +/- 1.9 & $\overline{\textcolor{darkgreen}{\bf88.4}}$ +/- 2.0 & $\uline{\textcolor{red}{\bf80.6}}$ +/- 2.5 & $\uline{\textcolor{red}{\bf81.6}}$ +/- 2.4 & $\uline{\textcolor{red}{\bf79.8}}$ +/- 2.5 & 85.0 +/- 2.2 \\
 &  & 0.3 & $\overline{\textcolor{darkgreen}{\bf98.1}}$ +/- 0.8 & $\overline{\textcolor{darkgreen}{\bf97.9}}$ +/- 0.9 & $\uline{\textcolor{red}{\bf95.9}}$ +/- 1.2 & $\uline{\textcolor{red}{\bf95.3}}$ +/- 1.3 & $\uline{\textcolor{red}{\bf94.6}}$ +/- 1.4 & $\overline{\textcolor{darkgreen}{\bf98.3}}$ +/- 0.8 \\
 &  & 0.35 & $\overline{\textcolor{darkgreen}{\bf99.7}}$ +/- 0.3 & $\overline{\textcolor{darkgreen}{\bf99.8}}$ +/- 0.3 & $\overline{\textcolor{darkgreen}{\bf99.4}}$ +/- 0.5 & $\overline{\textcolor{darkgreen}{\bf99.3}}$ +/- 0.5 & $\overline{\textcolor{darkgreen}{\bf99.6}}$ +/- 0.4 & $\overline{\textcolor{darkgreen}{\bf99.9}}$ +/- 0.2 \\
 &  & 0.4 & $\overline{\textcolor{darkgreen}{\bf100.0}}$ +/- 0.0 & $\overline{\textcolor{darkgreen}{\bf100.0}}$ +/- 0.0 & $\overline{\textcolor{darkgreen}{\bf100.0}}$ +/- 0.0 & $\overline{\textcolor{darkgreen}{\bf100.0}}$ +/- 0.0 & $\overline{\textcolor{darkgreen}{\bf99.9}}$ +/- 0.2 & $\overline{\textcolor{darkgreen}{\bf99.9}}$ +/- 0.2 \\
\cline{2-9}
 & \multirow[t]{6}{*}{0.3} & 0.3 & 5.1 +/- 1.4 & 4.3 +/- 1.3 & 4.6 +/- 1.3 & 4.0 +/- 1.2 & 4.2 +/- 1.2 & 4.4 +/- 1.3 \\
 &  & 0.4 & $\overline{\textcolor{darkgreen}{\bf37.2}}$ +/- 3.0 & $\overline{\textcolor{darkgreen}{\bf37.9}}$ +/- 3.0 & $\overline{\textcolor{darkgreen}{\bf32.9}}$ +/- 2.9 & $\overline{\textcolor{darkgreen}{\bf36.9}}$ +/- 3.0 & $\uline{\textcolor{red}{\bf31.7}}$ +/- 2.9 & $\overline{\textcolor{darkgreen}{\bf35.7}}$ +/- 3.0 \\
 &  & 0.5 & $\overline{\textcolor{darkgreen}{\bf90.1}}$ +/- 1.9 & $\overline{\textcolor{darkgreen}{\bf89.6}}$ +/- 1.9 & $\overline{\textcolor{darkgreen}{\bf88.2}}$ +/- 2.0 & $\overline{\textcolor{darkgreen}{\bf86.2}}$ +/- 2.1 & $\uline{\textcolor{red}{\bf82.4}}$ +/- 2.4 & $\overline{\textcolor{darkgreen}{\bf88.6}}$ +/- 2.0 \\
 &  & 0.55 & $\overline{\textcolor{darkgreen}{\bf98.6}}$ +/- 0.7 & $\overline{\textcolor{darkgreen}{\bf98.6}}$ +/- 0.7 & $\overline{\textcolor{darkgreen}{\bf97.0}}$ +/- 1.1 & $\uline{\textcolor{red}{\bf96.6}}$ +/- 1.1 & $\uline{\textcolor{red}{\bf94.1}}$ +/- 1.5 & $\overline{\textcolor{darkgreen}{\bf97.4}}$ +/- 1.0 \\
 &  & 0.6 & $\overline{\textcolor{darkgreen}{\bf99.7}}$ +/- 0.3 & $\uline{\textcolor{red}{\bf99.5}}$ +/- 0.4 & $\overline{\textcolor{darkgreen}{\bf99.7}}$ +/- 0.3 & $\overline{\textcolor{darkgreen}{\bf99.8}}$ +/- 0.3 & $\uline{\textcolor{red}{\bf99.2}}$ +/- 0.6 & $\overline{\textcolor{darkgreen}{\bf100.0}}$ +/- 0.0 \\
 &  & 0.65 & $\overline{\textcolor{darkgreen}{\bf100.0}}$ +/- 0.0 & $\overline{\textcolor{darkgreen}{\bf100.0}}$ +/- 0.0 & $\overline{\textcolor{darkgreen}{\bf100.0}}$ +/- 0.0 & $\overline{\textcolor{darkgreen}{\bf100.0}}$ +/- 0.0 & $\overline{\textcolor{darkgreen}{\bf99.9}}$ +/- 0.2 & $\overline{\textcolor{darkgreen}{\bf100.0}}$ +/- 0.0 \\
\cline{2-9}
 & \multirow[t]{6}{*}{0.5} & 0.5 & 5.0 +/- 1.4 & 5.2 +/- 1.4 & 3.5 +/- 1.1 & 3.5 +/- 1.1 & 4.3 +/- 1.3 & 4.7 +/- 1.3 \\
 &  & 0.6 & $\overline{\textcolor{darkgreen}{\bf36.1}}$ +/- 3.0 & $\overline{\textcolor{darkgreen}{\bf36.1}}$ +/- 3.0 & $\overline{\textcolor{darkgreen}{\bf32.5}}$ +/- 2.9 & $\overline{\textcolor{darkgreen}{\bf31.1}}$ +/- 2.9 & $\uline{\textcolor{red}{\bf29.0}}$ +/- 2.8 & $\overline{\textcolor{darkgreen}{\bf32.2}}$ +/- 2.9 \\
 &  & 0.7 & $\overline{\textcolor{darkgreen}{\bf88.8}}$ +/- 2.0 & $\overline{\textcolor{darkgreen}{\bf90.0}}$ +/- 1.9 & $\overline{\textcolor{darkgreen}{\bf87.7}}$ +/- 2.0 & $\overline{\textcolor{darkgreen}{\bf86.9}}$ +/- 2.1 & $\uline{\textcolor{red}{\bf83.9}}$ +/- 2.3 & $\overline{\textcolor{darkgreen}{\bf87.8}}$ +/- 2.0 \\
 &  & 0.75 & $\overline{\textcolor{darkgreen}{\bf98.3}}$ +/- 0.8 & $\overline{\textcolor{darkgreen}{\bf98.2}}$ +/- 0.8 & $\overline{\textcolor{darkgreen}{\bf98.0}}$ +/- 0.9 & $\overline{\textcolor{darkgreen}{\bf97.0}}$ +/- 1.1 & $\uline{\textcolor{red}{\bf95.9}}$ +/- 1.2 & $\overline{\textcolor{darkgreen}{\bf97.7}}$ +/- 0.9 \\
 &  & 0.8 & $\overline{\textcolor{darkgreen}{\bf100.0}}$ +/- 0.0 & $\overline{\textcolor{darkgreen}{\bf99.9}}$ +/- 0.2 & $\overline{\textcolor{darkgreen}{\bf99.9}}$ +/- 0.2 & $\overline{\textcolor{darkgreen}{\bf99.9}}$ +/- 0.2 & $\uline{\textcolor{red}{\bf99.3}}$ +/- 0.5 & $\overline{\textcolor{darkgreen}{\bf100.0}}$ +/- 0.0 \\
 &  & 0.85 & $\overline{\textcolor{darkgreen}{\bf100.0}}$ +/- 0.0 & $\overline{\textcolor{darkgreen}{\bf100.0}}$ +/- 0.0 & $\overline{\textcolor{darkgreen}{\bf100.0}}$ +/- 0.0 & $\overline{\textcolor{darkgreen}{\bf100.0}}$ +/- 0.0 & $\overline{\textcolor{darkgreen}{\bf99.9}}$ +/- 0.2 & $\overline{\textcolor{darkgreen}{\bf100.0}}$ +/- 0.0 \\
\cline{2-9}
 & \multirow[t]{5}{*}{0.7} & 0.7 & 5.9 +/- 1.5 & 5.0 +/- 1.4 & 4.8 +/- 1.3 & 5.5 +/- 1.4 & 3.9 +/- 1.2 & 4.1 +/- 1.2 \\
 &  & 0.8 & $\overline{\textcolor{darkgreen}{\bf45.1}}$ +/- 3.1 & $\overline{\textcolor{darkgreen}{\bf44.8}}$ +/- 3.1 & $\overline{\textcolor{darkgreen}{\bf40.7}}$ +/- 3.0 & $\uline{\textcolor{red}{\bf37.4}}$ +/- 3.0 & $\uline{\textcolor{red}{\bf38.5}}$ +/- 3.0 & $\overline{\textcolor{darkgreen}{\bf43.1}}$ +/- 3.1 \\
 &  & 0.85 & $\overline{\textcolor{darkgreen}{\bf82.7}}$ +/- 2.3 & $\overline{\textcolor{darkgreen}{\bf81.4}}$ +/- 2.4 & $\overline{\textcolor{darkgreen}{\bf79.8}}$ +/- 2.5 & $\uline{\textcolor{red}{\bf75.6}}$ +/- 2.7 & $\uline{\textcolor{red}{\bf73.6}}$ +/- 2.7 & $\overline{\textcolor{darkgreen}{\bf79.6}}$ +/- 2.5 \\
 &  & 0.9 & $\overline{\textcolor{darkgreen}{\bf98.2}}$ +/- 0.8 & $\overline{\textcolor{darkgreen}{\bf98.4}}$ +/- 0.8 & $\uline{\textcolor{red}{\bf94.6}}$ +/- 1.4 & $\uline{\textcolor{red}{\bf96.1}}$ +/- 1.2 & $\uline{\textcolor{red}{\bf95.5}}$ +/- 1.3 & $\overline{\textcolor{darkgreen}{\bf98.0}}$ +/- 0.9 \\
 &  & 0.95 & $\overline{\textcolor{darkgreen}{\bf99.9}}$ +/- 0.2 & $\overline{\textcolor{darkgreen}{\bf100.0}}$ +/- 0.0 & $\overline{\textcolor{darkgreen}{\bf99.9}}$ +/- 0.2 & $\overline{\textcolor{darkgreen}{\bf99.9}}$ +/- 0.2 & $\overline{\textcolor{darkgreen}{\bf99.8}}$ +/- 0.3 & $\overline{\textcolor{darkgreen}{\bf99.8}}$ +/- 0.3 \\
\cline{2-9}
 & \multirow[t]{3}{*}{0.9} & 0.9 & 5.3 +/- 1.4 & 5.3 +/- 1.4 & 4.6 +/- 1.3 & 4.7 +/- 1.3 & 5.1 +/- 1.4 & 4.2 +/- 1.2 \\
 &  & 0.95 & $\overline{\textcolor{darkgreen}{\bf28.4}}$ +/- 2.8 & $\overline{\textcolor{darkgreen}{\bf33.0}}$ +/- 2.9 & $\uline{\textcolor{red}{\bf19.2}}$ +/- 2.4 & $\uline{\textcolor{red}{\bf21.9}}$ +/- 2.6 & $\overline{\textcolor{darkgreen}{\bf28.5}}$ +/- 2.8 & 25.6 +/- 2.7 \\
 &  & 0.99 & 87.3 +/- 2.1 & $\overline{\textcolor{darkgreen}{\bf92.9}}$ +/- 1.6 & 72.7 +/- 2.8 & $\uline{\textcolor{red}{\bf66.9}}$ +/- 2.9 & 89.0 +/- 1.9 & $\overline{\textcolor{darkgreen}{\bf92.1}}$ +/- 1.7 \\
\cline{1-9} \cline{2-9}
\bottomrule
\end{tabular}

    NA: DBCD Neyman allocation, TNA: tempered DBCD Neyman allocation, CMDP-P: constrained Markov decision process maximizing power, CMDP-BP: constrained Markov decision process maximizing power constraining patient benefit, BRAR: Bayesian response-adaptive randomization, EA: equal allocation.
\end{table}

\end{appendix}

\end{document}